\documentclass[11pt,leqno]{amsart}
\usepackage{amssymb}
\allowdisplaybreaks
\usepackage{mathrsfs}
\usepackage{amsmath, amssymb, mathtools, nicefrac}
\usepackage{csquotes}
\usepackage{abstract}
\usepackage[hidelinks]{hyperref}
\usepackage{marginnote}
\usepackage{xcolor}
\usepackage{graphicx}

\usepackage{etoolbox}
\makeatletter
\patchcmd{\@maketitle}{\newpage}{}{}{} 
\makeatother
\numberwithin{equation}{section}

\theoremstyle{definition}
\newtheorem{definition}{Definition}[section]

\newtheorem{remark}[definition]{Remark}
\theoremstyle{plain}
\newtheorem{theorem}[definition]{Theorem}
\newtheorem{lemma}[definition]{Lemma}
\newtheorem{corollary}[definition]{Corollary}
\newtheorem{prop}[definition]{Proposition}
\newtheorem{assumption}[definition]{Assumption}

\newcommand{\A}{\mathbf{A}}
\newcommand{\B}{\mathbf{B}}
\newcommand{\C}{\mathcal{C}}

\newcommand{\E}{\mathfrak{E}}
\newcommand{\f}{\change{f}}
\newcommand{\g}{\overline{g}}
\newcommand{\G}{\underline{G}}
\renewcommand{\H}{\mathcal{H}}
\newcommand{\I}{\mathbb{I}}
\renewcommand{\j}{\jmath}

\renewcommand{\L}{\mathcal{L}}
\newcommand{\M}{\overline{M}}

\newcommand{\N}{\mathbb{N}}

\newcommand{\nabbar}{\overline{\nabla}}
\renewcommand{\O}[1]{\mathcal{O}\left(#1\right)}

\newcommand{\R}{\mathbb{R}}

\newcommand{\Ric}{\text{\normalfont{Ric}}}
\newcommand{\Riem}{\text{\normalfont{Riem}}}
\newcommand{\supp}{\mathrm{supp}}
\renewcommand{\S}{\mathbb{S}}

\newcommand{\vol}[1]{{\text{\normalfont{dvol}}}_{#1}}
\newcommand{\X}{\bm{X}}
\newcommand{\Z}{\mathbb{Z}}
\renewcommand{\epsilon}{\varepsilon}

\newcommand{\phibar}{{\phi}_{FLRW}}
\newcommand{\del}{\partial}
\newcommand{\Lap}{\Delta}
\renewcommand{\div}{\text{\normalfont{div}}}

\newcommand{\Gamhat}{\hat{\Gamma}}

\newcommand{\nabsak}{\underline{\nabla}}

\newcommand{\Ltilde}{\tilde{\L}}

\newcommand{\numberthis}{\addtocounter{equation}{1}\tag{\theequation}}

\renewcommand{\theequation}{\arabic{section}.\arabic{equation}}

\usepackage{bm}

\newcommand{\change}[1]{#1}

\newcommand{\changereport}[1]{\color{black}#1\color{black}}

\title{Quiescent Big Bang formation in $2+1$ dimensions}
\author[L.~Urban]{Liam Urban}
\address{Faculty of Physics, University of Vienna, Währinger Straße 17, 1090 Vienna, Austria}
\email{liam.urban@ univie.ac.at}

\begin{document}

\maketitle

\begin{abstract}
In this paper, we study the past asymptotics of $(2+1)$-dimensional solutions to the Einstein scalar-field Vlasov system which are close to Friedman-Lema\^\i tre-Robertson-Walker spacetimes on an initial hypersurface diffeomorphic to a closed orientable surface $M$ of arbitrary genus. We prove that such solutions are past causally geodesically incomplete\change{, form a curvature singularity }and exhibit stable Kretschmann scalar blow-up in the contracting direction. In particular, \change{the spacetime is }$C^2$-inextendible towards the past where causal geodesics become incomplete. Moreover, we show that geometry and matter are asymptotically velocity term dominated toward the past, remaining close to their background counterparts. \change{When viewed on the co-mass shell, the Vlasov distribution in particular converges to a limiting distribution on the Big Bang hypersurface, while asymptotics of the spatial metric lead to slight degeneracies when trying to control the components of the Vlasov energy-momentum tensor. These also manifest when viewing the distribution function on the mass shell, where velocities generally approach a smooth one-dimensional subbundle of the tangent bundle, leading the distribution to become highly anisotropic, \changereport{when viewed in terms of the geometry induced by the constant curvature spatial reference metric}. Compared to previous results in higher dimensions, in particular \cite{FU25}}, inhomogeneous terms in the wave and Vlasov equations factor in more strongly in our setting, which \change{a priori creates additional hurdles at high orders while largely keeping the quiescent system intact.}\\
\indent As a corollary, our main result shows that the Strong Cosmic Censorship conjecture holds for certain polarized $U(1)$-symmetric solutions to the Einstein vacuum equations that emanate from a spatial hypersurface diffeomorphic to $M\times\S^1$. 
\end{abstract}

\section{Introduction}

In this paper, we consider $(2+1)$-dimensional solutions $(\M=I\times M,\g)$ to the Einstein scalar-field Vlasov (ESFV) system 
\begin{subequations}\label{eq:EEq}
\begin{align}
\label{eq:EEq1}\Ric[\g]_{\mu\nu}-\frac12 R[\g]\,\g_{\mu\nu}=&\,8\pi\left(T_{\mu\nu}^{SF}+T_{\mu\nu}^{Vl}\right)\,,\\
T_{\mu\nu}^{SF}=&\,\nabbar_{\mu}\phi\,\nabbar_{\nu}\phi-\frac12\,\g_{\mu\nu}\,\nabbar^{\alpha}\phi\,\nabbar_{\alpha}\phi\,,\\
T_{\mu\nu}^{Vl}=&\,\int_{P_{t,x}}p_\mu\,p_{\nu}\,\f(t,x,p)\,\vol{P_{t,x}}\,,\label{eq:EEqMatter-2+1}\\
\label{eq:EEqW}\square_{\g}\phi=&\,0\,,\\
\label{eq:EEqVl-2+1}\change{\mathcal{X}}_{\g}\f=&\,0
\end{align}
\end{subequations}
which are close to Friedman-Lema\^\i tre-Robertson-Walker (FLRW) solutions in the sense of Lemma \ref{lem:FLRW-2+1}. Here, $I\subseteq(0,\infty)$ is an open interval, $M$ is a closed orientable surface of arbitrary \changereport{genus and }the fiber $P_{(t,x)}$ of the \change{co-mass shell }$P$ in \eqref{eq:EEqMatter-2+1} is as in \eqref{eq:mass-shell-2+1}, where we treat Vlasov matter in both the massive ($m=1$) and the massless ($m=0$) case. \change{Furthermore, $\phi$ is a scalar function on $\M$, $\f$ a nonnegative scalar function on $P$ and $\change{\mathcal{X}}_{\g}$ is the geodesic spray with respect to $\g$}. Our main result can be summarized as follows.

\begin{theorem}[Stability of FLRW solutions to the ESFV system in $2+1$ dimensions]\label{thm:main-intro-2+1}
Let $(M,\mathring{g},\mathring{k},\mathring{\pi},\mathring{\psi},\mathring{f})$ be initial data\footnote{see Section \ref{subsec:init-2+1}} to the ESFV system that is close to FLRW initial data such that the momentum support of $\mathring{f}$ is compact. In the massless case, suppose in addition that one has 
\[\inf\{ \lvert p\rvert_{\mathring{g}_x}\,\vert\,(x,p)\in \supp\mathring{f}\}>0\,.\]

Then, the past maximal globally hyperbolic development of said data in the ESFV system is causally geodesically incomplete. In fact, the Kretschmann scalar exhibits stable blow-up \changereport{toward the past boundary of said development}. Furthermore, the solution is asymptotically velocity term dominated\change{\footnote{i.e., the asymptotic behaviour of the solution agrees, at leading order, with a formal analogue of the $(3+1)$-decomposed Einstein equations in which all spatial derivatives are dropped}, the Vlasov distribution remains close to its FLRW counterpart, and so does its momentum support if it is close to that of the FLRW distribution function initially.}
\end{theorem}

The full version of the main theorem can be found in Theorem \ref{thm:main-2+1}. We note that, while we will often restrict ourselves to constant mean curvature (CMC) data throughout, this is not a true restriction and can be assumed without loss of generality, see Section \ref{subsec:lwp-2+1}. \change{Additionally, when viewing the Vlasov particle distribution as a function of the mass shell, degeneracies in the metric lead to concentration of particle velocities in the negative eigendirection of the renormalised shear, as is discussed in Remark \ref{rem:Vlasov-contra-2+1}.}

The BKL conjecture (after Belinskiǐ, Khalatinov and \changereport{Lifshitz}, see \change{\cite{LK63,BKL70,BK73,BKL82}}) predicts that past asymptotics of cosmological spacetimes are dominated by geometric oscillations for typical matter models, while scalar field matter is predicted to have a stabilizing effect. Theorem \ref{thm:main-intro-2+1} confirms that near-FLRW solutions to the ESFV system in \change{$(2+1)$ }dimensions are subject to this stabilization, often referred to as quiescent Big Bang formation. This extends the results of \cite{FU25} to three spacetime dimensions. It also turns out that the more rigid geometric structure in the $(2+1)$-dimensional setting allows for an alternative, simpler proof. In particular, that inhomogeneous terms in both matter equations exhibit weaker decay turns out not to be a significant obstruction to our approach, see the discussion in Section \ref{sec:proof-outline}.\\

For $M=\mathbb{S}^2$, the reference FLRW spacetime exhibits a Big Crunch singularity toward the future at some $t=T$ with the same asymptotic properties as the Big Bang singularity, where the scale factor is preserved by the reflection $t\mapsto T-t$. Applying this throughout, Theorem \ref{thm:main-intro-2+1} also shows that near-FLRW data is future stable and causally incomplete for $M=\mathbb{S}^2$ and exhibits entirely analogous asymptotics toward the Big Crunch as toward the Big Bang.\\

Additionally, as we discuss in Section \ref{subsec:lit}, Theorem \ref{thm:main-2+1} also implies \change{a stability result for solutions to the $(3+1)$-dimensional Einstein vacuum equations in polarized $U(1)$-symmetry, i.e., vacuum spacetimes that contain a non-degenerate spacelike Killing vector field. Such solutions can be written as spacetimes over $I\times M\times \S^1$ equipped with a metric of the form
\begin{gather}
\label{eq:u1-corresp} \check{g}=e^{-2\sqrt{4\pi}\,\phi}\,\g+e^{2\sqrt{4\pi}\,\phi}\,(dx^3)^2
\end{gather}
where $\phi$ is a scalar function on $I\times M$, $\g$ is spacetime metric on $I\times M$ and $(dx^3)^2$ is the standard coordinate metric on $\mathbb{S}^1$. Using this notation, one obtains the following.}

\begin{corollary}[Cosmic Censorship and stable curvature blow-up for polarized $U(1)$-symmetric vacuum solutions]\label{cor:u1-intro} \change{Consider past maximal solutions to the $(3+1)$-dimensional Einstein vacuum equations of the form \eqref{eq:u1-corresp} with
\begin{equation*}
\g=-dt^2+a^2\,\gamma\text{ and }\del_t\phi=C\,a^{-2},\,\del_{x^i}\phi=0
\end{equation*}
where $(M,\gamma)$ is a closed orientable surface of constant curvature and $C<0$}.\footnote{Requiring $C<0$ excludes extremal Kasner-like spacetimes that do not feature a Big Bang singularity. For further discussion of the reference solutions, see Remark \ref{rem:u1-ref}.} Let $(\check{M},\check{g})$ be the past maximal development of nearby polarized $U(1)$-symmetric initial data \change{within the Einstein vacuum equations. Then, $(\check{M}, \check{g})$ is a polarized $U(1)$-symmetric spacetime that }exhibits a crushing singularity with curvature blow-up toward the past boundary and thus is past $C^2$-inextendible. More precisely, the rescaled quantity $e^{-4\sqrt{4\pi}\,\phi}\,\check{\mathcal{K}}$ exhibits stable blow-up, where $\check{\mathcal{K}}$ is the Kretschmann scalar associated with $\check{g}$.
\end{corollary}
The case where $(M,\gamma)\cong(\mathbb{T}^2,\change{(dx^1)^2+(dx^2)^2})$ is also treated in \cite[Theorem 6.3]{RodSpFou20}. However, the Kretschmann scalar asymptotics in Corollary \ref{cor:u1-intro} slightly differ from those in \cite[Theorem 6.3]{RodSpFou20} since we slice $(\check{M},\check{g})$ differently: The CMC \change{foliation }$(M_t)_{t\in(0,t_0]}$ of $(\M,\g)$ underlying Theorem \ref{thm:main-intro-2+1} induces a foliation $(M_t\times\S^1)_{t\in(0,t_0]}$ on $(\check{M},\check{g})$ that leads to Corollary \ref{cor:u1-intro}. This foliation is not CMC and, thus, qualitatively different from the CMC slicing used in \cite{RodSpFou20}, leading to slightly different expressions of the asymptotics. Nevertheless, Corollary \ref{cor:u1-intro} extends the core result on past $C^2$-inextendibility and stable curvature blow-up to the cases where $(M,\gamma)$ is a closed (orientable) constant curvature surface with $\mathrm{genus}(M)=0$ or $\mathrm{genus}(M)\geq 2$. 

\subsection{Previous work}\label{subsec:lit}

\change{$(2+1)$-dimensional gravity often provides tractable toy models compared to its higher dimensional analogues, especially for the purposes of Quantum Gravity.\footnote{See \cite{Witten88} and the recent review \cite{Car25}} In these contexts, one often attempts to resolve classical singularities such as the Big Bang singularity with the help of a given Quantum Gravity theory. Nevertheless, it is still useful even from this persective to study the properties of Big Bang singularities in classical $(2+1)$ gravity to understand where precisely deviations between these theories manifest. However, to the author's knowledge, the only results on past asymptotics of cosmological $(2+1)$ spacetimes are either restricted to vacuum solutions or follow implicitly from symmetry-reduced higher dimensional settings.\\

For $(2+1)$-dimensional cosmological \changereport{\textit{vacuum} }solutions, the geometric structure of the past singularity is well understood, albeit often without providing concrete asymptotics toward the singularity. Commonly, one tackles these with CMC foliations, and in particular with a CMC gauge that reduces the Einstein vacuum equations to a Hamilton-Jacobi system, see \cite{Mon89,Mon08}. In particular, it was shown by Andersson, Moncrief and Tromba in \cite{AMT97} that any $(2+1)$-dimensional vacuum spacetime containing a CMC hypersurface exhibits a crushing singularity in the direction of contraction. An alternative approach was pioneered by Mess in \cite{Mess07}\footnote{see also \cite{Mess07Notes} for further discussion of this work}, by which these spacetimes can be classified using measured geodesic laminations. This classification, which does not make reference to a time function by itself, was related to the cosmological perspective by Benedetti and Guadagnini in \cite{BenGua01}. Therein, the authors construct a cosmological time function that measures the geodesic distance of an event from the initial singularity, and the Big Bang can be viewed as a real tree\footnote{\changereport{A real tree is a metric space $(\mathcal{T},d)$ in which, for any distinct points $p,q\in\mathcal{T}$, there exists a unique curve connecting $p$ and $q$ that is isometric to the arc
\[\gamma: [0,d(p,q)]\longrightarrow(\R,d_{Eucl}),\quad \gamma(t)=t\,.\]}} that is in one-to-one correspondence to a measured geodesic lamination.
That this description of the initial singularity is independent of the choice of time function and thus also applies to CMC foliations was shown in \cite{And05,Bel14,Bel17}. Similarly, Mondal proved in \cite{Mondal20} that, in an appropriate CMC gauge, any vacuum solution curve follows a unique ray in the configuration space of the Hamilton-Jacobi system which approaches a unique point in the space of projective measured laminations toward the initial singularity. To the author's knoweldge, the only extension of such classification approaches to settings containing matter occurs in \cite{BenGua00} which studies certain Einstein solutions coupled to finitely many particles.\\}

\change{Moving on the Einstein scalar-field system and using the notation introduced in \eqref{eq:u1-corresp}, it was shown by Moncrief in \cite{Mon86} that a polarized $U(1)$-symmetric $(3+1)$-dimensional vacuum spacetime is a solution to the Einstein vacuum equations }if and only if $(I\times M,\g,\phi)$ is a solution to the $(2+1)$-dimensional Einstein scalar-field system.\footnote{The factor $\sqrt{4\pi}$ is, in \cite{Mon86}, $\sqrt{0.5}$ since the Einstein equations are normalized differently in that work compared to \eqref{eq:EEq1}.} Due to precisely this correspondence, \change{one is }able to obtain Corollary \ref{cor:u1-intro} from Theorem \ref{thm:main-intro-2+1}.\\
When considering two-dimensional spatial slices with $\mathrm{genus}(M)\geq 2$, the corresponding vacuum Milne solution is nonlinearly stable toward the future within the polarized $U(1)$-symmetry class; see \cite{CBM01}. Toward the past, the first results for the three-dimensional Einstein scalar-field system appear in \cite{IsMon02} and \cite{CBIM04}, in which Fuchsian methods are used to construct families of analytic solutions evolving from the Big Bang singularity which then translate to analytic \change{$(3+1)$}-dimensional $U(1)$-symmetric vacuum spacetimes evolving from the past singularity.\\

\change{In $(3+1)$ dimensions, the original BKL conjecture in \cite{BKL70} predicts that vacuum solutions generally become highly oscillatory as they approach the Big Bang, unless the solution satisfies an integrability condition that holds in polarized $U(1)$-symmetry. It was also predicted therein that this behaviour persists in the presence of most matter sources. }However, the authors of the conjecture themselves predicted that scalar field matter should form an exception to this and exhibit quiescent Big Bang formation in which geometric oscillations are dampened; see \cite{BK73}. This quiescent behaviour has been established rigorously in various recent past stability results for cosmological solutions to the Einstein scalar field system in \change{($d+1$) dimensions for $d\geq 3$}; see, for example, \changereport{\cite{Rodnianski2018, Rodnianski2014, BeyOl21, Speck2018, RodSp22, FU24, RodSpFou20, BOZ25, Zhe26, GPR23, FGR26, Li24}}. In particular, these results on quiescence admit asymptotic control on the solution variables and, consequently, the Kretschmann scalar. This means that the Strong Cosmic Censorship conjecture is verified for these solutions in the following sense: Towards the past, as causal geodesics become incomplete, the spacetime is $C^2$-inextendible since the Kretschmann scalar blows up.\\

Of particular relevance to \change{this }work, Fournodavlos, Rodnianski and Speck have studied the past asymptotics of the so-called generalized Kasner solutions in \cite{RodSpFou20}, i.e., of solutions \change{on }$(0,\infty)\times \mathbb{T}^d$ to the Einstein scalar-field and Einstein vacuum systems of the form
\begin{gather*}
\g_{\textrm{\normalfont Kas},d}=-dt^2+\sum_{i=1}^dt^{2p_i}\,\change{(dx^i)}^2,\quad \sum_{i=1}^dp_i=1,\ \sum_{i=1}^dp_i^2=1-8\pi\,A^2,\\
\del_t\phi_{\textrm{\normalfont Kas},d}=A\,t^{-1},\,\del_{\change{x^i}}\phi=0\,.
\end{gather*}
They showed that non-extremal\footnote{Kasner solutions where none of the Kasner exponents vanish and thus a past singularity forms.} vacuum Kasner spacetimes for $d=3$ are nonlinearly stable within the class of polarized $U(1)$-symmetric initial data toward the past. By the same correspondence as discussed above, this immediately implies nonlinear stability of $(2+1)$-dimensional Einstein scalar-field spacetimes of the form
\begin{align*}
\g_{\textrm{\normalfont Kas},2}=&\,-S^{2p_3}\,dS^2+S^{2(p_1+p_3)}\,\change{(dx^1)}^2+S^{2(p_2+p_3)}\,\change{(dx^2)}^2,\ \del_S\phi=p_3\,S^{-1}, \del_{\change{x^i}}\phi=0\,.
\end{align*}
After applying the coordinate transformation $dt=S^{p_3}\,dS,\ \tilde{x}_i=(1+p_3)^{-\frac12}\,\change{x^i}$, this becomes
\begin{align*}
\g_{\textrm{\normalfont Kas},2}=&\,-dt^2+t^{2q}\,\change{(d\tilde{x}^1)}^2+t^{2(1-q)}\,\change{(d\tilde{x}^2)}^2,\ \del_t\phi=\text{sgn}(p_3)\sqrt{q\,(1-q)}\,t^{-1},\del_{\change{\tilde{x}^i}}\phi=0,
\end{align*}
where $q=(p_1+p_3)/(1+p_3)$.\\

However, it is not obvious that this result extends to $M\times \mathbb{S}^1$ with $\text{genus}(M)\neq 1$: The proof crucially relies on a near-diagonal evolutionary system for the structure coefficients of a Fermi-Walker transported orthonormal frame, in which the only coefficients which may cause instability are automatically zero within the polarized $U(1)$-symmetry class. When perturbing spherical or hyperboloidal spatial metrics, however, these coefficients are no longer perturbed near zero, introducing linear terms of order $t^{-1}$ to the system that cause additional roadblocks to the analysis. Parts of \cite{RodSpFou20} have been extended to non-toroidal spatial geometries in \cite{GPR23}, which provides general conditions for \change{initial data for }the Einstein scalar-field system in spatial dimension $d\geq 3$ to form past singularities. However, \cite{GPR23} does not cover the polarized $U(1)$-symmetric vacuum setting or $d=2$. Corollary \ref{cor:u1-intro} fills in this gap in results on stable Big Bang formation in the case where $\text{genus}(M)\neq 1$. Since the \change{constant time }foliation of $\check{M}$ one obtains from the CMC foliation $(M_t)_{t\in(0,t_0)}$ of $\M$ is no longer CMC itself, the Kretschmann scalar needs to be rescaled to account for this difference in foliation compared to \cite[Theorem 6.3]{RodSpFou20}. Similarly, the induced foliation is also not immediately compatible with the setup in the \changereport{recent preprints \cite{FG24,FGR26}, which consider the singular initial value problem for quiescent Big Bang scenarios and its relation to quiescent asymptotics. The work \cite{FG24} does cover polarized $U(1)$-symmetric initial data on the singularity, but as discussed in \cite[Remark 1.20]{FG24}, it also only considers locally Gaussian developments emanating from the singularity, which our solutions may not admit. While \cite{FGR26} is able to show that quiescent Einstein scalar-field initial data converges to data on the singularity in the sense of \cite{FG24} for $d\geq 3$, it is again not clear how to extend this to the polarized $U(1)$-symmetric setting. Nevertheless, that quiescent asymptotics for the Einstein scalar-field system can be fully described in this manner indicates that such an extension could be viable in principle.}\\

To conclude the discussion of Einstein \change{vacuum and scalar-field }systems, we briefly discuss \change{the vacuum asymptotics in Gowdy symmetry }that effectively reduces the Einstein equations to a $(1+1)$-dimensional problem\change{, see \cite{Gowdy74,Ring10}}. Polarised solutions that are not quotients of Minkowski were shown to be stable in \cite{CIM90}. In \cite{Ring04}, it was also shown that a large class of unpolarized $\mathbb{T}^3$ Gowdy symmetric solutions are past stable, including perturbations of Kasner spacetimes with distinct exponents where the directions of symmetry are the directions in which the Kasner exponents are positive. For perturbations of Kasner spacetimes where one of these exponents is negative, solutions undergo a Kasner bounce if they break polarization, making the Kasner solutions generically unstable, see \change{\cite{Ring09Annals, Li24Gowdy}}. Consequently, one must also expect to encounter this instability for general non-polarized $U(1)$-symmetric perturbations. Indeed, Corollary \ref{cor:u1-intro} covers Kasner(-like) spacetimes with Kasner exponents $(2/3,2/3,-1/3)$ where the direction of symmetry is associated with the negative exponent.\\

Given that scalar fields seem to be very robust in stabilizing past asymptotics, one may wonder whether quiescent Big Bang formation persists when coupling with additional, more physical matter models. This was done for FLRW solutions with spatial geometry $\mathbb{T}^d$ with subcritical Euler fluids in spatial dimension $d\geq 3$ in \change{\cite{BO24EESF}}. Alternatively, since our universe can be modelled as a collisionless self-gravitating gas, Vlasov matter is also a model of cosmological interest; see \cite[Section 1.2 and 4.1]{BT08}.

Future stability of Einstein-Vlasov solutions with zero cosmological constant is well understood for the Minkowski spacetime and the $(3+1)$-dimensional Milne spacetime; see \cite{Tay17, LiTay20, FJS21, BFJSTh21} and \cite{AndFaj20} respectively. For accelerated expansion, future stability results can also be found in \cite{AndRin16} and \cite{Rin13}. The former covers the Einstein-Vlasov system in $\mathbb{T}^3$-Gowdy symmetry, the latter the Einstein Vlasov scalar-field system without symmetry or homgeneity assumptions.

In \change{$(2+1)$ }dimensions, isotropic solutions to the massive Einstein-Vlasov system on $I\times M$ of the form
\[-4\,dt^2+c\,t^2\,\gamma\]
where $\gamma$ is a homogeneous metric on $M$ and $c$ only depends on the genus of $M$ were shown to be nonlinearly stable in \cite{Faj17,Faj18}\change{, }see also \cite{Faj16} for an asymptotic analysis of isotropic solutions. In the negatively curved case, this corresponds to \change{future }stability of the vacuum Milne solutions; see Remark \ref{rem:covering-bases}.\\
 
Toward the past, results had long been restricted to the homogeneous case in \change{$(2+1)$ } gravity (see \cite{Faj17m0}) or to \change{$(3+1)$ }gravity under strong symmetry assumptions (see \changereport{\cite{Rein96, TR06, Tch04} }in \change{plane}, spherical and hyperbolic symmetry, as well as many results on Cosmic Censorship, including \cite{Wea04, DR06, DaRe16, Smu08, Smu11}). In recent joint work \cite{FU25} with David Fajman, we have shown the past stability of FLRW solutions to the Einstein scalar-field Vlasov system in \change{$(3+1)$ }dimensions, in which the scalar field dominates Vlasov matter and thus the spacetime asymptotics essentially fall in line with the Einstein scalar-field system discussed above. We note that the asymptotic equation of state satisfied by Vlasov matter matches precisely the critical threshold excluded in \cite{BO24EESF}. This is reflected in \change{the components of the Vlasov energy-momentum tensor only remaining close to their FLRW counterparts up to a leading order offset, which is caused by the interplay of metric degeneracies and momentum characteristics, both in \cite{FU25} and in Theorem \ref{thm:main-2+1} below. This effect is seen more strongly when viewing the Vlasov distribution on the mass shell, where they lead to anisotropic velocity concentration, see Remark \ref{rem:Vlasov-contra-2+1}}. \changereport{Complementary to \cite{FU25}, it was shown in \cite{AHS25} that generalised Kasner spacetimes are also stable within the Einstein scalar-field Vlasov Maxwell system for $m=0$ if they are \enquote{strongly subcritical}, i.e., the Kasner exponents are in the range of stability studied in \cite{RodSpFou20} and, additionally, the spacetime is not so anisotropic such that non-trivial Vlasov matter would become asymptotically dominant. While the proof only applies to $d\geq 3$ and small, massless Vlasov perturbations, the fact that Vlasov matter remains compatible with quiescence for any reasonable degree of anisotropy suggests that the assumption of isotropy of the background in the present work can be relaxed.\\}

While there are issues in exploiting the conformal structure of \change{$(2+1)$ }gravity toward the past as in \cite{Faj17,Faj18}, working in lower spatial dimensions allows \change{one }to simplify the proof compared to \cite{FU25}. In particular, many of the steps taken to control Vlasov matter can be directly extended from \cite{FU25}. The main contribution of the present paper thus lies in controlling the evolution of geometric variables, as discussed in Section \ref{sec:proof-outline}. 

\subsection{Comments on the proof and outline}\label{sec:proof-outline}
In this section, \change{I }outline the structure of this paper and comment on the main features of the proof of Theorem \ref{thm:main-intro-2+1}, the full version of which is given in Theorem \ref{thm:main-2+1}.\\

This work relies on a \change{$(2+1)$ }decomposition of the ESFV system using constant mean curvature (CMC) gauge with zero shift, often referred to as a constant mean curvature transported coordinate (CMCTC) gauge. While this means we cannot make use of the conformal structure of \change{$(2+1)$ }gravity used in prior work (see Section \ref{subsec:lit}), using CMCTC gauge seems difficult to avoid: Given that the AVTD picture means we have to expect the asymptotic results \eqref{eq:asymp-quiesc} and \eqref{eq:asymp-renorm}, the spatial metric will degenerate toward the Big Bang. However, the shift vector is determined up to a conformal Killing field via an elliptic equation \change{in the underlying gauge }in which metric terms explicitly occur and \change{thus, as in \cite[Remark 1.4]{FU24},  would }inherit any such degeneracy; \change{see also \cite[p.\,137, (3.28)]{Mon86} for the conformal $(2+1)$ gauge}. This would cause linear terms of order $t^{-1}$ or worse to appear in the evolution equations and prevent one from closing energy estimates.\\

In Section \ref{sec:prelim-2+1}, we state the ESFV equations in CMCTC gauge and collect the reference solutions as well as some facts regarding their scale factors. The considerations that need to be taken to ensure local well-posedness and to obtain sufficient continuation criteria are sketched in Section \ref{sec:id-bs}, along with setting up the necessary norms and energies to state the initial data assumptions and set up the bootstrap argument. \change{I }also discuss the analogous background solutions and initial data notions for the polarized $U(1)$-symmetric vacuum setting in Remark \ref{rem:u1-ref} and Remark \ref{rem:u1-init}.\\

Sections \ref{sec:AP} and \ref{sec:en-est-2+1} contain low order a priori bounds and energy estimates for the solution variables\change{\footnote{\change{i.e., the shear, the lapse, the scalar field and the Vlasov distribution, along with the spatial metric and the associated Ricci curvature}}}, respectively. \change{Proofs are kept }brief when they closely follow the arguments in \cite{FU25}. The essential differences in the approach compared to \cite{FU25} occur in Sections \ref{subsec:en-est-k} and \ref{subsec:en-est-metric}. Here, the fact that the Ricci tensor $\Ric[G]$ is pure trace means that any curvature terms in the evolution of the shear $\hat{k}^\sharp$ (see \eqref{eq:REEqk}) can be ignored in the evolution of $\E^{(L)}(\hat{k},\cdot)$. This significantly simplifies the energy estimates for the geometry as a whole, with two exceptions at order zero and at top order to deal with metric error terms that Vlasov matter introduces. 

In order to control Vlasov matter at order zero, one needs to control derivatives of the reference distribution, which introduces terms containing $\Gamma-\Gamhat$. To then not lose derivatives in the zero order energy estimate, we need a scaled low order energy estimate for the first order energy of $\hat{k}$, which we establish in Lemma \ref{lem:en-est-khat-scaled} using the Codazzi constraint. On the other hand, since Vlasov matter is controlled using the Sasaki metric \eqref{eq:sasaki-metric} which contains first order terms in the spatial metric, one incurs high order curvature terms that do not effectively cancel as they do in the geometric evolution. Since these arise from horizontal derivative terms in the Vlasov equation, \change{i.e., terms including lifts of spatial derivatives to the co-mass shell which are scaled favourably, }the scaled top order curvature energy estimate \eqref{eq:en-est-Ric-top-2+1} is sufficient to close the estimates. Since we only have to control the scalar curvature, one can use the momentum constraint to obtain this estimate directly without requiring additional high order estimates for the shear. Altogether, this allows us to establish a combined total energy estimate in Section \ref{subsec:en-est-total} for the total energy \eqref{eq:total-en-def}\change{,  see \eqref{eq:total-en-imp-2+1}.}\\

\change{However, as in \cite{FU25}, \eqref{eq:total-en-imp-2+1} on its own only provides weak control on components of the spacetime metric and on Vlasov matter since all of the respective energies enter into the total energy \eqref{eq:total-en-def} with a time-dependent weight. For the total Vlasov energy, these rates are introduced to mitigate the effect of borderline and linear shear terms in the commuted Vlasov equation, as well as of metric errors from the linearisation near $f_{FLRW}$, which would otherwise prevent us from closing the bootstrap argument. The additional weighted shear and curvature energies are then needed to obtain a closed estimate from the total energy, since Vlasov matter appears linearly in the evolution equations of all other variables. Since the Vlasov terms that enter the geometric evolution do so at a rate that is strictly weaker than $t^{-1}$, this is still sufficient to obtain an appropriate total energy estimate. In turn, the resulting total energy bound controls Vlasov matter sufficiently well to yield a strong bound on the quiescent system, see \eqref{eq:quiesc-en-imp-2+1}.}

To control Vlasov matter, one can use a hierarchy between horizontal and vertical derivatives in the rescaled Vlasov equation \eqref{eq:REEqVlasov} analogous to that in \cite{FU25}. This allows us to introduce scaled Vlasov energies into the total energy in a way that admits a closed total estimate and to obtain the non-scaled Vlasov energy estimates \eqref{eq:vlasov-indiv}, improving the remaining Vlasov bootstrap assumptions.\\

We note that one might expect three potential roadblocks from differences in decay rates compared to \cite{FU25}:
\begin{enumerate}
\item In the background solution, the Vlasov density $\rho_{FLRW}^{Vl}$ behaves as $t^{-\frac32}$ instead of $t^{-\frac43}$ toward the Big Bang, thus potentially having a larger impact on the spacetime asymptotics compared to the scalar field whose matter density behaves like $t^{-2}$ in both cases.
\item In the rescaled wave equation, spatial derivative terms are of order $1$ instead of order $t^\frac13$; see \eqref{eq:REEqWave-2+1} and \cite[(2.20p)]{FU25}.
\item In the rescaled Vlasov equation, the horizontal term diverges at order $t^{-\frac12}$ instead of $t^{-\frac13}$; see \eqref{eq:REEqVlasov} and \cite[(2.19a)]{FU25}.
\end{enumerate}
\change{The first point does lead to weaker convergence rates for the lapse and consequently all other variables in Theorem \ref{thm:main-2+1} due to its effect on the lapse equations; see \eqref{eq:REEqN}. }However, the Vlasov density and pressure are still strictly asymptotically dominated by scalar field matter, which suffices for our purposes. The other two points could only become problematic at top order, since at low orders, it similarly is sufficient that either of these rates are integrable in time to obtain the bounds in Section \ref{sec:AP}. \change{At top order, to deal with issues that do arise from the third point, the scaled energy estimate for the scalar curvature remains sufficient, see Lemma \ref{lem:en-est-R} and \eqref{eq:Ric-en-top-imp}.}\\

Once the energy bounds in Corollary \ref{cor:en-imp} are obtained, the main theorem in Section \ref{sec:main} follows by standard arguments as in, for example, \cite{Rodnianski2014,Speck2018,FU25}. \change{The asymptotic behaviour of the dual Vlasov distribution function $f^\sharp$ on the mass shell is related to that of the spatial metric since the renormalisation of velocities is similar to that used for the renormalised spatial metric $\mathcal{G}$ in Corollary \ref{cor:renorm}, which converges toward the Big Bang hypersurface. The metric $\mathcal{G}$ in Corollary \ref{cor:renorm} also corresponds to the spatial metrics used in the framework of \cite{GPR23}.}

\change{Finally, translating these results to the polarized $U(1)$ symmetric vacuum setting, }Corollary \ref{cor:u1} follows by taking the asymptotic control on the solution variables established in Theorem \ref{thm:main-2+1} and, with these in hand, computing the norm of the second fundamental form as well as the Kretschmann scalar with respect to the induced foliation. The necessary curvature formulas are collected in Section \ref{subsec:curvature-u1}. 

\subsection{Further applications}

To conclude the introduction, \change{I }briefly discuss two potential extensions of this work:

\begin{remark}[Polarized $U(1)$-symmetric Einstein-Vlasov solutions]\label{rem:wild-speculation}
Given Theorem \ref{thm:main-intro-2+1} and Corollary \ref{cor:u1-intro}, one is led to ask whether our approach can be applied to polarized $U(1)$-symmetric solutions to the Einstein-Vlasov system in \change{$(3+1)$ dimensions}, especially given that the scalar field driving quiescence in \change{$(2+1)$ }gravity for the ESFV system dominates Vlasov matter terms. When performing the same decomposition for the Einstein-Vlasov system as in \cite{Mon86}, one clearly obtains a different system from the one in Section \ref{subsec:eq-2+1}, since the Vlasov and wave equations become directly coupled. However, since we essentially deal with the scalar field and Vlasov matter components separately (see Section \ref{sec:proof-outline}), one \change{might }be able to apply similar ideas if Vlasov matter does not determine the spacetime geometry of the reference solution. Specifically, one \change{might }be able to apply similar ideas to spacetimes of the form
\[\left(\M=I\times\S^1\times M,\ \g=-e^{2\mu(s,r)}dr^2+e^{2\lambda(s,r)}dx^2+r^2\gamma, \change{f}=0\right)\,.\]
within polarized $U(1)$-symmetric solutions to the Einstein-Vlasov equations. In particular, after reducing to the corresponding \change{$(2+1)$ }system, the estimates for the geometry in Sections \ref{subsec:en-est-k} and \ref{subsec:en-est-metric} should largely extend. \change{If similar matter estimates to the ones in this work were to still go through, one would likely be able to show that the Kretschmann scalar blow-up exhibited by these specific vacuum Kasner-like solutions is stable within polarized $U(1)$-symmetric solutions to the Einstein equations, extending similar results in surface symmetry in \cite{Rein96} to polarized $U(1)$-symmetry. Furthermore, the resulting quiescent asymptotics would lead $f$ to asymptote to a small limiting distribution on the cotangent bundle of the Big Bang hypersurface. However, the matter estimates in the $(2+1)$ system would need to be fully reworked as they are now directly and fully coupled, while it is essential in this work that we can ensure quiescence before sharply controlling Vlasov matter. This is left as an open problem for future work.}
\end{remark}

\begin{remark}[Theorem \ref{thm:main-intro-2+1} and global stability]\label{rem:covering-bases}
\change{As discussed after Theorem \ref{thm:main-intro-2+1}, for $\text{genus}(M)=0$, time reflection symmetry implies that the results also show stability of the Big Crunch. A straightforward Cauchy stability argument as in \cite[Section 10]{FU24} shows that the past and future stability regimes can be connected to obtain a global stability result. Furthermore, one can likely also obtain global stability results for $\text{genus}(M)\geq 2$:}

Future global stability for the vacuum solutions on $(0,\infty)\times M\times \S^1$, where $M$ is of genus $\mathrm{genus}(M)\geq 2$, was shown in \cite{CBM01} for small polarized $U(1)$-symmetric data.\change{\footnote{See also the future stability of $(3+1)$ Milne spacetimes proven in \cite{AM03}.} }This was done precisely by showing that the \change{$(2+1)$ }Milne solution is future stable within the Einstein scalar-field equations, using the correspondence from \cite{Mon86} discussed above.

When all Vlasov terms are set to zero, this can be combined with Theorem \ref{thm:main-intro-2+1} (and, consequently, Corollary \ref{cor:u1-intro}) to a global stability result for the $(2+1)$-dimensional Einstein scalar-field system (respectively, a global stability result within the polarized $U(1)$-symmetry class of \change{$(3+1)$ }vacuum solutions): The difference in gauge choice between \cite{CBM01} and this work essentially only comes down to whether the shift determines $\del_t\sigma$ or vice versa (see \cite[p.\,1014]{CBM01}), where $\sigma$ is conformal to the evolved spatial metric $g$ with $R[\sigma]=-1$. Thus, one can straightforwardly transform between these gauge choices on a given initial data hypersurface. From there, the regimes of past and future stability can \change{again be connected using Cauchy stability.}

On the other hand, \change{\cite{Faj17} }builds upon the approach to control geometry in \cite{CBM01} with an energy formalism for massive Vlasov matter to show future stability of Milne spacetimes in the $(2+1)$-dimensional Einstein-Vlasov system. Thus, one can likely combine both results with Theorem \ref{thm:main-intro-2+1} to obtain global stability for near-FLRW data to the massive ESFV system with negatively curved spatial geometry, but we omit details to focus on Big Bang asymptotics.
\end{remark}

\section{Preliminaries}\label{sec:prelim-2+1}

\subsection{Notation}\label{subsec:not}

We will denote the Levi-Civita connection of a Semi-Riemannian metric $\mathfrak{g}$ by $\nabla[\mathfrak{g}]$ and its Christoffel symbols by $\Gamma[\mathfrak{g}]$. The associated volume form is denoted by $\vol{\mathfrak{g}}$, the associated trace operator by $\mathrm{tr}_\mathfrak{g}$ and the divergence with respect to $\nabla[\mathfrak{g}]$ by $\div_{\mathfrak{g}}$. In the case where $\mathfrak{g}$ is Riemannian, we denote the induced inner product and pointwise norm by $\langle\,\cdot\,,\,\cdot\,\rangle_{\mathfrak{g}}$ and $\lvert\,\cdot\,\rvert_\mathfrak{g}$ respectively and write $\langle\,\cdot\,\rangle_{\mathfrak{g}}=\sqrt{1+\lvert\,\cdot\,\rvert_{\mathfrak{g}}^2}$. For the Riemann curvature tensor $\Riem[\mathfrak{g}]$ of $\mathfrak{g}$, we use the sign convention
\[\Riem[\mathfrak{g}](X,Y)Z=\left(\nabla[\mathfrak{g}]_X\nabla[\mathfrak{g}]_Y-\nabla[\mathfrak{g}]_Y\nabla[\mathfrak{g}]_X-\nabla[\mathfrak{g}]_{[X,Y]}\right)Z\,.\]
for vector fields $X,Y$ and $Z$. The Ricci and scalar curvature are denoted by $\Ric[\mathfrak{g}]$ and $R[\mathfrak{g}]$ respectively. Recall that, when $\mathfrak{g}$ is a two-dimensional Riemannian metric, one has
\changereport{\[\Riem[\mathfrak{g}]_{ijkl}=\frac{R[\mathfrak{g}]}2\left(\mathfrak{g}_{ik}\,\mathfrak{g}_{jl}-\mathfrak{g}_{il}\,\mathfrak{g}_{jk}\right)\ \text{and}\ \Ric[\mathfrak{g}]_{ij}=\frac{R[\mathfrak{g}]}2\mathfrak{g}_{ij}\,.\]}

For most of this paper, we work with a $(2+1)$-dimensional spacetime $(\M, \g)$ that admits a time function $t \in C^\infty (\M)$ whose level sets $M_s= t^{-1} (\{s\}),\ s\in t(\M)$ are diffeomorphic to a closed orientable surface $M$ and have constant mean curvature for all \change{$s\in t(\M)$}. In particular, $\M$ is diffeomorphic to $I \times M$ where $I \subset \R$ is an open interval.\\

In constant mean curvature transported coordinate (CMCTC) gauge, $\g$ takes the form
\begin{equation}\label{eq:TC}
\g = - n^2 \, d t^2 + g
\end{equation}
where $g$ is a time-dependent Riemannian metric on $M$. \change{We will abbreviate $\nabla\equiv\nabla[g]$ and $\nabbar\equiv\nabla[\g]$. In our convention, the lapse $n$ is positive, i.e..,
\[n=(-\g(\nabbar t,\nabbar t))^{-\frac12}\,.\]}
The future directed timelike unit normal $e_0$ to $M_t$ is given by $e_0=n^{-1}\del_t$.\\

We will often work with the time-dependent rescaled Riemannian metric $G = a(t)^{-2} \, g$ on $M_t$ where the scale factor $a\in C^\infty(I)$ is introduced in Lemma \ref{lem:FLRW-2+1}. Note that, for a time-dependent vector field $X$ and a time-dependent tensor field $\mathfrak{T}$ on $M$, one has $\nabla[G]_X\mathfrak{T}=\nabla_X\mathfrak{T}$. We use $\Lap_G$ to denote the Laplace-Beltrami operator with respect to $G$. \\
\indent The second fundamental form is a symmetric $(0,2)$-tensor on $M_t$ given by
\[k(X,Y)=-\g(\nabbar_Xe_0,Y)\,.\]
\indent We will assume that the mean curvature $\tau:=\mathrm{tr}_gk$ of $M_t$ is given by 
\begin{equation}\label{eq:CMC-2+1}
\tau=-2\,\frac{\dot{a}}{a}\,.
\end{equation}

In the following, \change{Latin }indices $i,j,\dots$ refer to indices with respect to a given local coordinate system $(x^i)_{i=1,2}$ on $M$ or to the momentum coordinates \change{$(p_i)_{i=1,2}$ on the co-mass shell }induced by such coordinates; see Section \ref{subsec:resc}.

Lowercase Greek letters are spacetime indices with respect to $\M$, running from $0$ to $2$. The index $0$ always denotes tensor components relative to $e_0$. We use Einstein summation convention if not stated otherwise.\\

We raise and lower indices with respect to $\g$ for tensors on $\M$ and $g$ for those on $M$ without comment. When we raise indices of a tensor $\mathfrak{T}$ on $M$ with respect to $G$, the result is denoted by $\mathfrak{T}^\sharp$. To simplify notation, we will occasionally write $\mathfrak{T}_1\ast\mathfrak{T}_2$ for an arbitrary contraction of time-dependent tensor fields on $M$ with respect to $G$ up to multiplicative constant and $\nabla^I\mathfrak{T}$ for $I\in\{0,1,2,\dots\}$ when $I$ covariant derivatives of a time-dependent tensor $\mathfrak{T}$ are taken. \\

We will also work with polarized $U(1)$-symmetric $(3+1)$-dimensional spacetimes $(\check{M}\cong\overline{M}\times \S^1,\check{g})$. \changereport{In such settings, $\del_{\change{x^3}}$ }refers to the extended coordinate derivative on $\mathbb{S}^1$ and the index $3$ to components with respect to $\del_{\change{x^3}}$. \\

Finally, for real functions $\zeta_1,\zeta_2$ with $\zeta_2\geq 0$, we write $\zeta_1\lesssim\zeta_2$ if there exists a constant $K>0$ such that $\max\{\zeta_1,0\}\leq K\zeta_2$. If $\zeta_1$ is also nonnegative, we write $\zeta_1\simeq \zeta_2$ if $\zeta_1\lesssim \zeta_2$ and $\zeta_2\lesssim \zeta_1$ holds.

\subsection{The reference solutions}\label{subsec:FLRW}

Our main result deals with perturbations of FLRW backgrounds with isotropic scalar field and Vlasov matter. Using standard formulas for warped product spacetimes as in, for example, \cite[p.\,345]{ONeill83}, one shows that such backgrounds take the following form in the ESFV system \eqref{eq:EEq}.

\begin{lemma}[FLRW solutions to the ESFV system]\label{lem:FLRW-2+1}
Let $(\M,\g_{FLRW})$ be an FLRW spacetime of the form
\[\M=(0,T)\times M,\quad \g_{FLRW}=-dt^2+a(t)^2\,\gamma\]
where $T\in (0,\infty]$, $a\in C^\infty(0,T)$ is a positive function and $(M,\gamma)$ is a closed orientable surface with constant curvature $\kappa\in\R$. Further, let $\phibar$ and $\f_{FLRW}$ be spatially homogeneous and isotropic solutions to \eqref{eq:EEqW} and \eqref{eq:EEqVl-2+1} for $\g=\g_{FLRW}$, i.e., functions such that
\begin{subequations}
\begin{gather}
\label{eq:FLRW-wave-2+1}
\del_t\phibar=C\,a(t)^{-2},\quad \nabla\phibar=0, \\
\label{eq:FLRW-Vlasov-2+1}
\change{{\f}_{FLRW}(t,x,p_i)={\mathcal{F}}\left((\gamma^{-1})^{ij}\,p_i\,p_j\right)}
\end{gather}
for some $C\in\R$ and ${\mathcal{F}}\in C_c^\infty([0,\infty),[0,\infty))$.
We define
\begin{align}
{\rho}_{FLRW}^{Vl}(t):=&\ 2\pi\int_0^\infty R\sqrt{m^2\,a(t)^2+R^2}\,\mathcal{F}(R^2)\,dR\,,\\
{\mathfrak{p}}_{FLRW}^{Vl}(t):=&\,\pi\int_0^\infty \frac{R^3}{\sqrt{m^2\,a(t)^2+R^2}}\,\mathcal{F}(R^2)\,dR\,,
\end{align}
where $m=0$ in the massless case and $m=1$ in the massive case. 

Then, $(\M,\g_{FLRW},\phi_{FLRW},\f_{FLRW})$ is a solution to the ESFV system \eqref{eq:EEq} if and only if the Friedman equation is satisfied:
\begin{equation}\label{eq:Friedman-2+1}
\dot{a}^2=4\pi\,C^2\,a^{-2}+4\pi\,a^{-1}\left(\rho^{Vl}_{FLRW}+2\,\mathfrak{p}^{Vl}_{FLRW}\right)-{\kappa}\,.
\end{equation}
\end{subequations}
\end{lemma}
Note that $\rho^{Vl}_{FLRW}$ and $\mathfrak{p}^{Vl}_{FLRW}$ are rescalings of the non-zero components of $T^{Vl}[\g_{FLRW},\f_{FLRW}]$, i.e.,
\begin{equation*}
\rho^{Vl}_{FLRW}=a^3\,T_{00}^{Vl}[\g_{FLRW},\f_{FLRW}]\,\ \text{ and }\ 
\mathfrak{p}^{Vl}_{FLRW}=\frac12\,a^3\,\mathrm{tr}_{a^2\gamma}\left(T^{Vl}[\g_{FLRW},\f_{FLRW}]\right)\,.
\end{equation*}
Moreover, observe that $(\M,\g)$ is of the form \eqref{eq:TC} and that the mean curvature $\tau$ of the constant time surfaces $M_t$ satisfies \eqref{eq:CMC-2+1}. In particular, \eqref{eq:Friedman-2+1} implies the following identities for the mean curvature $\tau$.
\begin{subequations}
\begin{align}
\label{eq:tau2}\tau^2=&\ 16\pi\,C^2\,a^{-4}+16\pi\,a^{-3}\left(\rho^{Vl}_{FLRW}+2\,\mathfrak{p}^{Vl}_{FLRW}\right)-4\,\kappa\,a^{-2}\,,\\
\label{eq:delt-tau-2+1}\del_t\tau=&\ 16\pi\,C^2\,a^{-4}+8\pi\,a^{-3}\left(\rho^{Vl}_{FLRW}+4\,\mathfrak{p}^{Vl}_{FLRW}\right)-2\,\kappa\,a^{-2}\,.
\end{align}
\end{subequations}

Throughout, we assume $a(0)=0$, that $a$ is positive somewhere and that the scalar field is non-trivial, i.e., that $C\neq 0$. Under these assumptions, one obtains the following result as in \cite[Lemma 2.2]{FU25} and \cite[Lemma 3.7]{Speck2018}.

\begin{lemma}[Properties of the scale factor]\label{lem:scale-factor-2+1} The singular initial value problem \eqref{eq:Friedman-2+1} with $a(0)=0$ and $C\neq 0$ has a unique nonnegative solution on the maximal interval of existence $(0,T)$ for some $T\in(0,\infty]$. For $\kappa\leq 0$, one has $T=\infty$. Else, one has $T<\infty$, $\left(-\frac{T}2,\frac{T}2\right)\ni t\mapsto a\left(t+\frac{T}2\right)$ is even and the scale factor can be continuously extended to $a(T)=0$.\\

As $t\downarrow 0$, one has $a(t)\simeq \sqrt{t}$. Further, given $t_0<\frac{T}2$, $t\in(0,t_0)$ and $q\in(0,\infty)$, there exists a constant $K$ independent of $t,t_0$ and $q$ such that the following estimates hold.
\begin{subequations}
\begin{align}
\label{eq:Friedman-ineq}\sqrt{4\pi}\,\lvert C\rvert\,a(t)^{-1}\leq&\ \dot{a}(t)+\max\{0,\kappa\}\,,\\
\int_t^{t_0}a(s)^{-2-q}\,ds\leq&\,\frac{K}q\,a(t)^{-q}\,.\label{eq:a-int-est}
\end{align}
\end{subequations}
\end{lemma}
We recall that, for $\kappa>0$, the time reflection symmetry of $a$ implies that the entirety of Theorem \ref{thm:main-2+1} also applies toward the future for $t_0^\prime\in(\frac{T}2,T)$: Similar to how one argues in Section \ref{subsec:lwp-2+1}, one can use Cauchy stability in harmonic gauge to extend past $t=\frac{T}2$ while controlling the perturbation size and can then find a nearby CMC hypersurface to restart the stability analysis.

\begin{remark}[Reference solutions for Corollary \ref{cor:u1-intro}]\label{rem:u1-ref}
For $f_{FLRW}=0$, the FLRW solutions above correspond to the following polarized $U(1)$-symmetric $(3+1)$-dimensional vacuum solutions for some $t_1>0$.
\begin{equation}\label{eq:u1-FLRW}
\check{g}=\exp\left(2\sqrt{4\pi}\,C\int_t^{t_1}a(s)^{-2}\,ds\right)\left(-dt^2+a(t)^2\,\gamma\right)+\exp\left(-2\sqrt{4\pi}\,C\int_t^{t_1}a(s)^{-2}\,ds\right)\change{(dx^3)}^2\,.
\end{equation}
Using the curvature formulas in Section \change{\ref{sec:app} }for $\nabla\phi=0$, $\hat{k}=0$ and $n=1$ as well as the Friedman equation \eqref{eq:Friedman-2+1}, one computes that the Kretschmann scalar 
\[\check{\mathcal{K}}=\Riem[\check{g}]_{\alpha\beta\gamma\delta}\,\Riem[\check{g}]^{\alpha\beta\gamma\delta}\]
satisfies
\begin{align*}
\exp\left(4\sqrt{4\pi}C\int_t^{t_1}a(s)^{-2}\right)\check{\mathcal{K}}=&\,128\pi \,C^2\,a(t)^{-4}\left(\frac{\tau(t)}2+\sqrt{4\pi}\,C\,a(t)^{-2}\right)\\
=&\,128\pi\,C^2\,a(t)^{-8}\left(\sqrt{4\pi}\,C-\sqrt{4\pi\,C^2-\kappa\,a(t)^2}\right)^2\,.\numberthis\label{eq:Kretsch-ref}
\end{align*}

To better understand these solutions, we transform the spacetime metric $\check{g}$ into the more familiar Kasner-like form, depending on the signs of $\kappa$ and $C$. \\

For $\kappa=0$, the Friedman equation \eqref{eq:Friedman-2+1} with $a(0)=0$ is solved by $a(t)=2\,\pi^\frac14\,\lvert C\rvert^\frac12\,t^\frac12$\,.
In particular, there exists $t_1\in(0,\infty)$ such that $a(t_1)=1$. \\
\indent Assuming in addition that $C<0$, we use \eqref{eq:Friedman-2+1} to introduce the new time coordinate $S$.
\[\dot{S}(t)=\exp\left(\sqrt{4\pi}\,C\int_t^{t_1}a(s)^{-2}\,ds\right)=a(t),\ S(0)=0\,.\]
It follows that $S(t)$ is proportional to $t^\frac32$ and, consequently, that $\check{g}$ corresponds to the Kasner solution
\[\check{g}=-dS^2+S^\frac43\,\left(\change{(d{x}^1)}^2+\change{(d{x}^2)}^2\right)+S^{-\frac23}\,\change{(d{x}^3)}^2\,,\]
possibly upon rescaling the spatial coordinates by constants. Furthermore, one then has
\[a(t)^4\,\check{\mathcal{K}}=\exp\left(4\sqrt{4\pi}\,C\int_t^{t_1}a(s)^{-2}\,ds\right)\check{\mathcal{K}}=2048\,\pi^2\,C^4\, a(t)^{-8}\]
and thus that $\check{\mathcal{K}}$ is proportional to $S^{-4}$.

For $C>0$ and $\kappa=0$, one finds that the solution can be transformed to the following extremal Kasner solution: 
\begin{align*}
\check{g}=&\,-dS^2+(d\tilde{x}^1)^2+(d\tilde{x}^2)^2+S^2(d\tilde{x}^3)^2\,.
\end{align*}
Here, $S$ is proportional to $a$ and $\tilde{x}^i$ are time-independent rescalings of the original coordinates. One also sees from \eqref{eq:Kretsch-ref} that the Kretschmann scalar vanishes. Indeed, this simply is the Minkowski metric expressed in an unusual coordinate system.\\

For $\kappa\neq 0$, one obtains that the spacetime metric can be written as
\begin{equation}\label{eq:u1-ref-metric}
\check{g}=-dS^2+b(S)^2\,\gamma+b_3(S)^2\,\change{(dx^3)}^2\,.
\end{equation} 
where the scale factors asymptotically correspond to their spatially flat counterparts at leading order, as we compute in Section \ref{subsec:ref-u1}.

For $C<0$, one has $b(S)\simeq S^\frac23$ and $b_3(S)\simeq S^{-\frac13}$ as $S$ approaches $0$ and the leading order terms in $e^{4\sqrt{4\pi}\,\phi_{FLRW}}\check{\mathcal{K}}$ are also analogous to the spatially flat case.

For $C>0$, we obtain $b(S)\simeq 1$ and $b_3(S)\simeq S$ as $S$ approaches $0$. However, unlike in the spatially flat setting, $(\check{M},\check{g})$ exhibits a curvature singularity of order 
$S^{-2}$. Since the metric still asymptotes towards an extremal Kasner-like metric, one cannot expect such singularities to be stable. In fact, they appear to be unstable within polarized $U(1)$-symmetric solutions; see Remark \ref{rem:extreme}. This is why we do not consider $C>0$ (i.e., $\phi>0$) in Corollary \ref{cor:u1-intro}.
\end{remark}

\subsection{Rescaled variables and Sasaki metrics}\label{subsec:resc}

Before declaring our normalised solution variables, we have to introduce canonical local coordinates on the \change{co-mass shell}

\change{\begin{equation}\label{eq:mass-shell-2+1}
P=\bigsqcup_{t\in I}P_{t}\,,\quad P_{t}=\left\{(t,x,p)\in T^\ast\M\ \vert\ x\in M,\ p\in T^\ast_{(t,x)}\M,\ p_\mu \,p^\mu=-m^2,\ p^0>0\right\}\,.
\end{equation}}
Let $(U,(x^i)_{i=1,2})$ be a coordinate neighbourhood of $M$ and
\[V:=\left\{(t,x,p)\in P\ \vert\ x\in U\right\}\,.\] 
Then, for $w=(s,y,q)\in V$, we define $y^0(w)=s$ and $y^i(w)=x^i(y)$ for $i=1,2$. Furthermore, we introduce the canonical momentum coordinates \change{$(p_\mu)_{\mu=0,1,2}$ by
\[q=\left.p_0(w)\,\theta^0\right\vert_y+\left.p_i(w)\,d{x^i}\right\vert_{y}\,,\]
where $\theta^0$ is the dual coframe to the future directed unit normal $e_0$. 
Note that the co-mass shell relation and \eqref{eq:TC} imply
\[p^0=-p_0=\sqrt{m^2+\lvert p\rvert_g^2}\,.\]}
Then, \change{$(y^0,y^1,y^2,p_1,p_2)$ }is a local coordinate system on $V$ canonically derived from coordinates $(x^1,x^2)$ on $U$ \change{such that $(y^i,p_i)$ are time-independent}. In all that follows, we will use such canonically induced coordinates on the \change{co-mass shell }without further comment. We will write $t$ instead of $y^0(t,x,p)$ and $x^i$ instead of $y^i(t,x,p)$. \\

Now, we can introduce the following rescaled solution variables.

\begin{definition}[Expansion normalised variables]\label{def:resc} Consider a solution $(\M,\g,\phi,\overline{f})$ of the ESFV system \eqref{eq:EEq} in CMCTC gauge, i.e., where $\g$ is of the form \eqref{eq:TC} and $\M$ is foliated by constant time slices $M_t$ with constant mean curvature $\tau$. We define the following normalised variables as
\begin{subequations}
\begin{gather}
\label{eq:resc-G}G_{ij}=a^{-2}g_{ij},\ 
\ N=n-1,\ \Psi=a^2e_0\phi-C\,.
\end{gather}
We decompose the second fundamental form $k_{ij}$ as
\begin{equation}\label{eq:khat}
k_{ij}=\frac{\tau}2\,g_{ij}+\hat{k}_{ij}
\end{equation}
and refer to the tracefree spatial tensor $\hat{k}$ as the shear. 
Additionally, we define the rescaled momentum variables
\change{\begin{gather}
\label{eq:resc-Vlasov-mom}v_i=\,p_i\quad \text{and} \quad v^0=a\,p^0=\sqrt{m^2\,a^2+\lvert v\rvert_G^2}\,.
\end{gather}}
The rescaled Vlasov energy-momentum components are defined as follows:
\begin{gather*}
\rho^{Vl}=a^3\,T_{00}^{Vl},\ \changereport{\j^{Vl}_i=-a^2\,T_{0i}^{Vl}=a^2{(T^{Vl})^{0}}_i},\ \mathfrak{p}^{Vl}=\frac12\, a^3\,\text{tr}_gT^{Vl},\ S_{ij}^{Vl}=a^3\,T^{Vl}_{ij}\\
S_{ij}^{Vl,\parallel}=S_{ij}^{Vl}-\mathfrak{p}^{Vl}G_{ij}\,.
\end{gather*}
\end{subequations}
Note that $S^{Vl,\parallel}$ vanishes if the Vlasov distribution is isotropic.
\end{definition}

The rescaled metric $G$ on $M$ induces the Sasaki metric $\G$ on \change{$T^\ast M$ }as follows: For the local vector field frame \change{$\{\A_1,\A_2,\B^1,\B^2\}$ }given by
\begin{subequations}\label{eq:AB-2+1}
\change{\begin{align}
\A_i=&\,\del_{x^i}+v_l\,\Gamma[G]^l_{ij}\,\del_{v_j}\ \text{and}\\
\B^i=&\,\del_{v_i}\,,
\end{align}}
\end{subequations}
we set
\change{\begin{equation}\label{eq:sasaki-metric}
\G(\A_i,\A_j)=G_{ij},\ \G(\A_i,\B^j)=0,\ \G(\B^i,\B^j)=(G^{-1})^{ij}\ \text{for}\ i,j=1,2\,.
\end{equation}}
\noindent The Levi-Civita connection with respect to $\G$ is denoted by $\nabsak\equiv\nabla[\G]$. \change{Such metrics were first studied on the tangent bundle by Sasaki in \cite{Sak58}. We refer the recent review article \cite{A-CGaSar22} for a more detailed overview of this equivalent formulation on the cotangent bundle and co-mass shell and its application to relativistic kinetic theory.}

Note that the horizontal derivatives $\{\A_1,\A_2\}$ are orthogonal to the vertical derivatives \change{$\{\B^1,\B^2\}$ }with respect to $\G$. The latter spans an integrable distribution on \change{$T^\ast M$ }whose leaves are the tangent spaces \changereport{$T^\ast_xM,\ x\in M$}. Thus, we can distinguish between covariant derivatives taken in vertical directions, denoted by $\nabsak_{\mathrm{vert}}$ and those taken in horizontal directions, denoted by $\nabsak_{\mathrm{hor}}$, without making reference to local coordinates. \\
Similarly, this induces the Riemannian metric $\G\vert_{\mathrm{vert}}$ on the vertical distribution by
\change{\[\G\vert_{\mathrm{vert}}(\B^i,\B^j)=(G^{-1})^{ij}\,.\]}
Viewing $(G_{ij})$ in transported local coordinates as a matrix and denoting its determinant by $\det G$, note that one has
\change{\begin{equation}\label{eq:vol-form-Sasaki}
\vol{(\G\vert_{\mathrm{vert}})_x}=\sqrt{\det G^{-1}_x}\,d{v}^1\wedge dv^2=\vol{G^{-1}_{x}}\,.
\end{equation}}
In the following, we will write \changereport{$\nabsak_i=\nabsak_{\A_i}$ and \change{$\nabsak^{i+2}=\nabsak_{\B^i}$ }for $i=1,2$.\footnote{\changereport{This is not to be confused with slot labelling in abstract index notation.}} }We will also use the momentum-rescaled Sasaki metric ${\G}_0$ which is determined by 
\change{\[\G_0(\A_i,\A_j)=G_{ij},\ \G_0(\A_i,\B^j)=0,\ \G_0(\B^i,\B^j)=(v^0)^{-2}\,(G^{-1})^{ij}\,.\]}
The reference metrics $\underline{\gamma}$ and $\underline{\gamma}_0$ on $P_t$ are defined analogously for $i,j=1,2$.
\change{\begin{align*}
\underline{\gamma}\left(\del_{x^i}+v_m\Gamma[\gamma]^m_{li}\,\del_{v_l},\del_{x^j}+v_s\Gamma[\gamma]^s_{rj}\,\del_{v_r}\right)=&\,\gamma_{ij},\\ \underline{\gamma}(\del_{x^i}+\Gamma[\gamma]^m_{li}\,\del_{v_l},\del_{v_j})=&\,0\,,\\
\underline{\gamma}(\del_{v_i},\del_{v_j})=&\,(\gamma^{-1})^{ij}\,;\\[0.5em]
\underline{\gamma}_0\left(\del_{x^i}+v_m\Gamma[\gamma]^m_{li}\,\del_{v_l},\del_{x^j}+v_s\Gamma[\gamma]^s_{rj}\,\del_{v_r}\right)=&\,\gamma_{ij},\\ \underline{\gamma}_0(\del_{x^i}+v_m\Gamma[\gamma]^m_{li}\,\del_{v_l},\del_{v_j})=&\,0\,,\\
\underline{\gamma}_0(\del_{v_i},\del_{v_j})=&\,(m^2+\lvert v\rvert_\gamma^2)^{-1}\,\gamma_{ij}\,.
\end{align*}}

\subsection{The expansion normalised $2+1$ system in CMCTC gauge}\label{subsec:eq-2+1}

The rescaled metric and shear satisfy the following evolution equations.
\begin{subequations}
\begin{align*}
\numberthis\label{eq:REEqG-2+1}\del_tG_{ij}=&\,-2\,a^{-2}\,(N+1)\,\hat{k}_{ij}+2\,N\,\frac{\dot{a}}a\,G_{ij}\,,\\
\numberthis\label{eq:REEqk}\del_t{(\hat{k}^{\sharp})^i}_{j}=&\,\frac12\,(N+1)\,R[G]\,\I^{i}_{j}+\left(4\pi\,a^{-1}\,\rho^{Vl}_{FLRW}-2\,\kappa\right)\I^i_j+\tau\,N\,{(\hat{k}^{\sharp})^i}_{j}-\nabla^{\sharp i}\nabla_jN\\
&\quad -8\pi\nabla^{\sharp i}\phi\,\nabla_j\phi-8\pi\,a^{-1}\left({(S^{Vl,\parallel})^{\sharp i}}_j+\changereport{\frac12(\rho^{Vl}-\rho^{Vl}_{FLRW})\,\I^i_j}\right)\,.
\end{align*}
The Christoffel symbols and the scalar curvature with respect to $G$ evolve as follows:
\begin{align*}
\numberthis\label{eq:REEqChr-2+1}\del_t\Gamma[G]^i_{jl}=&\,-a^{-2}\left[\nabla_j\left((N+1)\,{(\hat{k}^{\sharp})^i}_l\right)+\nabla_l\left((N+1)\,{(\hat{k}^{\sharp})^i}_j\right)-\nabla^{\sharp i}\left((N+1)\,\hat{k}_{jl}\right)\right]\\
&\quad+\frac{\dot{a}}a\left[\I^i_l\,\nabla_jN+\I^i_j\,\nabla_lN-G_{jl}\,\nabla^{\sharp i}N\right]\,,
\\
\numberthis\label{eq:REEqR}\del_tR[G]=&\,-2a^{-2}\,\div_G\div_G\left((N+1)\,\hat{k}\right)-4\,N\,\frac{\dot{a}}a\,\left(R[G]-2\kappa\right)-8\,\kappa\,N\,\frac{\dot{a}}a \,.
\end{align*}
The Hamiltonian and momentum constraints take the form
\begin{align}
R[G]-2\,\kappa=&\,a^{-2}\,\lvert\hat{k}\rvert_G^2+8\pi\,a^{-2}\left[\lvert \Psi\rvert^2+2\,C\,\Psi\right]+16\pi\,a^{-1}\left({\rho}^{Vl}-{\rho}^{Vl}_{FLRW}\right),\label{eq:REEqHam-2+1}\\
(\div_G\hat{k})_{i}=&\,-8\pi\,(\Psi+C)\,\nabla_i\phi-8\pi\,{\j}_i^{Vl}\label{eq:REEqMom-2+1}.
\end{align}
Using \eqref{eq:tau2}-\eqref{eq:delt-tau-2+1}, the lapse equations can be written as
\begin{equation}\label{eq:REEqN}
\L N=H,\quad \Ltilde N=\tilde{H}
\end{equation}
with the elliptic operators
\[\L\zeta=a^2\Lap_G\zeta-h\,\zeta\quad \text{and}\quad \Ltilde\zeta=a^2\,\Lap_G\zeta-\tilde{h}\,\zeta\,,\]
where
\begin{align*}
h=&\,\left[16\pi\,C^2+16\pi\,a\left(\rho^{Vl}_{FLRW}+\mathfrak{p}^{Vl}_{FLRW}\right)-2\,\kappa\,a^2\right]+H\,,\\
\tilde{h}=&\,\left[16\pi\,C^2+8\pi a\left(\rho^{Vl}_{FLRW}+2\,\mathfrak{p}^{Vl}_{FLRW}\right)-2\,\kappa\,a^2\right]+\tilde{H}\,,\\
H=&\,\lvert \hat{k}\rvert_G^2+8\pi\,\Psi^2+16\pi\,C\,\Psi+8\pi\,a\left(\rho^{Vl}-\rho^{Vl}_{FLRW}\right)\text{ and}\\
\tilde{H}=&\,a^2\left(R[G]-2\kappa\right)-8\pi\,a^2\,\lvert\nabla\phi\rvert_G^2-8\pi\,a\left(\rho^{Vl}-\rho^{Vl}_{FLRW}\right)\,.
\end{align*}
The wave equation takes the form
\begin{align}
\label{eq:REEqnablaphi}\del_t\nabla\phi=&\,a^{-2}\,(\Psi+C)\,\nabla N+a^{-2}\,(N+1)\,\nabla\Psi\,,\\
\label{eq:REEqWave-2+1}\del_t\Psi=&\,\div_G\left((N+1)\,\nabla\phi\right)+2\,N\,\frac{\dot{a}}a\,(\Psi+C)\,.
\end{align}
\change{The Vlasov equation can be written as
\begin{equation}\label{eq:REEqVlasov}
\del_t{f}=-a^{-1}\,\frac{v^{\sharp j}}{v^0}\,(N+1)\,{\A}_j{f}+a^{-1}\,v^0\,\nabla_{j}N\,{\B}^jf.
\end{equation}

To extend the Vlasov equation to higher orders, we introduce the following differential operator $\X$ acting on time-dependent tensors $\mathfrak{T}$ on \change{$T^\ast M$}.

\begin{align*}\numberthis\label{eq:REEqVlasov-gen}
\X\mathfrak{T}:=&\,-a^{-1}\,\frac{v^j}{v^0}\,(N+1)\,\nabsak_j\mathfrak{T}+a^{-1}\,v^0\,\nabla_{j}N\,\B^{j}\mathfrak{T}\,.
\end{align*}}
\end{subequations}

Additionally, for time-dependent functions $\zeta$  on $M$ and $\xi$ on $TM$, \eqref{eq:REEqG-2+1} implies the following formulas that we will use below without further comment.
\begin{align*}
\del_t\int_M\zeta\,\vol{G}=&\,\int_M\left(\del_t\zeta-N\tau\zeta\right)\,\vol{G}\\
\change{\del_t\int_{T^\ast M}\xi\,\vol{\G}=}&\change{\,\int_{T^\ast M}\del_t\xi\,\vol{\G}}
\end{align*}
Finally, note
\begin{equation}\label{eq:nabla-int-TM}
\nabla_{i_1}\dots\nabla_{i_\ell}\int_{\change{T^\ast_\bullet M}}\xi\,\change{\vol{(\G\vert_{vert})_{\bullet}}}=\int_{\change{T^\ast_\bullet M}}\nabsak_{i_1}\dots\nabsak_{i_\ell}\xi\,\change{\vol{(\G\vert_{vert})_{\bullet}}}\,.
\end{equation}

\section{Initial data, local well-posedness and bootstrap assumptions}\label{sec:id-bs}

In this section, we state the initial data assumptions, describe how they initialize a unique maximal globally hyperbolic development toward the past and state the bootstrap assumptions.

\subsection{Solution norms and energies}

Let $\ell\geq 0$ be an integer and $\mu\in[0,\infty)$. For functions $\zeta$ on $\M$ and $\xi$ on the mass shell $P$, we introduce the following Sobolev and supremum norms:
\begin{align*}
\|\zeta\|_{H^\ell_G(M_t)}=&\,\left(\sum_{I=0}^\ell\int_{M_t}\lvert \nabla^I\zeta\rvert_G^2\,\vol{G}\right)^\frac12&
\|\xi\|_{H_{\mu,\G_0}^\ell(\change{T^\ast M_t})}=&\,\left(\sum_{I=0}^\ell\int_{\change{T^\ast M_t}}\langle v\rangle_G^{2\mu}\,\lvert\nabsak^{I}\xi\rvert_{\G_0}^2\,\vol{\G}\right)^\frac12\\
\|\zeta\|_{C^\ell_G(M_t)}=&\,\sum_{I=0}^\ell\sup_{x\in M_t}\lvert \nabla^I\zeta(t,x)\rvert_G&
\|\xi\|_{C_{\mu,\G_0}^\ell(\change{T^\ast M_t})}=&\,\sum_{I=0}^\ell\sup_{(x,v)\in \change{T^\ast M_t}}\langle v\rangle_G^{\mu}\,\left\lvert\nabsak^{I}\xi(t,x,v)\right\rvert_{\G_0}
\end{align*}
\vspace{-1em}
\change{\[\|\xi\|_{C_{\mu,\underline{\gamma}_0}^\ell(\change{T^\ast M_t})}=\sum_{I=0}^\ell\sup_{(x,v)\in \change{T^\ast M_t}}\langle v\rangle_{\gamma}^{\mu}\,\left\lvert\nabla[\underline{\gamma}]^{I}\xi(t,x,v)\right\rvert_{{\underline{\gamma}}_0}\]}

Using the rescaled variables introduced above, we can now define our main solution norms:
\begin{definition}[Solution norms]
\begin{subequations}
\begin{align*}
\numberthis\label{def:H}\H(t)=&\,\|\Psi\|_{H_G^{10}(M_t)}+\|\nabla\phi\|_{H_G^{9}(M_t)}+a(t)\,\|\nabla\phi\|_{\dot{H}^{10}_G(M_t)}+\|\hat{k}\|_{H_G^{10}(M_t)}+\|G-\gamma\|_{H^{10}_G(M_t)}\\
&\,+\|R[G]-2\,\kappa\|_{H^{8}_G(M_t)}+a(t)\,\|R[G]-2\,\kappa\|_{\dot{H}^{9}_G(M_t)}+\|f-f_{FLRW}\|_{H^{10}_{1, {\G}_0}(\change{T^\ast M_t})}\\
&\,+a(t)^{-1}\,\|N\|_{H^{8}_G(M_t)}\\
\numberthis\label{def:C}\C(t)=&\,\|\Psi\|_{C_G^{8}(M_t)}+\|\nabla\phi\|_{C_G^{7}(M_t)}+a(t)\,\|\nabla\phi\|_{C_G^8(M_t)}+\|\hat{k}\|_{C_G^{8}(M_t)}+\|G-\gamma\|_{C^{8}_G(M_t)}\\
&\,+\|R[G]-2\,\kappa\|_{C^{6}_G(M_t)}+\|\rho^{Vl}-\rho^{Vl}_{FLRW}\|_{C^8_G(M_t)}+\|S^{Vl,\parallel}\|_{C^8_G(M_t)}\\
&\,+a(t)^{-1}\,\|N\|_{C^{6}_G(M_t)}
\end{align*}
\end{subequations}
\end{definition}
We will see in the proof of Theorem \ref{thm:main-2+1} that $\mathcal{H}$ essentially controls $\mathcal{C}$, as in \cite[Corollary 7.8]{FU25}. Thus, the goal is to control $\mathcal{H}$ in a bootstrap argument, using the following associated energies:

\begin{definition}[Energies]\label{def:en}
Let $K,L$ be integers with $L\geq K\geq 0$. \change{We define the following energies for the scalar field, the shear, the lapse, the scalar curvature and the Vlasov distribution.}
\begin{subequations}
\begin{align}
\label{eq:def-en-sf}\E^{(L)}(\phi,\cdot)=&\,(-1)^L\int_{M}\left[\Psi\,\Lap_G^L\Psi-a^2\,\phi\,\Lap_G^{L+1}\phi\right]\vol{G}\,,\\
\E^{(L)}(\hat{k},\cdot)=&\,(-1)^L\int_M\langle\hat{k},\Lap_G^L\hat{k}\rangle_G\,\vol{G}\,,\\
\E^{(L)}(N,\cdot)=&\,(-1)^L\int_MN\Lap_G^LN\,\vol{G}\,,\\
\label{eq:def-en-R}\E^{(L)}(R,\cdot)=&\,(-1)^L\int_M \left(R[G]-2\,\kappa\right)\Lap_G^{L}(R[G]-2\,\kappa)\,\vol{G}\,,\\
\label{eq:def-en-Vl0}\E^{(L)}_{\mu,0}(f,\cdot)=&\,\int_{\change{T^\ast M_t}}\langle v\rangle^{2\mu}_G\,\lvert \nabsak^L_{vert}(f-f_{FLRW})\rvert^2_{\G_0}\,\vol{\G}\,,\\
\label{eq:def-en-VlK}\E^{(L)}_{\mu,K}(f,\cdot)=&\,\int_{\change{T^\ast M_t}} \langle v\rangle^{2\mu}_G\,\lvert\nabsak_{\mathrm{vert}}^{L-K}\nabsak_{\mathrm{hor}}^{K}f\rvert_{\G_0}^2\,\vol{\G}\,,\\
\E^{(\leq L)}_{\mu,\leq K}(f,\cdot)=&\sum_{R=1}^L\sum_{S=0}^{\min\{K,R\}}\E^{(R)}_{\mu,S}(f,\cdot)\,.
\end{align}
Let $\epsilon>0$ denote the square root of the perturbation size, see Assumption \ref{ass:init-2+1}, \change{and let $\omega>0$ be sufficiently small, as discussed below. }Then, we define the following total energies:
\change{\begin{align*}
\numberthis\label{eq:total-en-def}\E^{(L)}_\textrm{\normalfont total}=&\,\E^{(L)}(\phi,\cdot)+\left(\epsilon^\frac14+a^\frac{\omega}4}\right)\,\E^{(L)}(\hat{k},\cdot)+\left(\epsilon^\frac12+a^\frac{\omega}2\right)\,\E^{(L-2)}(R,\cdot)\\
&\,+\E^{(L)}_{total,Vl}+a^{2+\frac{\omega}4}\E^{(L-1)}(R,\cdot)\\
&\,+\begin{cases}
a^\frac{\omega}2\|{G^{\pm 1}-\gamma^{\pm 1}}\|_{L^2_G}^2+a^{2+\frac{\omega}2}\,\|\Gamma-\Gamhat\|_{L^2_G}^2 & L=0\\
a^\frac{\omega}2\|{G^{\pm 1}-\gamma^{\pm 1}}\|_{H^2_G}^2+a^{\frac{\omega}2}\,\|\Gamma-\Gamhat\|_{{H}^1_G}^2 & L=2\\
0 & \text{else}
\end{cases}\,,\\
\numberthis\label{eq:total-en-def-Vl}\E^{(L)}_\textrm{\normalfont total,Vl}=&\,\sum_{K=0}^La^{(K+1)\omega}\E^{(L)}_{1,K}(f,\cdot)\,,\\
\numberthis\label{eq:total-en-def-quiesc-2+1}\E^{(L)}_\textrm{\normalfont quiesc}=&\,\E^{(L)}(\phi,\cdot)+\epsilon^\frac14\E^{(L)}(\hat{k},\cdot)+\epsilon^\frac12\E^{(L)}(R,\cdot)\,.
\end{align*}
\end{subequations}
\end{definition}
\change{The total Vlasov energy is effectively weighted by $a^\omega$ to mitigate terms in the Vlasov integral estimates that would otherwise diverge at order $t^{-1}$ or slightly worse and prevent one from being able to obtain strong bounds from energy estimates. Essentially, one is able to show weak control on Vlasov matter by obtaining improved estimates for $\E^{(L)}_\textrm{\normalfont total}$, which suffices to obtain a strong bound on $\E^{(L)}_\textrm{\normalfont quiesc}$ since the scalar field dominates Vlasov matter toward the Big Bang. In turn, this strong bound on the quiescent variables yields a strong bound on Vlasov matter.  However, as in \cite{FU25}, one needs to introduce time-weighted geometric energies and metric error terms to ensure that the Vlasov energy hierarchy is compatible with the mechanism in $\E^{(L)}_{quiesc}$. We note that, in contrast to \cite{FU25}, we do not require an $\epsilon$-weight on the high order curvature term in $\E^{(L)}_\textrm{\normalfont total}$, since such a term now no longer occurs naturally in the quiescent energy due to the simplified energy estimates for $\hat{k}$ and $R$, see Section \ref{sec:en-est-2+1}.\\ }
In the following, $\omega$ will be small enough, but significantly larger than the perturbation size $\epsilon^2$ and the bootstrap parameter $\eta$; see Assumption \ref{ass:bs}. If $\eta$ and $\epsilon$ are taken to be sufficiently small, one can choose, for example, $\omega=\frac1{100}$.\\

Furthermore, we will often assume $L$ to be an even integer in the sequel to simplify the exposition. For the energies defined in \eqref{eq:def-en-sf}-\eqref{eq:def-en-R}, integration by parts immediately implies that, for $c>0$,
\begin{equation}\label{eq:ibp-trick-2+1}
\E^{(L-1)}\leq\frac12\left(\E^{(L)}+\E^{(L-2)}\right),\quad\text{and}\quad a^{-c\sqrt{\epsilon}}\,\E^{(L-1)}\leq\frac12\left(\E^{(L)}+a^{-2c\sqrt{\epsilon}}\E^{(L-2)}\right).
\end{equation}
Due to the weights in \eqref{eq:def-en-Vl0}-\eqref{eq:def-en-VlK}, this is only approximately true for Vlasov energies, but this will still be sufficient to improve \eqref{eq:total-en-def}; see \eqref{eq:ibp-Vlasov-1-1} for the Vlasov energy analogue. We also note that \eqref{eq:ibp-Vlasov-1-1} is the reason why one needs to consider metric error terms of up to second order in $\E^{(2)}_\mathrm{ total}$.

\subsection{Initial data}\label{subsec:init-2+1}

\change{Initial data for }the $(2+1)$-dimensional ESFV system \eqref{eq:EEq} takes the form $(M,\mathring{g},\mathring{k},\mathring{\pi},\mathring{\psi},\mathring{f})$, 
where $\mathring{g}$ and $\mathring{k}$ are symmetric $(0,2)$-tensors on $M$, $\mathring{\pi}$ is \changereport{an exact }(0,1)-tensor and $\mathring{\psi}$ and $\mathring{f}$ are scalar fields on $M$ and \change{$T^\ast M$}, respectively and which satisfy the constraint equations
\begin{align*}
R[\mathring{g}]+(\mathrm{tr}_{\mathring{g}}\mathring{k})^2-\lvert\mathring{k}\rvert_{\mathring{g}}^2=&\,8\pi\left[\mathring{\psi}^2+\lvert\mathring{\pi}\rvert_{\mathring{g}}^2\right]+16\pi\,\int_{\change{T^\ast_{\bullet}M}}\sqrt{m^2+\lvert p\rvert_{\mathring{g}}^2}\,\mathring{f}(\cdot,p)\,\vol{\change{\mathring{g}_\bullet^{-1}}}\,,\\
\div_{\mathring{g}}\mathring{k}=&\,-8\pi\,\mathring{\psi}\,\nabla[\mathring{g}]\mathring{\pi}-8\pi\int_{\change{T^\ast_{\bullet}M}}p\,\mathring{f}(\cdot,p)\,\vol{\change{\mathring{g}_\bullet^{-1}}}\,.
\end{align*}
We also assume that $\mathring{f}$ is nonnegative and that the momentum support of $\mathring{f}$ is compact, i.e., that
\[\sup\{\lvert p\rvert_{\mathring{g}}\,\vert\,(x,p)\in\supp\mathring{f}\}<\infty\,.\]
In the massless case $(m=0)$, we additionally assume that
\begin{equation}\label{eq:massless-cond}
\inf\{\lvert p\rvert_{\mathring{g}}\,\vert\,(x,p)\in\supp\mathring{f}\}>0\,.
\end{equation}
This is required to ensure well-posedness\change{, }see Section \ref{subsec:lwp-2+1}. The CMC condition is enforced by requiring 
\[\text{tr}_{\mathring{g}}\mathring{k}=-2\,\frac{\dot{a}(t_0)}{a(t_0)}\,.\] 
That such data, embedded into a spacetime $(\M,\g)$ solving \eqref{eq:EEq}, admits unique maximal globally hyperbolic developments follows from standard arguments as in \cite{FB52,CBGer69}.\\

We assume such data $(g,k,\nabla\phi,e_0\phi,\f)$ to be close to FLRW data in the following sense:

\begin{assumption}[Initial data assumption]\label{ass:init-2+1} At the initial time $t_0\in(0,\change{T/2})$ (see Lemma \ref{lem:scale-factor-2+1}) and $\epsilon\in(0,1)$ small enough, we take \change{initial data for the ESFV }system to be close to FLRW data in the following sense:
\begin{subequations}
\begin{align}
\label{eq:init-HC}\H(t_0)+\C(t_0)\leq&\,\epsilon^2\,,\\
\label{eq:init-Vlasov-low}\|f-f_{FLRW}\|_{C_{1,{\G}_0}^5(\change{T^\ast M_{t_0}})}\leq&\,\,\epsilon^2\,.
\end{align}
\end{subequations}
\end{assumption}
Note that \eqref{eq:init-HC} implies that all energies in Definition \ref{def:en} of order $L\leq 10$ are bounded by $\epsilon^4$ at $t=t_0$. 

\begin{remark}[Polarized $U(1)$-symmetric \change{initial data for the $(3+1)$ Einstein }vacuum equations]\label{rem:u1-init}
For Corollary \ref{cor:u1-intro}, we consider initial data $(M\times \S^1,\mathring{\tilde{g}},\mathring{\tilde{k}})$ to the vacuum equations that is polarized and $U(1)$-symmetric, i.e., $\mathring{\tilde{g}}$ and $\mathring{\tilde{k}}$ satisfy
\[\mathring{\tilde{g}}_{13}=\mathring{\tilde{g}}_{23}=\mathring{\tilde{k}}_{13}=\mathring{\tilde{k}}_{23}=0\]
on top of the vacuum constraint equations. Given the correspondence \eqref{eq:u1-corresp}, this induces initial data for the scalar field given by
\begin{subequations}\label{eq:u1-init-corresp}
\begin{equation}
\mathring{\pi}_i=\frac1{2\sqrt{4\pi}}\,\del_i\log(\mathring{\tilde{g}}_{33}),\quad \mathring{\psi}=-\frac{1}{\sqrt{4\pi}}\,\left(\mathring{\tilde{g}}_{33}\right)^{-\frac32}\,\mathring{\tilde{k}}_{33},
\end{equation}
and consequently data for the metric and the second fundamental form by 
\begin{equation}
\mathring{g}_{ij}=\left(\mathring{\tilde{g}}_{33}\right)^{-1}\,\mathring{\tilde{g}}_{ij},\quad \mathring{k}_{ij}=\sqrt{\left(\mathring{\tilde{g}}_{33}\right)}\,\mathring{\tilde{k}}_{ij}-\sqrt{4\pi}\,\mathring{\psi}\,\mathring{g}_{ij}\,.
\end{equation}
\end{subequations}
\end{remark}

\subsection{Local well-posedness}\label{subsec:lwp-2+1}

A local well-posedness theory for the $(2+1)$-dimensional Einstein-Vlasov-scalar-system in CMCTC gauge can be developed following the same arguments as in \cite[Section 3.3]{FU25}. More precisely, one first combines the well-posedness results \cite[Corollary 24.10]{Rin13} and \cite[Paper B, Lemma 9.2]{Sve12} in harmonic gauge with the approach in \cite{FajKr20} to find a CMC hypersurface close to general near-FLRW data, allowing one to assume that any initial data close to FLRW data is CMC without loss of generality. Then, one follows the proof of \cite[Theorem 14.1]{Rodnianski2014} to obtain local well-posedness results in CMCTC gauge, along with continuation criteria. \change{We note, as in \cite{FU25}, that a direct extension of the arguments in \cite{Rin13,Sve12} requires working with the dual Vlasov distribution $f^\sharp$ in terms of canonical coordinates on the mass shell, but the spacetime is sufficiently regular that one can translate between both frameworks at no additional cost.}\\

The new nuance compared to the $(3+1)$-dimensional case arises when considering the results in \cite[Paper B]{Sve12}, which are needed to obtain an analogue of the harmonic gauge well-posedness result \cite[Lemma 3.6]{FU25} for $m=0$, since \cite[Paper B]{Sve12} is restricted to odd spatial dimensions. However, as noted in \cite[Remark 5.2]{Sve12}, all of the results apply to even spatial dimensions if one adds an additional momentum weight $\langle p\rangle_\gamma^\alpha\lvert p\rvert_\gamma^{-\alpha}$ for $\alpha\in\R\backslash\Z,\,\alpha>0$ to all weighted norms for the Vlasov distribution function throughout the work and adapts the function spaces accordingly. Consequently, analogous additional momentum weights would have to enter our regularity assumptions and continuation criteria a priori. However, since we assume \eqref{eq:massless-cond} anyhow in the massless case, one verifies that these weights are uniformly bounded from above and below using the method of characteristics, as in Lemma \ref{lem:APMom-2+1}. Similarly, adding to the continuation criteria that
\[t\mapsto \inf\{\lvert p\rvert_g\,\vert\, x\in M,\,(t,x,p)\in\supp\f\}\]
is uniformly bounded from below by a positive constant is sufficient to resolve this technical hurdle.\\

\change{\noindent In summary, the following holds.
\begin{lemma}[Local well-posedness and continuation criteria for the $(2+1)$ ESFV system]
For sufficiently small $\epsilon>0$, the $(2+1)$-dimensional ESFV system \eqref{eq:EEq} in CMCTC gauge with initial data as in Assumption \ref{ass:init-2+1} is locally well-posed, launching a foliation $(M_t)_{t\in(t_-,t_0]}$ of $\M$ for some $t_->0$. The solution can be extended to $t=t_-$ as long as, approaching $t_-$, the following continuation criteria hold for some $Q>0$:
\begin{enumerate}
\item The eigenvalues of $g_{t,x}$ remain bounded from below by $Q$.
\item One has $n\geq Q$.
\item One has $\lvert e_0\phi\rvert^2+\lvert\nabla\phi\rvert_g^2\geq Q$.
\item For $m=0$, one has
\[\inf_{t\in(t_-,t_0]}\inf\{\lvert p\rvert_g\,\vert\, x\in M,\,(t,x,p)\in\supp\f\}\geq Q\,.\]
\item The following quantities remain uniformly bounded.
\begin{gather*}
\|g\|_{C^2_\gamma(M_t)},\ \|\hat{k}\|_{C^1_\gamma(M_t)},\ \|e_0\phi\|_{C^1_\gamma(M_t)},\ \change{\|\nabla\phi\|_{C^1_\gamma(M_t)},\ }\|\rho^{Vl}\|_{C^1_\gamma(M_t)},\ \|\j^{Vl}\|_{C^1_\gamma(M_t)},\ \|S^{Vl}\|_{C^1_\gamma(M_t)},\\
\change{\sup_{x\in M}\int_{T^\ast_xM}\langle p\rangle_\gamma^2\,\left[\lvert\f(t,x,p)\rvert^2+(m^2+\lvert p\rvert_\gamma^2)\,\lvert \del_{p_i}\f(t,x,p)\rvert_\gamma^2\right]\,\vol{\mathring{g}^{-1}_x}}\,.
\end{gather*}
\end{enumerate}
\end{lemma}}

With regards to the final bound for the integrated Vlasov distribution, the equivalent criterion \cite[(3.11c)]{FU25} in the higher dimensional setting contains the additional weight $\langle p\rangle_\gamma^{2\ell}(p_\gamma^0)^{-2\ell}$, where the integer $\ell$ depends on the regularity of the initial data. However, in the massive case, this weight is identical to $1$ and in the massless case, it is uniformly bounded from above and below on the support of $f$ and can thus be dropped due to criterion (4).\\

In addition, one can choose the data to be sufficiently regular such that all energies used within this paper are continuously differentiable in time, which we tacitly assume on top of Assumption \ref{ass:init-2+1} without loss of generality. This does not impact the continuation criteria (1)--(5).

\begin{remark}[On local well-posedness for polarized $U(1)$-symmetric vacuum solutions]
For the sake of this work, we refrain from separately establishing a local well-posedness statement for the vacuum solutions arising from polarized $U(1)$-symmetric vacuum data as discussed in Remark \ref{rem:u1-init}. Instead, the local well-posedness result for the \change{Einstein (Vlasov) scalar-field system }in CMCTC gauge and $(2+1)$-dimensions is sufficient for our purposes: Polarized $U(1)$-symmetric vacuum data induces data for the Einstein scalar-field system; see Remark \ref{rem:u1-init}. The maximal Einstein scalar-field solution this launches then induces a globally hyperbolic development of the original vacuum data by the correspondence \eqref{eq:u1-corresp}. $U(1)$-symmetry and polarization are then propagated in the induced foliation $(M_t\times\S^1)_{t\in(t_-,t_0]}$ of $\check{M}$ by construction. While this development is not a priori past maximal, as evidenced by the extremal Kasner solution discussed in Remark \ref{rem:u1-ref}, for the solutions in Corollary \ref{cor:u1-intro}, we show that the Kretschmann scalar of the vacuum solution blows up as $t\downarrow 0$, ensuring that the full past development is obtained. 
\end{remark}

\subsection{Bootstrap assumptions}

The bootstrap interval will be denoted by $(t_{Boot},t_0]$, $t_{Boot}\in[0,t_0)$. On this interval, we make the following bootstrap assumptions:

\begin{assumption}[Bootstrap assumptions]\label{ass:bs}
For $\eta\gg\epsilon^\frac18>0$ sufficiently small, one has
\begin{subequations}
\begin{align}
\label{eq:bs-ass-HC}\C(t)\lesssim&\, \epsilon\,a(t)^{-c\,\eta}\quad\text{and}\\
\label{eq:bs-ass-Vlasov}\|f-f_{FLRW}\|_{C^5_{1,\G_0}(\change{T^\ast M_t})}\lesssim&\,{\epsilon}^\frac14\,a(t)^{-c\,\eta}
\end{align}
for any $t\in (t_{Boot},t_0]$.
\end{subequations}
\end{assumption}
One may take $\eta=\epsilon^\frac1{16}$. The goal of Sections \ref{sec:AP} and \ref{sec:en-est-2+1} is to obtain \eqref{eq:bs-imp} for some $c^\prime>0$, which improves the bootstrap assumptions for sufficiently small $\epsilon>0$. Consequently, we may implicitly update the constant $c$ from line to line throughout the bootstrap argument.

Throughout the bootstrap argument in Section \ref{sec:AP} and \ref{sec:en-est-2+1}, we always assume $t\in(t_{Boot},t_0]$.

\section{A priori estimates}\label{sec:AP}

\subsection{Strong low order pointwise bounds}
In this section, we collect strong bounds in $C_G^\ell$ for $\ell\leq 6$ that immediately follow from the bootstrap assumption. For the most part, these bounds can be derived as in \cite[Section 4]{FU25}, so we only sketch these arguments. 

\begin{lemma}[Zero order bounds for kinetic variables]\label{lem:AP0-2+1}
\begin{subequations}
\begin{align}
\|\hat{k}\|_{C^0_G(M_t)}\lesssim&\,\epsilon\,,\label{eq:APkhat0}\\
\|\Psi\|_{C^0(M_t)}\lesssim&\,\epsilon\,.\label{eq:APPsi0}
\end{align}
\end{subequations}
\end{lemma}
\begin{proof}
Using \eqref{eq:REEqk} and that the $\hat{k}$ is tracefree, we compute
\begin{align*}
\frac12\,\del_t\left(\lvert \hat{k}\rvert_G^2\right)=&\,\frac12\,\del_t\left((\hat{k}^{\sharp})^{b}_{\ c}\,(\hat{k}^{\sharp})^c_{\ b}\right)\\
=&\,-\langle \nabla^2N,\hat{k}\rangle_G-8\pi\,\langle \nabla\phi\,\nabla\phi,\hat{k}\rangle_G-8\pi\,a^{-1}\,\langle S^{Vl,\parallel},\hat{k}\rangle_G-N\,\tau\,\lvert \hat{k}\rvert_G^2\,.
\end{align*}

The bootstrap assumptions \eqref{eq:bs-ass-HC} imply
\begin{equation*}
\lvert\hat{k}\rvert_G(t)\lesssim \lvert\hat{k}\rvert_G(t_0)+\int_t^{t_0}\epsilon\, a(s)^{-1-c\eta}\,ds +\int_t^{t_0}\epsilon\,a(s)^{-1-c\sqrt{\epsilon}}\,\lvert \hat{k}\rvert_G(s)\,ds\,.
\end{equation*}
Estimate \eqref{eq:APkhat0} now follows from the Gronwall lemma, Lemma \ref{lem:scale-factor-2+1} and the initial data assumption \eqref{eq:init-HC}. An analogous argument proves \eqref{eq:APPsi0}.
\end{proof}
As a direct consequence, we obtain that the momentum support remains compact and, in the massless case, bounded away from the origin:
\begin{lemma}[Momentum support bounds]\label{lem:APMom-2+1} The following momentum support bound holds.
\change{\begin{subequations}
\begin{equation}\label{eq:APMom-2+1}
\sup_{x\in M}\sup_{v\in\supp f(t,x,\cdot)}v^0\lesssim a(t)^{-c\sqrt{\epsilon}},\quad \sup_{x\in M}\sup_{v\in\supp f(t,x,\cdot)}\lvert v\rvert_\gamma\lesssim 1\,.
\end{equation}
For $m=0$, one additionally has
\begin{equation}\label{eq:APMomMassless-2+1}
\sup_{x\in M}\sup_{v\in\supp f(t,x,\cdot)}\lvert v\rvert_G^{-1}\lesssim a(t)^{-c\sqrt{\epsilon}},\quad \sup_{x\in M}\sup_{v\in\supp f(t,x,\cdot)}\lvert v\rvert_\gamma^{-1}\lesssim 1\,.
\end{equation}
\end{subequations}}
\end{lemma}
\begin{proof}Let $x\in M_t$ and $v\in\supp f(t,x,\cdot)$. By the method of characteristics, the solution of \eqref{eq:REEqVlasov} is given by
\[f(t,x,v)=\mathring{f}(t,X(t;x,v),V(t;x,v)),\]
where, writing $V^0=\sqrt{m^2\,a^2+\lvert v\rvert_G^2}\,$, $V$ solves the initial value problem
\change{\begin{align*}
\dot{V}_i=&\,-a^{-1}\,(N+1)\,(G^{-1})^{jk}\Gamma[G]^l_{ij}\,\frac{V_kV_l}{V^0}+a^{-1}V^0\,\nabla^{\sharp i}N\,,\\
V(t_0;x,v)=&\,v
\end{align*}}
and $X$ is given by
\begin{align*}
\dot{X}^i=&\,a^{-1}\,(N+1)\,\frac{(G^{-1})^{ij}V_j}{V^0}\,,\\
X(t_0;x,v)=&\,x\,.
\end{align*}
\change{The bootstrap assumption \eqref{eq:bs-ass-HC} then yields
\[\lvert \dot{V}\rvert_G\lesssim a^{-1-c\sqrt{\epsilon}}\,V^0\,.\]

The strong low order bound \eqref{eq:APkhat0},  along with \eqref{eq:bs-ass-HC}, implies $\lvert \del_tG\rvert_G\lesssim \epsilon a^{-2}$ using \eqref{eq:REEqG-2+1}. Thus, \eqref{eq:APMom-2+1} follows from these bounds using the Gronwall lemma and that the momentum support is bounded initially. \eqref{eq:APMomMassless-2+1} is proven analogously, noting that $\del_t\gamma=0$.
\end{proof}}

From this, one obtains the following bounds, which in turn control low order nonlinear error terms in the energy estimates:

\begin{lemma}[Strong low order $C_G$-bounds]\label{lem:AP-2+1}
\begin{subequations}
\begin{align}
\|\hat{k}\|_{C_G^6(M_t)}+\|\Psi\|_{C_G^6(M_t)}\lesssim&\,{\epsilon}\,a(t)^{-c\sqrt{\epsilon}}\,,\label{eq:APkin}\\
\|\nabla\phi\|_{C_G^5(M_t)}+\|G-\gamma\|_{C_G^6(M_t)}+\|R[G]-2\,\kappa\|_{C_G^4(M_t)}\lesssim&\,\sqrt{\epsilon}\,a(t)^{-c\sqrt{\epsilon}}\,,\label{eq:APder}\\
\|f-f_{FLRW}\|_{C_{1,\G_0}^5(\change{T^\ast M_t})}\lesssim&\,\sqrt{\epsilon}\,a(t)^{-c\sqrt{\epsilon}}\,,\label{eq:APVlasov-2+1}\\
\|\rho^{Vl}-\rho^{Vl}_{FLRW}\|_{C_G^5(M_t)}+\|S^{Vl,\parallel}\|_{C_G^5(M_t)}+\|\j^{Vl}\|_{C_G^5(M_t)}\lesssim&\,\sqrt{\epsilon}\,a(t)^{-c\sqrt{\epsilon}}\,.\label{eq:APVlasovder}
\end{align}
\end{subequations}
\end{lemma}
\begin{proof}
\eqref{eq:APkin}-\eqref{eq:APder} are proven as in \cite[Lemma 4.2]{FU25}, \eqref{eq:APVlasov-2+1} as in \cite[Lemma 4.8 and 4.9]{FU25} and \eqref{eq:APVlasovder} as in \cite[Corollary 4.10]{FU25}. Considering the shear bound in \eqref{eq:APkin}, for example, note that the tracefree part $(\del_t\hat{k})^\parallel$ of $\del_t\hat{k}$ satisfies the following bound for $J\in\{1,2,\dots,6\}$; see \eqref{eq:REEqk}.
\[\left\|(\del_t\hat{k}^\sharp)^\parallel\right\|_{C^{J}_G(M_t)}\lesssim \tau(t)\,\|N\|_{C^J_G(M_t)}\,\|\hat{k}\|_{C^J_G(M_t)}+\|N\|_{C^{J+2}_G(M_t)}+\|\nabla\phi\|_{C^{J}_G(M_t)}^2+a(t)^{-1}\|S^{Vl,\parallel}\|_{C^J_G(M_t)}\,.\]
One also obtains that
\begin{align*}
\left\|[\del_t,\nabla^J]\hat{k}^\sharp\right\|_{C^0_G(M_t)}\lesssim&\,\tau(t)\,\|N\|_{C^J_G(M_t)}\,\|\hat{k}\|_{C^J_G(M_t)}+a(t)^{-2}\,\|\hat{k}\|_{C^0_G(M_t)}\,\|\hat{k}\|_{\dot{C}^J_G(M_t)}\\
&\,+a(t)^{-2}\,\|N+1\|_{C^J_G(M_t)}\,\|\hat{k}\|_{C^{J-1}_G(M_t)}^2\,.
\end{align*}
Assume that the bound $\|\hat{k}\|_{C^{J-1}_G(M_t)}\lesssim \epsilon\,a(t)^{-c\sqrt{\epsilon}}$ is already established -- for $J=1$, this follows from \eqref{eq:APkhat0}. The bootstrap assumption \eqref{eq:bs-ass-HC} and the zero order \eqref{eq:APkhat0} then imply 
\[-\del_t\|\nabla^J\hat{k}\|^2_{\dot{C}^{J}_G(M_t)}\lesssim \epsilon\,a(t)^{-2}\,\|\hat{k}\|^2_{\dot{C}^J_G(M_t)}+\left(\epsilon^2\,a(t)^{-2-c\sqrt{\epsilon}}+\epsilon\,a(t)^{-1-c\eta}\right)\|\hat{k}\|_{\dot{C}^J_G(M_t)}\,.\]
This is sufficient to prove the bound at order $J$ using the Gronwall lemma. Completing the iteration, we obtain the shear bound in \eqref{eq:APkin}. The remaining bound in \eqref{eq:APkin} follows analogously using the already obtained shear bound and \eqref{eq:APder} along with \eqref{eq:APkin}.\\

Regarding \eqref{eq:APVlasov-2+1}, one considers the rescaled Vlasov equation \eqref{eq:REEqVlasov} as an inhomogeneous transport equation for $f-f_{FLRW}$, where the inhomogeneities arise from reference terms depending on derivatives of $f_{FLRW}$, which can be bounded using \eqref{eq:APkin}-\eqref{eq:APder}. At order $0$, the bound follows by controlling $f-f_{FLRW}$ along the characteristics, similar to the proof of Lemma \ref{lem:APMom-2+1}.

From here, one proceeds iteratively after commuting the Vlasov equation with $\A^{K}\B^{J-K}$ for increasing $K=0,\dots,J$.\footnote{Since $M$ is compact and $f$ has compact momentum support, one can restrict the rest of this proof to coordinate neighbourhoods.} For $J=1,K=0$, using \eqref{eq:APkhat0} to control the borderline shear term and \eqref{eq:bs-ass-Vlasov} to control the horizontal derivative terms, one can apply the Gronwall lemma along with \eqref{eq:APMom-2+1} to obtain
\change{\[(v^0)\lvert\B(f-f_{FLRW})\rvert_{\G\vert_{vert}}\lesssim \epsilon^\frac14\,.\]}
Using this bound for $J=1,K=1$, one locally obtains
\[\lvert\A(f-f_{FLRW})\rvert_{G}\lesssim \sqrt{\epsilon}\,a^{-c\sqrt{\epsilon}}\]
and then \eqref{eq:APVlasov-2+1} for $C^1_{1,\G_0}(\change{T^\ast M_t})$ by repeating the first step with this improved bound in place of the bootstrap assumption and combining the local estimates. For $J=2,3,4$ and $5$, one proceeds analogously. Finally, \eqref{eq:APVlasovder} is an immediate consequence of \eqref{eq:nabla-int-TM}, \eqref{eq:APMom-2+1} and \eqref{eq:APVlasov-2+1}.
\end{proof}

\subsection{Relating Sobolev norms to energies}

With these low order estimates in hand, we can now show that our energies are almost coercive:

\begin{prop}[Near-coercivity of energies]\label{lem:near-coerc} Let $\ell$ be an integer with $0\leq\ell\leq 9$. Then, the following estimates hold, where energies with superscript $\ell-2$ or $\ell-3$ only appear when $\ell\geq 5$:
\begin{subequations}\label{eq:near-coerc}
\begin{align*}
\numberthis\label{eq:near-coerc-SF}\|\Psi\|_{H^{\ell}_G(M_t)}+a(t)\,\|\nabla\phi\|_{H^{\ell}_G(M_t)}\lesssim&\,\sqrt{\E^{(\ell)}(\phi,t)}+a(t)^{-c\sqrt{\epsilon}}\left(\sqrt{\E^{(\leq \ell-2)}(\phi,t)}+\epsilon \sqrt{\E^{(\leq \ell-2)}(R,t)}\right)\,,\\
\numberthis\label{eq:near-coerc-khat}\|\hat{k}\|_{H^{\ell}_G(M_t)}\lesssim&\,\sqrt{\E^{(\ell)}(\hat{k},t)}+a(t)^{-c\sqrt{\epsilon}}\left(\sqrt{\E^{(\leq \ell-2)}(\hat{k},t)}+\epsilon^2\sqrt{\E^{(\leq \ell-3)}(R,t)}\right)\,,\\
\numberthis\label{eq:near-coerc-N}\|N\|_{H^{\ell}_G(M_t)}\lesssim&\,\sqrt{\E^{(\ell)}(N,t)}+a(t)^{-c\sqrt{\epsilon}}\left(\sqrt{\E^{(\leq \ell-3)}(N,t)}+\epsilon \sqrt{\E^{(\leq \ell-3)}(R,t)}\right)\,,\\
\numberthis\label{eq:near-coerc-R}\|R[G]-2\,\kappa\|_{H^\ell_G(M_t)}\lesssim&\,\sqrt{\E^{(\ell)}(R,t)}+a(t)^{-c\sqrt{\epsilon}}\sqrt{\E^{(\leq \ell-2)}(R,t)}\,,\\
\numberthis\label{eq:near-coerc-Vl1}\|\rho^{Vl}-\rho^{Vl}_{FLRW}\|_{H^\ell_G(M_t)}\lesssim&\, a(t)^{-c\sqrt{\epsilon}}\left(\sqrt{\E^{(\leq \ell)}_{1,\leq \ell}(f,t)}+\|\change{G^{- 1}-\gamma^{-1}}\|_{L^2_G(M_t)}\right)\,,\\
\numberthis\label{eq:near-coerc-Vl2}\|S^{Vl,\parallel}\|_{H^\ell_G(M_t)}\lesssim&\, \changereport{a(t)^{-c\sqrt{\epsilon}}\left(\sqrt{\E^{(\leq \ell)}_{1,\leq \ell}(f,t)}+\|G^{-1}-\gamma^{-1}\|_{L^2_G(M_t)}}\right)\,,\\
\numberthis\label{eq:near-coerc-Vl3}\|\j^{Vl}\|_{H^\ell_G(M_t)}\lesssim&\,a(t)^{-c\sqrt{\epsilon}}\sqrt{\E^{(\leq \ell)}_{1,\leq \ell}(f,t)}\,,\\
\numberthis\label{eq:near-coerc-Vl4}\change{\|f-f_{FLRW}\|_{H^\ell_{1,\underline{G}_0}(T^\ast M_t)}\lesssim}&\change{\,\sqrt{\E^{(\ell)}_{1,\leq \ell}(f,t)}+a(t)^{-c\sqrt{\epsilon}}\left(\sqrt{\E^{(\leq \ell-1)}_{1,\leq\ell-1}(f,t)}+\epsilon\,\sqrt{\E^{(\leq \ell-2)}(R,t)}}\right)\,,\\
&\,\change{+a(t)^{-c\sqrt{\epsilon}}\left(\|\Gamma-\Gamhat\|_{H^{\ell-1}_G(M_t)}+\|G^{\pm 1}-\gamma^{\pm 1}\|_{H_G^{\ell-1}(M_t)}\right)}\,.
\end{align*}
\end{subequations}
We also have the following bound
\begin{align}
\label{eq:norm-est-nablaphi}\|\nabla\phi\|_{H^9(M_t)}\lesssim&\,a(t)^{-c\sqrt{\epsilon}}\left(\sqrt{\E^{(\leq 10)}(\hat{k},t)}+\sqrt{\epsilon}\,\sqrt{\E^{(\leq 9)}(\phi,t)}+\sqrt{\E^{(\leq 9)}_{1,\leq 9}(f,t)}+\sqrt{\E^{(\leq 8)}(R,t)}\right)
\end{align}
\end{prop}
\begin{proof}
For \eqref{eq:near-coerc-SF}-\eqref{eq:near-coerc-R}, the proof proceeds as in \cite[Lemma 4.5]{FU24} by integration by parts and commuting derivatives, where low order terms in the resulting commutator errors are controlled using Lemma \ref{lem:AP-2+1}. Conversely, \eqref{eq:near-coerc-Vl1}-\eqref{eq:near-coerc-Vl3} follow as in \cite[Lemma 4.14]{FU25}, using \eqref{eq:nabla-int-TM}, \eqref{eq:APMom-2+1} as well as Jensen's inequality. The zero order metric error terms in \eqref{eq:near-coerc-Vl1}-\eqref{eq:near-coerc-Vl3} arise from the reference quantity in the zero order estimate. \change{The Vlasov distribution bound \eqref{eq:near-coerc-Vl4} is proven as in \cite[Lemma 4.13]{FU25}. In particular, higher order metric terms in the second line only arise from computing derivatives of the FLRW reference distribution and not from rearranging derivatives into the specific order required by the Vlasov energy.}

Finally, by rearranging the momentum constraint \eqref{eq:REEqMom-2+1}, differentiating and applying Lemma \eqref{eq:APkin} and \eqref{eq:APVlasovder} to control lower order terms, one obtains
\[\|\nabla\phi\|_{H^9(M_t)}\lesssim a(t)^{-c\sqrt{\epsilon}}\,\left(\|\hat{k}\|_{H^{10}_G(M_t)}+\|\j^{Vl}\|_{H^{9}_G(M_t)}\right)+\sqrt{\epsilon}\,a(t)^{-c\sqrt{\epsilon}}\,\|\Psi\|_{H^{9}_G(M_t)}\,.\]
The bound \eqref{eq:norm-est-nablaphi} then follows from applying \eqref{eq:near-coerc-khat}, \eqref{eq:near-coerc-Vl3} and \eqref{eq:near-coerc-SF}.
\end{proof}

Finally, we can prove the following lemma that allows us to extend \eqref{eq:ibp-trick-2+1} to energies of the Vlasov distribution except for one special case.

\begin{lemma}
Let $K,L$, with \change{$0< K< L$ }and $(K,L)\neq (1,2)$. Then, the following estimates hold.
\begin{subequations}
\begin{equation}\label{eq:ibp-Vlasov-gen}
a(t)^{-c\sqrt{\epsilon}}\,\E^{(L-1)}_{1,K}(f,t)\lesssim\E^{(L)}_{1,K}(f,t)+a(t)^{-2c\sqrt{\epsilon}}\,\change{\E^{(L-2)}_{1,\min\{K,L-2\}}(f,t)}
\end{equation}
Moreover, given $\omega>0$, one has
\begin{align*}
\numberthis\label{eq:ibp-Vlasov-1-1}
a(t)^{2\omega-c\sqrt{\epsilon}}\E^{(1)}_{1,1}(f,t)\lesssim a(t)^{3\omega}&\,\E^{(2)}_{1,2}(f,t)+a(t)^{\omega-2c\sqrt{\epsilon}}\,\E^{(0)}_{1,0}(f,t)\\
&\,+\change{a(t)^{\frac{\omega}2}}\left(a(t)^{\frac{\omega}2}\,\|\Gamma-\Gamhat\|_{H^1_G(M_t)}^2+a(t)^{\frac{\omega}2}\,\|\change{G^{\pm 1}-\gamma^{\pm 1}}\|_{H^2_G(M_t)}^2\right)\,.
\end{align*}
\end{subequations}
\end{lemma}
\begin{proof}
The first bound follows from a straightforward integration by parts, applying Young's inequality and rearranging. For \eqref{eq:ibp-Vlasov-1-1}, we proceed similarly. Here, we expand $f$ into $(f-f_{FLRW})+f_{FLRW}$ and use
\[\nabsak_{\mathrm{hor}}f_{FLRW}=v\ast v\ast \change{\gamma^{\pm 1}}\ast (\Gamma-\Gamhat)\ast\mathcal{F}^\prime(\lvert v\rvert_\gamma^2)\,\]
as well as \eqref{eq:APMom-2+1} and \eqref{eq:APder} to bound the momentum supports of $f$ and $f_{FLRW}$ respectively by $a^{-c\sqrt{\epsilon}}$ up to constant. For details, we refer to the proof of \cite[Lemma 7.3]{FU25}. 
\end{proof}

\section{Energy estimates}\label{sec:en-est-2+1}

\noindent For all of the estimates in this section, take $L\in\{0,2,4,6,8,10\}$ and energies of \enquote{negative order} are understood to vanish.

The core simplification in our argument compared to that in \cite{FU24,FU25} occurs in Section \ref{subsec:en-est-k}, which is the core geometric estimate. To handle high order curvature terms arising from Vlasov matter, we additionally need to prove a new scaled integral energy estimate in Lemma \ref{lem:en-est-R} using the momentum constraint. The remaining parts of the argument are similar to those in \cite[Sections 5-6]{FU24} or \cite[Sections 5-6]{FU25}; we will keep their discussions relatively brief. 

\subsection{Elliptic lapse estimate}

First, we derive elliptic bounds for the lapse. \change{The former is required to estimate high order lapse terms through the evolution equations, while the latter is needed  to obtain sharp bounds on the lapse and improve the bootstrap assumptions.}

\begin{lemma}[Estimates for the lapse]\label{lem:en-lapse} The following energy estimates hold.
\begin{align*}
\numberthis\label{eq:en-lapse1}a(t)^4\,&\,\E^{(L+2)}(N,t)+a(t)^2\,\E^{(L+1)}(N,t)+\E^{(L)}(N,t)\\
\lesssim &\,\epsilon^2\,\E^{(L)}(\hat{k},t)+\E^{(L)}(\phi,t)+a(t)^{2-c\sqrt{\epsilon}}\left(\E^{(\leq L)}_{1,\leq L}(f,t)+\|G^{-1}-\gamma^{-1}\|_{L^2_G(M_t)}^2\right)\\
&\,+\epsilon^2\,a(t)^{-c\sqrt{\epsilon}}\left(\E^{(\leq L-2)}(\hat{k},t)+\E^{(\leq L-2)}(\phi,t)\right)+\epsilon\,a(t)^{-c\sqrt{\epsilon}}\,\E^{(\leq L-3)}(R,t)\,,\\
\numberthis\label{eq:en-lapse2}a(t)^4\,&\,\E^{(L+2)}(N,t)+a(t)^2\,\E^{(L+1)}(N,t)+\E^{(L)}(N,t)\\
\lesssim&\,a(t)^{4-c\sqrt{\epsilon}}\,\E^{(\leq L)}(R,t)+\epsilon\,a(t)^{2-c\sqrt{\epsilon}}\,\E^{(\leq L)}(\phi,t)\\
&\,\change{+a(t)^{2-(L+1)\omega-c\sqrt{\epsilon}}\,a(t)^{(L+1)\omega}\E^{(\leq L)}_{1,\leq L}(f,t)+a(t)^{2-\frac{\omega}2-c\sqrt{\epsilon}}\,a(t)^\frac{\omega}2\|G^{-1}-\gamma^{-1}\|_{L^2_G(M_t)}^2}\,.
\end{align*}
\end{lemma}
\begin{proof}
To obtain these estimates from \eqref{eq:REEqN}, we verify that solutions to the equation $\L\zeta=Z$ (respectively, $\Ltilde\zeta=Z$) satisfy
\[a^2\,\|\Lap_G\zeta\|_{L^2_G}+a\,\|\nabla\zeta\|_{L^2_G}+\|\zeta\|_{L^2_G}\lesssim\|Z\|_{L^2_G}\]
since the estimates then follow as in \cite[Lemma 5.4]{FU25}. To this end, it suffices to check that $h$ (respectively, $\tilde{h}$) is bounded from below by a positive constant and that $\lvert\nabla h\rvert_G\lesssim\sqrt{\epsilon}$ (respectively, $\lvert \nabla\tilde{h}\rvert_G\lesssim\sqrt{\epsilon}$). The former follows from \eqref{eq:APPsi0} and \eqref{eq:APkhat0} (respectively, \eqref{eq:APder}) as well as \eqref{eq:APVlasovder}, the latter from \eqref{eq:APkin} (respectively, \eqref{eq:APder}) and \eqref{eq:APVlasovder}.
\end{proof}

\subsection{Estimates for the second fundamental form}\label{subsec:en-est-k}

Using that the curvature tensor is pure trace, we obtain a simple energy estimate for the second fundamental form.

\begin{lemma}[Energy estimate for second fundamental form]\label{lem:en-k} 
\begin{subequations}
\begin{align*}
\numberthis\label{eq:en-k}\E^{(L)}(\hat{k},t)\lesssim&\,\epsilon^4+\int_t^{t_0}\change{\left(\epsilon^\frac18\,a(s)^{-2}+a(s)^{-1-c\sqrt{\epsilon}}\right)}\,\E^{(L)}(\hat{k},s)+\epsilon^{-\frac18}\,a(s)^{-2}\E^{(L)}(\phi,s)\,ds\\
&\,+\int_t^{t_0}\epsilon^\frac38\,a(s)^{-2}\cdot\epsilon^\frac12\E^{(L-2)}(R,s)+\change{a(s)^{-1-(L+1)\omega-c\sqrt{\epsilon}}\,\E^{(\leq L)}_{total,Vl}(s)}\,ds\\
&\,+\int_t^{t_0}\epsilon\,a(s)^{-2-c\sqrt{\epsilon}}\,\E^{(\leq L-2)}(\hat{k},s)+\epsilon\,a(s)^{-2-c\sqrt{\epsilon}}\,\E^{(\leq L-2)}(\phi,s)\,ds\\
&\,+\int_t^{t_0}\epsilon^\frac78\,a(s)^{-2-c\sqrt{\epsilon}}\,\E^{(\leq L-4)}(R,s)+\change{a(s)^{-1-\frac{\omega}2-c\sqrt{\epsilon}}\,a(s)^\frac{\omega}2\|G^{-1}-\gamma^{-1}\|_{L^2_G(M_s)}^2}\,ds\,,
\end{align*}
\change{\begin{align*}
\numberthis\label{eq:en-k-scaled}a(t)^\frac{\omega}4\E^{(L)}(\hat{k},\changereport{t})&\,\lesssim\epsilon^4+\int_t^{t_0}\left(\epsilon^\frac18\,a(s)^{-2}+a(s)^{-2+\frac{\omega}8}+a(s)^{-1-c\sqrt{\epsilon}}\right)\,a(s)^{\frac{\omega}4}\E^{(L)}(\hat{k},s)\,ds\\
&\,+\int_t^{t_0}a(s)^{-2+\frac{\omega}8}\E^{(L)}(\phi,s)+{\epsilon}^\frac12 a(s)^{-2+\frac{\omega}8}\,\cdot{\epsilon}^\frac12\E^{(L-2)}(\Ric,s)\,ds\\
&\,+\int_t^{t_0}a(s)^{-1-(L+1)\omega-c\sqrt{\epsilon}}\E^{(\leq L)}_{total,Vl}(s)+a(s)^{-1-\frac{\omega}2-c\sqrt{\epsilon}}\,\cdot a(s)^{\frac{\omega}2}\|G^{-1}-\gamma^{-1}\|_{L^2_G(M_s)}^2\,ds\\
&\,+\int_t^{t_0}\epsilon a(s)^{-2-c\sqrt{\epsilon}}\cdot a(s)^{\frac{\omega}4}\E^{(\leq L-2)}(\hat{k},s)+\epsilon a(s)^{-2-c\sqrt{\epsilon}+\frac{\omega}8}\E^{(\leq L-2)}(\phi,s)\,ds\\
&\,+\int_t^{t_0}\epsilon^\frac32 a(s)^{-2-c\sqrt{\epsilon}+\frac{\omega}8}\,\cdot\epsilon^\frac12\E^{(\leq L-4)}(R,s)\,ds\,.
\end{align*}}
\end{subequations}
\end{lemma}
\begin{proof}
\change{We first prove \eqref{eq:en-k} for $L\geq 4$ and sketch the necessary adaptions for \eqref{eq:en-k-scaled}, the estimates for $L\leq 2$ are proven entirely analogously. }Commuting \eqref{eq:REEqk} with $\Lap_G^\frac{L}2$, we obtain
\begin{subequations}\label{eq:REEq-khat-comm}
\begin{align}
\label{eq:REEq-khat-comm-1}\del_t\Lap_G^\frac{L}2{(\hat{k}^{\sharp})^i}_j=&\,\tau\left({(\hat{k}^{\sharp})^i}_j\,\Lap_G^\frac{L}2N+N\,\Lap_G^\frac{L}2{(\hat{k}^{\sharp})^i}_j\right)-\nabla^{\sharp i}\nabla_j\Lap_G^\frac{L}2N\\
\label{eq:REEq-khat-comm-2}&\,-16\pi\,{\left(\nabla\Lap_G^\frac{L}2\phi\otimes_{\text{symm}}\nabla\phi\right)^{\sharp i}}_{j}-8\pi\,a^{-1}\,\Lap_G^{\frac{L}2}{(S^{Vl,\parallel})^{\sharp i}}_j\\
\label{eq:REEq-khat-comm-3}&\,+\I^i_j\,\Lap_G^\frac{L}2\mathfrak{h}+[\del_t,\Lap_G^\frac{L}2]{(\hat{k}^{\sharp})^ i}_j+{(\mathfrak{K}_{\text{junk}})^i}_j\,,
\end{align}
\end{subequations}
where $\mathfrak{h}$ is a scalar function depending on $G,\,R[G],\,N$ and $f$ as well as reference quantities, and $\mathfrak{K}_\text{junk}$ collects lower order error terms arising from the product rule.

One has
\[\del_t\E^{(L)}(\hat{k},\cdot)=\int_M N\,\tau\,\lvert \Lap_G^\frac{L}2\hat{k}\rvert_G^2+2\,\langle\del_t\Lap_G^\frac{L}2\hat{k}^\sharp,{\Lap_G^\frac{L}2\hat{k}^{\sharp}}\rangle_G\,\vol{G}\,.\]

Consequently, the pure trace term in \eqref{eq:REEq-khat-comm-3} cancels when inserting the evolution equation in the second term. Note that this includes high order curvature terms in \eqref{eq:REEq-khat-comm}, i.e., those that don't arise from commutator errors. Abbreviating
\[\left(\del_t\Lap_G^\frac{L}2\hat{k}^\sharp\right)^\parallel=\del_t\Lap_G^\frac{L}2\hat{k}^\sharp-\I\,\Lap_G^\frac{L}2\mathfrak{h}\,,\]
the bootstrap assumption \eqref{eq:bs-ass-HC} implies
\begin{equation}\label{eq:khat-en-est-core}
-\del_t\E^{(L)}(\hat{k},t)\lesssim \epsilon\,a(t)^{-1-c\eta}\,\E^{(L)}(\hat{k},t)+\left\|\left(\del_t\Lap_G^\frac{L}2\hat{k}^\sharp\right)^\parallel\right\|_{L^2_G(M_t)}\sqrt{\E^{(L)}(\hat{k},t)}\,.
\end{equation}

Next, we estimate $\left(\del_t\Lap_G^\frac{L}2\hat{k}^\sharp\right)^\parallel$ by going through the remaining terms in \eqref{eq:REEq-khat-comm}. We apply \change{the bootstrap assumption \eqref{eq:bs-ass-HC} as well as \eqref{eq:APkhat0} }to bound the right hand side of \eqref{eq:REEq-khat-comm-1} in $L^2_G$ by
\begin{align}\label{eq:khat-comm-lapse-energ}
\lesssim&\,\epsilon\,a(t)^{-1-c\,\eta} \sqrt{\E^{(L)}(\hat{k},t)}+a(t)^{-2}\left(a(t)^{-2}\,\sqrt{a(t)^4\,\E^{(L+2)}(N,t)}+\epsilon\sqrt{\E^{(L)}(N,t)}\right)\,.
\end{align}
Regarding \eqref{eq:REEq-khat-comm-2}, we use \eqref{eq:APder} and \eqref{eq:near-coerc-Vl2} to obtain the bound
\begin{align*}
\lesssim&\,\sqrt{\epsilon}\,a(t)^{-1-c\sqrt{\epsilon}}\,\sqrt{\E^{(L)}(\phi,t)}+a(t)^{-1-c\,\sqrt{\epsilon}}\changereport{\left(\sqrt{\E^{(\leq L)}_{1,\leq L}(f,t)}+\|G^{-1}-\gamma^{-1}\|_{L^2_G(M)}\right)}\,.
\end{align*}
In combination, we obtain
\begin{align*}
\left\|\left(\del_t\Lap_G^\frac{L}2\hat{k}^\sharp\right)^\parallel\right\|_{L^2_G(M_t)}\lesssim&\,a(t)^{-2}\left(\sqrt{a(t)^4\,\E^{(L+2)}(N,t)}+\sqrt{\E^{(L)}(N,t)}\right)+\epsilon\,a(t)^{-1-c\eta}\sqrt{\E^{(L)}(\hat{k},t)}\\
&\,+\sqrt{\epsilon}\,a(t)^{-1-c\sqrt{\epsilon}}\left(\sqrt{\E^{(L)}(\phi,t)}+\sqrt{\E^{(L)}_{1,L}(f,t)}+\|\change{G^{-1}-\gamma^{-1}}\|_{L^2_G(M_t)}\right)\\
&\,+\|[\del_t,\Lap_G^\frac{L}2]\hat{k}\|_{L^2_G(M_t)}+\|\mathfrak{K}_\text{junk}\|_{L^2_G(M_t)}\,.
\end{align*}
Applying the elliptic estimate \eqref{eq:en-lapse1} for the lapse terms arising from \eqref{eq:khat-comm-lapse-energ}, the following bound thus follows from \eqref{eq:khat-en-est-core}.
\begin{subequations}
\begin{align}
-\del_t\E^{(L)}(\hat{k},t)\lesssim&\,\epsilon\,a(t)^{-1-c\,\eta}\E^{(L)}(\hat{k},t)\\
\label{eq:khat-diff-en-est-leading}&\,+\sqrt{\E^{(L)}(\hat{k},t)}\,\Bigr\{a(t)^{-2}\left(\epsilon\,\sqrt{\E^{(L)}(\hat{k},t)}+\sqrt{\E^{(L)}(\phi,t)}\right)\\
&\,\qquad+a(t)^{-1-c\sqrt{\epsilon}}\,\left(\sqrt{\E^{(\leq L)}_{1,\leq L}(f,t)}+\changereport{\|G^{-1}-\gamma^{-1}\|_{L^2_G(M_t)}}\right)\\
&\,\qquad+\epsilon\,a(t)^{-2-c\,\sqrt{\epsilon}}\left(\sqrt{\E^{(\leq L-2)}(\hat{k},t)}+\sqrt{\E^{(\leq L-2)}(\phi,t)}\right)\\
\label{eq:khat-diff-en-est-error}&\,\qquad+\sqrt{\epsilon}\,a(t)^{-2-c\sqrt{\epsilon}}\sqrt{\E^{(\leq L-3)}(R,t)}\\
&\,\qquad+\|[\del_t,\Lap_G^\frac{L}2]\hat{k}\|_{L^2_G(M_t)}+\|\mathfrak{K}_\text{junk}\|_{L^2_G(M_t)}\Bigr\}\,.
\end{align}
\end{subequations}
The curvature terms in the penultimate line arise from estimating commutator errors when proving \eqref{eq:en-lapse1}. \\
A naive estimate of \eqref{eq:khat-diff-en-est-leading} would lead to the borderline, unscaled term $a(t)^{-2}\,\E^{(L)}(\hat{k},t)$, which would prevent us from obtaining bootstrap improvements down the line. Thus, we apply Young's inequality as follows:
\[a^{-2}\,\sqrt{\E^{(L)}(\hat{k},\cdot)}\,\sqrt{\E^{(L)}(\phi,\cdot)}\lesssim \epsilon^\frac18\,a^{-2}\,\E^{(L)}(\hat{k},\cdot)+\epsilon^{-\frac18}\,a^{-2}\,\E^{(L)}(\phi,\cdot)\,.\]
Applying Young's inequality as well as \eqref{eq:ibp-trick-2+1} in a similar manner for the remaining terms, the estimate below follows.
\begin{subequations}\label{eq:diff-en-khat}
\begin{align}
\label{eq:diff-en-khat-1}-\del_t\E^{(L)}(\hat{k},t)\lesssim&\,\left(\epsilon^\frac18\,a(t)^{-2}+a(t)^{-1-c\sqrt{\epsilon}}\right)\E^{(L)}(\hat{k},t)+\epsilon^{-\frac18}\,a(t)^{-2}\,\E^{(L)}(\phi,t)\\
\label{eq:diff-en-khat-2}&\,+\change{a(t)^{-1-(L+1)\omega-c\sqrt{\epsilon}}\,a(t)^{(L+1)\omega}{\E^{(L)}_{1,L}(f,t)}+a(t)^{-1-\frac{\omega}2-c\sqrt{\epsilon}}\,a(t)^\frac\omega2\changereport{\|G^{-1}-\gamma^{-1}\|_{L^2_G(M_t)}^2}}\\
\label{eq:diff-en-khat-3}&\,\change{+\epsilon^\frac78\,a(t)^{-2}\,\E^{(\leq L-2)}(R,t)+\epsilon^\frac{15}8\,a(t)^{-2-c\sqrt{\epsilon}}\left(\E^{(\leq L-2)}(\hat{k},t)+\E^{(\leq L-2)}(\phi,t)\right)}\\
\label{eq:diff-en-khat-4}&\,\change{+\epsilon^\frac78\,a(t)^{-2-c\sqrt{\epsilon}}\,\E^{(\leq L-4)}(R,t)+\left(\|[\del_t,\Lap_G^\frac{L}2]\hat{k}\|_{L^2_G(M_t)}+\|\mathfrak{K}_\text{junk}\|_{L^2_G(M_t)}\right)\sqrt{\E^{(L)}(\hat{k},t)}}\,.
\end{align}
\end{subequations}
In particular, the highest order curvature term arises from applying \eqref{eq:ibp-trick-2+1} to the curvature term in \eqref{eq:khat-diff-en-est-error}. It remains to check that the \change{final terms }in \eqref{eq:diff-en-khat-4} are error terms dominated by what is already present in \eqref{eq:diff-en-khat-1}-\eqref{eq:diff-en-khat-3} \change{as well as the first term in \eqref{eq:diff-en-khat-4}}.

After some straightforward computations and applications of Lemma \ref{lem:near-coerc} to estimate Sobolev norms by energies, one sees that the commutator term satisfies
\begin{align*}
\|[\del_t,\Lap_G^\frac{L}2]\hat{k}\|_{L^2_G(M_t)}\lesssim&\ \epsilon\,a(t)^{-2}\,\sqrt{\E^{(\leq L)}(\hat{k},\cdot)}+\epsilon\,a(t)^{-2-c\sqrt{\epsilon}}\,\sqrt{\E^{(\leq L-2)}(\hat{k},\cdot)}\\
&\,+\epsilon\,a(t)^{-2-c\sqrt{\epsilon}}\,\sqrt{\E^{(\leq L-2)}(N,\cdot)}+\epsilon^2\,a(t)^{-2-c\sqrt{\epsilon}}\,\sqrt{\E^{(\leq L-3)}(R,\cdot)}\,.
\end{align*}
Moreover, $\mathfrak{K}_\text{junk}$ can be bounded as follows.
\begin{align*}
\|\mathfrak{K}_\text{junk}\|_{L^2_G}\lesssim&\ \epsilon\,a(t)^{-1-c\eta}\,\sqrt{\E^{(\leq L-1)}(\hat{k},\cdot)}+\epsilon\,a(t)^{-2-c\sqrt{\epsilon}}\,\sqrt{\E^{(\leq L-1)}(N,\cdot)}\\
&\,+\sqrt{\epsilon}\,a(t)^{-1-c\sqrt{\epsilon}}\,\sqrt{\E^{(\leq L-1)}(\phi,\cdot)}+\epsilon\,a(t)^{-1-c\eta}\,\sqrt{\E^{(\leq L-3)}(R,\cdot)}\,.
\end{align*}

After applying the same steps as above to these error upper bounds, one sees that \change{the final term in \eqref{eq:diff-en-khat-4} }indeed only contributes negligible error terms to the right hand side of \eqref{eq:diff-en-khat}. The statement then follows upon integrating \eqref{eq:diff-en-khat} over $[t,t_0]$ and rearranging.

\change{For \eqref{eq:en-k-scaled}, one similarly computes the analogue of \eqref{eq:diff-en-khat}, dropping $-(\del_ta^\frac{\omega}4)\E^{(L)}(\hat{k},\cdot)$ in an upper estimate, and immediately distribute weights as needed for borderline energy terms, i.e., for shear, scalar field and curvature energies
\begin{align*}
-\del_t\left(a^{\frac{\omega}4}\E^{(L)}(\hat{k},\cdot)\right)\lesssim&\,\epsilon a(t)^{-1-c\eta}\,\E^{(L)}(\hat{k},\cdot)\\
&\,+\sqrt{a^\frac{\omega}4\E^{(L)}(\hat{k},\cdot)}\Bigr\{\epsilon a^{-2}\sqrt{a^{\frac{\omega}4}\E^{(L)}(\hat{k},\cdot)}+a^{-2+\frac{\omega}8}\sqrt{\E^{(L)}(\phi,\cdot)}\\
&\,\qquad+a^{-1-c\sqrt{\epsilon}+\frac{\omega}8}\left(\sqrt{\E^{(\leq L)}_{1,\leq L}(f,\cdot)}+\|G-\gamma\|_{L^2_G}\right)\\
&\,\qquad+\epsilon a^{-2-c\sqrt{\epsilon}}\sqrt{a^\frac{\omega}4\E^{(\leq L-2)}(\hat{k},\cdot)}+\epsilon a^{-2-c\sqrt{\epsilon}+\frac{\omega}8}\sqrt{\E^{(\leq L-2)}(\phi,\cdot)}\\
&\,\qquad+a^{-2-c\sqrt{\epsilon}+\frac{\omega}8}\,\sqrt{\epsilon \E^{(\leq L-3)}(R,\cdot)}\Bigr\}\\
&\,+\sqrt{a^\frac{\omega}4\E^{(L)}(\hat{k},\cdot)}\cdot a^\frac{\omega}8\left(\|[\del_t,\Lap_G^\frac{L}2]\hat{k}\|_{L^2_G}+\|\mathfrak{K}_\text{junk}\|_{L^2_G}\right)\,.
\end{align*}
Again, one observes that terms introduced by the last line are strictly dominated by the remaining right hand side up to a constant, and applies the Young inequality to the remaining lines. The scaled bound \eqref{eq:en-k-scaled} then follows by integrating and inserting the initial data assumption.}
\end{proof}

To overcome the technical difficulty that the zero order Vlasov energy estimate will involve first order terms in the metric, we also require the following elliptic estimate for the shear.

\begin{lemma}[Scaled first order energy estimate for the second fundamental form]\label{lem:en-est-khat-scaled}
\begin{equation}\label{eq:en-est-khat-scaled}
a^2\,\E^{(1)}(\hat{k},\cdot)\lesssim \E^{(0)}(\phi,\cdot)+a^{2-c\sqrt{\epsilon}}\,\E^{(0)}(\hat{k},\cdot)+a^{2-c\sqrt{\epsilon}}\,\E^{(0)}_{1,0}(f,\cdot)\,.
\end{equation}
\end{lemma}
\begin{proof}
Let $u$ be an arbitrary tracefree $M_t$-tangent $(0,2)$-tensor and define the exterior covariant derivative
\[(d^\nabla u)_{ijl}=\nabla_lu_{ij}-\nabla_ju_{il}\]
which satisfies
\begin{align*}
\int_M\lvert d^\nabla u\rvert_g^2\,\vol{g}=&\,2\int_M\lvert \nabla u\rvert_g^2-\nabla_ju_{il}\,\nabla^{l}u^{ij}\,\vol{g}\,.
\end{align*}

By integration by parts and using $\Riem[g]_{ijlm}=\frac{R[g]}2\,(g_{il}\,g_{jm}-g_{im}\,g_{jl})$ as well as $\text{tr}_gu=0$, one obtains
\begin{align*}
\int_M\nabla_ju_{il}\,\nabla^{l}u^{ij}\,\vol{g}=&\,\int_M\left[\lvert \div_gu\rvert_g^2+\frac{R[g]}2\left(g_{jm}\,g_{il}-g_{ji}\,g_{lm}\right)u^{il}\,u^{mj}+\frac{R[g]}2\,\I^l_j\,u^{ij}\,u_{il}\right]\,\vol{g}\\
=&\,\int_M\left[\lvert\div_gu\rvert_g^2+R[g]\,\lvert u\rvert_g^2\right]\,\vol{g}\,.
\end{align*}
After rearranging and rescaling, this becomes
\begin{equation}\label{eq:div-curl-repl}
\int_M \lvert\nabla u\rvert_G^2\,\vol{G}=\int_M\left[\frac12\,\lvert d^\nabla u\rvert_G^2+\lvert \div_Gu\rvert_G^2+\frac34\,R[G]\,\lvert u\rvert_G^2\right]\,\vol{G}\,.
\end{equation}
Now considering $u=\hat{k}$ explicitly, note that the Codazzi equation for the \change{$(2+1)$-}dimensional spacetime metric $\g$ implies
\begin{align*}
d^\nabla\hat{k}_{ijl}=&\,\Riem[\g]_{0ijl}=\Ric[\g]_{0j}\,g_{il}-\Ric[\g]_{0l}\,g_{ij}
\end{align*}
since the shift vector vanishes in our transported coordinates. Now applying \eqref{eq:EEq1}, this yields
\begin{align*}
d^\nabla\hat{k}_{ijl}&\,=8\pi\,(\Psi+C)\,\left[G_{il}\,\nabla_j\phi-G_{ij}\,\nabla_l\phi\right]\changereport{-8\pi\left(G_{il}\,\j^{Vl}_j-G_{ij}\,\j^{Vl}_l\right)}
\end{align*}
Thus, also applying \eqref{eq:REEqMom-2+1} to compute $\div_G\hat{k}$, \eqref{eq:div-curl-repl} becomes
\begin{align*}
\E^{(1)}(\hat{k},\cdot)\lesssim \int_M \left[\lvert \Psi+C\rvert^2\,\lvert\nabla\phi\rvert_G^2+\lvert \j^{Vl}\rvert_G^2+\lvert R[G]\rvert\,\lvert \hat{k}\rvert_G^2\right]\,\vol{G}
\end{align*}
The a priori estimates \eqref{eq:APPsi0}, \eqref{eq:near-coerc-Vl3} and \eqref{eq:APder} then imply \eqref{eq:en-est-khat-scaled}.
\end{proof}

\subsection{Estimates for metric variables}\label{subsec:en-est-metric}

In this section, we collect various estimates for norms and energies associated with the metric, the tensor $\Gamma-\Gamhat$ and the scalar curvature. In particular, we need to derive integral estimates for $\Gamma-\Gamhat$ at low orders to deal with error terms arising from the Vlasov reference distribution \change{as well as }high order energy estimates for the curvature to control commutator errors from the Vlasov equation. In both cases, the key ingredient is the Codazzi equation -- for the former in the form of the elliptic estimate \eqref{eq:en-est-khat-scaled} and for the latter by directly applying the momentum constraint \eqref{eq:REEqMom-2+1}.

\begin{lemma}[Norm estimates for metric quantities] Let $I$ be an integer with $0\leq I\leq 10$. Then, the following estimates hold:
\begin{subequations}
\begin{align*}
\numberthis\label{eq:G-norm-int-est} a(t)^\frac{\omega}2\|\change{G^{\pm 1}-\gamma^{\pm 1}}&\,\|_{H^I_G(M_t)}^2\lesssim\epsilon^4+\int_t^{t_0}\left(a(s)^{-2+\change{\frac{\omega}8}-c\sqrt{\epsilon}}+a(s)^{-1-c\eta}\right) a(s)^\frac{\omega}2\,\|\change{G^{\pm 1}-\gamma^{\pm 1}}\|_{H^I_G(M_s)}^2\,ds\\
&\,+\int_t^{t_0}a(s)^{-2+\change{\frac{\omega}8}}\left(a(s)^\frac{\omega}2\E^{(\leq I)}(\hat{k},s)+\E^{(\leq I)}(\phi,s)+\change{\epsilon^\frac12\cdot\epsilon^\frac12}\E^{(\leq I-2)}(R,s)\right)\,ds\\
&\,+\int_t^{t_0}a(s)^{-1-c\sqrt{\epsilon}-\omega(I+\frac12)}\cdot a(s)^{(I+1)\omega}\,\E^{(\leq I)}_{1,\leq I}(f,s)\,ds\\
\numberthis\label{eq:G-norm-sup-est}\|\change{G^{\pm 1}-\gamma^{\pm 1}}&\,\|_{H^I_G(M_t)}^2\lesssim a(t)^{-c\epsilon^\frac18}\left(\epsilon^4+\epsilon^{-\frac14}\sup_{s\in(t,t_0)}\left(\|N\|_{H^I_G(M_s)}^2+\|\hat{k}\|_{H^I_G(M_s)}^2\right)\right)
\end{align*}
Additionally, one has the following low order integral estimates:
\begin{align*}
\numberthis\label{eq:en-est-Chr0}a(t)^{2+\frac{\omega}2}\|\Gamma-\Gamhat\|_{L^2_G(M_t)}^2&\,\lesssim\epsilon^4+\int_t^{t_0}\change{\left(a(s)^{-2+{\frac{\omega}8}}+a(s)^{-1-c\sqrt{\epsilon}-\omega}\right)}\cdot a(s)^{2+\frac{\omega}2}\|\Gamma-\Gamhat\|_{L^2_G(M_t)}^2\,ds\\
&\,+\int_t^{t_0}\change{a(s)^{-2+\frac{\omega}8-c\sqrt{\epsilon}}}\,\left(\E^{(0)}(\phi,s)+a(s)^\frac{\omega}4\E^{(0)}(\hat{k},s)\right)\,ds\\
&\,+\int_t^{t_0}a(s)^{-1-c\sqrt{\epsilon}-\frac{\omega}2}\,\left(a(s)^{\omega}\E^{(0)}_{1,0}(f,s)+a(s)^\frac{\omega}2\|\change{G^{-1}-\gamma^{-1}}\|_{L^2_G(M_s)}^2\right)\,ds\\
\numberthis\label{eq:en-est-Chr1}a(t)^{\frac{\omega}2}\,\|\Gamma-\Gamhat\|^2_{{H}^1_G(M_t)}&\,\lesssim\epsilon^4+\int_t^{t_0}\left(a(s)^{-2+\change{\frac{\omega}8}}+a(s)^{-1-c\sqrt{\epsilon}\change{-}\frac{\omega}2}\right)\,\change{a(s)^{\frac{\omega}2}}\,\|\Gamma-\Gamhat\|^2_{{H}^1_G(M_s)}ds\\
&\,+\int_t^{t_0}a(s)^{-2+\change{\frac{\omega}8-c\sqrt{\epsilon}}}\left(\change{a(s)^\frac{\omega}4}\E^{(\leq 2)}(\hat{k},s)+\E^{(\leq 2)}(\phi,s)\right)\,ds\\
&\,+\int_t^{t_0}a(s)^{-1-c\sqrt{\epsilon}\change{-\omega}}\left(a(s)^{2\omega}\,\E^{(\leq 1)}_{1,\leq 1}(f,s)+a(s)^{\frac{\omega}2}\,\|\change{G^{-1}-\gamma^{-1}}\|_{L^2_G(M_s)}^2\right)\,ds
\end{align*}
\end{subequations}
\end{lemma}
\begin{proof}
\eqref{eq:G-norm-int-est} and \eqref{eq:G-norm-sup-est} follow as in \cite[Lemma 6.14]{FU24} by commuting \eqref{eq:REEqG-2+1} with $\nabla^I$ and applying the lapse estimate \eqref{eq:en-lapse1} as well as the near-coercivity estimates from Lemma \ref{lem:near-coerc}. 
Regarding \eqref{eq:en-est-Chr0}, the evolution equation \eqref{eq:REEqChr-2+1} implies, using the bootstrap assumption \eqref{eq:bs-ass-HC} for the lapse,
\change{\begin{align*}
-\del_t\left(a^{2+\change{\frac{\omega}2}}\|\Gamma-\Gamhat\|_{L^2_G(M_\cdot)}^2\right)\lesssim&\,\epsilon a^{-c\eta}\cdot a^{2+\change{\frac{\omega}2}}\|\Gamma-\Gamhat\|_{L^2_G(M_\cdot)}^2\\
&\,\underbrace{+a^{2+\change{\frac{\omega}2}}\left(a^{-2}\sqrt{\E^{(\leq 1)}(\hat{k},\cdot)}+a^{-2}\sqrt{\E^{(1)}(N,\cdot)}\right)\|\Gamma-\Gamhat\|_{L^2_G(M_\cdot)}}_{(\ast)}
\end{align*}
Regarding $(\ast)$, one can apply shear estimate in Lemma \ref{lem:en-est-khat-scaled} and the lapse estimate \eqref{eq:en-lapse1} for $L=0$ to obtain the following estimate
\begin{align*}
(\ast)=&\,a^{-2+\frac\omega8}\left(\sqrt{a^{2+\frac{\omega}4}\E^{(\leq 1)}(\hat{k},\cdot)}+\sqrt{a^{2+\frac{\omega}4}\E^{(\leq 1)}(N,\cdot)}\right)\sqrt{a^{2+\frac{\omega}2}\|\Gamma-\Gamhat\|_{L^2_G(M_\cdot)}^2}\\
\lesssim&\,\left(a^{-2+\frac{\omega}4}\sqrt{\E^{(0)}(\phi,\cdot)}+\left(\epsilon a^{-2+\frac{\omega}8}+a^{-1-c\sqrt{\epsilon}}\right)\sqrt{a(s)^{\frac{\omega}4}\E^{(0)}(\hat{k},\cdot)}\right.\\
&\,\left.+a^{-1-c\sqrt{\epsilon}-\frac{\omega}2}\sqrt{a^\omega\E^{(0)}_{1,0}(f,\cdot)}+a^{-1-c\sqrt{\epsilon}+\frac{\omega}4}\|G^{-1}-\gamma^{-1}\|_{L^2_G(M_\cdot)}\right)\sqrt{a^{2+\frac{\omega}{2}}\|\Gamma-\Gamhat\|_{L^2_G(M_\cdot)}^2}
\end{align*}
\eqref{eq:en-est-Chr0} then follows by applying the Young inequality.}\\
\eqref{eq:en-est-Chr1} is proven identically, except that one applies the Young inequality to estimate the shear term directly and inserts \eqref{eq:en-lapse1} at order $L=2$. 
\end{proof}

To control high order curvature terms arising from the upcoming Vlasov energy estimates, it will be necessary to control curvature terms at order $L-1$ to control the other solution variables at order $L$. Since we only need to control the scalar curvature, the momentum constraint admits a convenient integral estimate that immediately regains this potential loss in regularity:

\begin{lemma}[Energy estimate for the curvature]\label{lem:en-est-R}

\change{\begin{subequations}
\begin{align*}
\numberthis\label{eq:en-est-Ric-2+1}\E^{(L-2)}(R,t)\lesssim\epsilon^4&\,+\int_t^{t_0}\left(\epsilon^\frac18\,a(s)^{-2}+a(s)^{-1-c\sqrt{\epsilon}}\right)\E^{(L-2)}(R,s)\,ds\\
+\int_t^{t_0}\Bigr\{&\,\epsilon^{-\frac18}\,a(s)^{-2}\left(\E^{(L)}(\phi,s)+\E^{(L)}(\hat{k},s)\right)\\
&\,+a(s)^{-1-c\sqrt{\epsilon}-(L+1)\omega}\left(a(s)^{(L+1)\omega}\,\E^{(\leq L)}_{1,\leq L}(f,s)+a(s)^\frac{\omega}2\,\|G^{-1}-\gamma^{-1}\|_{L^2_G(M_s)}^2\right)\\
&\,+\epsilon^{-\frac18}\,a(s)^{-2-c\sqrt{\epsilon}}\left(\E^{(\leq L-2)}(\phi,s)+\E^{(\leq L-2)}(\hat{k},s)\right)\\
&\,+\epsilon^\frac78\,a(s)^{-2-c\sqrt{\epsilon}}\,\E^{(\leq L-4)}(R,s)\Bigr\}\,ds\,,\\
\numberthis\label{eq:en-est-Ric-top-2+1}
a(t)^{2}\E^{(L-1)}(R,t)&\,\lesssim\epsilon^4+\int_t^{t_0}\left(\epsilon^\frac18\,a(s)^{-2}+a(s)^{-1-c\eta}\right)\, a(s)^2\,\E^{(L-1)}(R,s)\,ds\\
+\int_t^{t_0}&\Bigr\{\left(\epsilon^\frac{15}8\,a(s)^{-2}+a(s)^{-1-c\sqrt{\epsilon}}\right)\E^{(L)}(\hat{k},s)\\
&\,+\left(\epsilon^{-\frac18}\,a(s)^{-2}+a(s)^{-1-c\sqrt{\epsilon}}\right)\E^{(L)}(\phi,s)\\
&\,+a(s)^{-1-c\sqrt{\epsilon}-(L+1)\omega}\left(a(s)^{(L+1)\omega}\,\E^{(\leq L)}_{1,\leq L}(f,s)+a(s)^\frac{\omega}2\,\|G^{-1}-\gamma^{-1}\|_{L^2_G(M_s)}^2\right)\\
&\,+\epsilon\,a(s)^{-2}\,\E^{(L-2)}(R,s)+\epsilon\,a(s)^{-2-c\sqrt{\epsilon}}\,\E^{(\leq L-4)}(R,s)\\
&\,\left.+\left(\epsilon^\frac{15}8\,a(s)^{-2-c\sqrt{\epsilon}}+a(s)^{-1-c\sqrt{\epsilon}}\right)\left(\E^{(\leq L-2)}(\phi,s)+\E^{(\leq L-2)}(\hat{k},s)\right)\right\}\,ds\\
\numberthis\label{eq:en-est-Ric-scaled-2+1} a(t)^\frac{\omega}2\E^{(L-2)}&(R,t)\lesssim\epsilon^4+\int_t^{t_0}\left(\epsilon^\frac18 a(s)^{-2}+a(s)^{-2+\frac{\omega}8}+a(s)^{-1-c\sqrt{\epsilon}}\right)\E^{(L-2)}(R,s)\,ds\\
&\,+\int_t^{t_0}\Bigr\{a(s)^{-2+\frac{\omega}8}\left(\E^{(L)}(\phi,s)+\E^{(L)}(\hat{k},s)\right)\\
&\,\qquad +a(s)^{-1-c\sqrt{\epsilon}-(L+1)\omega}\left(\E^{(\leq L)}_{total,Vl}(s)+a(s)^\frac{\omega}2\|G^{-1}-\gamma^{-1}\|_{L^2_G(M_s)}^2\right)\\
&\,\qquad+a(s)^{-2-c\sqrt{\epsilon}+\frac{\omega}4}\left(\E^{(\leq L-2)}(\phi,s)+a(s)^\frac{\omega}4\E^{(\leq L-2)}(\hat{k},s)\right)\\
&\,\qquad+\epsilon^\frac78 a(s)^{-2-c\sqrt{\epsilon}}\E^{(\leq L-4)}(R,s)\Bigr\}\,ds\\
\numberthis\label{eq:en-est-Ric-scaled-top-2+1}a(t)^{2+\frac{\omega}4}\E^{(L-1)}&(R,t)\lesssim\epsilon^4+\int_t^{t_0}\left(\epsilon^\frac18 a(s)^{-2}+a(s)^{-2+\frac{\omega}8}+a(s)^{-1-c\eta}\right)\,a(s)^{2+\frac{\omega}4}\E^{(L-1)}(R,s)\,ds\\
&\,+\int_t^{t_0}\Bigr\{\left(\epsilon a(s)^{-2}+a(s)^{-1-c\sqrt{\epsilon}}\right)\cdot a(s)^\frac{\omega}4\E^{(L)}(\hat{k},s)\\
&\,\qquad+\left(a(s)^{-2+\frac{\omega}8}+a(s)^{-1-c\sqrt{\epsilon}}\right)\E^{(L)}(\phi,s)\\
&\,\qquad+a(s)^{-1-c\sqrt{\epsilon}+(L+1)\omega}\left(\E^{(\leq L)}_{total,Vl}(s)+a(s)^\frac{\omega}2\|G^{-1}-\gamma^{-1}\|_{L^2_G}^2\right)\\
&\,\qquad+\epsilon a(s)^{-2}\cdot a(s)^{\frac{\omega}4}\E^{(L-2)}(R,s)+\epsilon a(s)^{-2-c\sqrt{\epsilon}}\cdot a(s)^\frac{\omega}4\E^{(\leq L-4)}(R,s)\\
&\,\qquad+\left(\epsilon^2a(s)^{-2-c\sqrt{\epsilon}}+a(s)^{-1-c\sqrt{\epsilon}}\right)\left(\E^{(\leq L-2)}(\phi,s)+a(s)^{\frac{\omega}4}\E^{(\leq L-2)}(\hat{k},s)\right)\Bigr\}\,ds
\end{align*}
\end{subequations}}
\end{lemma}
\begin{proof}
\change{We focus on proving \eqref{eq:en-est-Ric-scaled-top-2+1} and \eqref{eq:en-est-Ric-scaled-2+1}, since \eqref{eq:en-est-Ric-scaled-2+1} and \eqref{eq:en-est-Ric-2+1} follow along similar, but simpler lines by commuting \eqref{eq:REEqR} with $\Lap_G^\frac{L}2$. }To prove \eqref{eq:en-est-Ric-top-2+1}, we first compute, using bootstrap assumption \eqref{eq:bs-ass-HC}, that
\begin{align*}\numberthis\label{eq:diff-est-R-basic}
-\del_t\left(a^{2}\,\E^{(L-1)}(R,\cdot)\right)\lesssim&\,\epsilon\,a^{-1-c\eta}\,a^{2}\,\E^{(L-1)}(R,\cdot)\\
&\,+\|a\,\del_t\nabla\Lap_G^{\frac{L}2-1} R[G]\|_{L^2_G(M_\cdot)}\,\sqrt{a^{2}\,\E^{(L-1)}(R,\cdot)}\,.
\end{align*}
To obtain a convenient expression for $a\,\del_t\nabla\Lap_G^{\frac{L}2-1}R[G]$, applying \eqref{eq:REEqMom-2+1} to \eqref{eq:REEqR} yields
\begin{align*}
a\,\del_tR[G]=&\,2\,N\,\tau\,a\left(R[G]-2\,\kappa\right)-8\,\kappa\,\dot{a}\,N-2\,a^{-1}\,\langle\hat{k},\nabla^2N\rangle_G\\
&\,+16\pi\,a^{-1}\,\langle\nabla N,\change{(\Psi+C)}\,\nabla\phi+\j^{Vl}\rangle_G\\
&\,+\change{16\pi\,a^{-1}\,(N+1)\left[\,(\Psi+C)\,\Lap_G\phi+\langle\nabla\phi,\nabla\Psi\rangle_G+\div_G\j^{Vl}\right]\,.}
\end{align*}
Commuting the above with $\nabla\Lap_G^{\frac{L}2-1}$, this schematically becomes the following.
\begin{subequations}\label{eq:REEqR-comm}
\begin{align}
\label{eq:REEqR-comm-1}a\,\del_t\nabla&\,\Lap_G^{\frac{L}2-1}R[G]=-4\,N\,\dot{a}\nabla\Lap_G^{\frac{L}2-1}\left(R[G]-2\,\kappa\right)-4\,\dot{a}\,R[G]\,\nabla\Lap_G^{\frac{L}2-1}N\\
\label{eq:REEqR-comm-2}&\,+a^{-1}\left(\hat{k}\ast\nabla^{L+1}N+\nabla\hat{k}\ast\nabla^{L} N+\nabla^2\hat{k}\ast\nabla^{L-1}N\right)\\
\label{eq:REEqR-comm-3}&\,+16\pi\,C\,a^{-1}\,\nabla\Lap_G^{\frac{L}2}\phi+a^{-1}\left(\nabla\Psi \ast\nabla^2\Lap_G^{\frac{L}2-1}\phi+\nabla^2\phi\ast\nabla\Lap_G^{\frac{L}2-1}\Psi+\nabla\phi\ast\nabla^{L}\Psi\right)\\
\label{eq:REEqR-comm-4}&\,{+16\pi\,a^{-1}\,\nabla\Lap_G^\frac{L}2\j^{Vl}}+a\,[\del_t,\nabla\Lap_G^{\frac{L}2-1}]R[G]+a\,\mathfrak{R}_{\text{junk},L-1}\
\end{align}
\end{subequations}
Herein, $\mathfrak{R}_{\text{junk},L-1}$ collects lower order error terms. Using the low order bounds \eqref{eq:APkhat0}, \eqref{eq:APkin} and \eqref{eq:APder} as well as the bootstrap bound \eqref{eq:bs-ass-HC} for the lapse, the $L^2_G$-norm of the right hand side of \eqref{eq:REEqR-comm-1}-\eqref{eq:REEqR-comm-2} is bounded by
\begin{align*}
\lesssim&\,\epsilon\,a^{-1-c\eta}\sqrt{a^2\,\E^{(L-1)}(R,\cdot)}+\epsilon\,a^{-2}\sqrt{a^2\,\E^{(L+1)}(N,\cdot)}\\
&\,+\epsilon\,a^{-1-c\sqrt{\epsilon}}\sqrt{\E^{(L)}(N,\cdot)}+a^{-1-c\sqrt{\epsilon}}\sqrt{\E^{(\leq L-1)}(N,\cdot)}\,.
\end{align*}
Using the low order estimates \eqref{eq:APkin} and \eqref{eq:APder} as well as \eqref{eq:ibp-trick-2+1}, that of \eqref{eq:REEqR-comm-3} is bounded by
\begin{align*}
\lesssim a^{-2}\sqrt{\E^{(L)}(\phi,\cdot)}+\left(\epsilon\,a^{-2-c\sqrt{\epsilon}}+\sqrt{\epsilon}\,a^{-1-c\sqrt{\epsilon}}\right)\sqrt{\E^{(\leq L-2)}(\phi,\cdot)}\,,
\end{align*}
while the Vlasov term in \eqref{eq:REEqR-comm-4} can be estimated directly using \eqref{eq:near-coerc-Vl3}. The remaining error terms in \eqref{eq:REEqR-comm-4} can be bounded as follows:
\begin{align*}
a\left\|[\del_t,\nabla\Lap_G^{\frac{L}2-1}]R[G]\right\|_{L^2_G}\lesssim&\,\epsilon \,a^{-2}\,\sqrt{a^2\,\E^{(L-1)}(R,\cdot)}+\epsilon\,a^{-1-c\sqrt{\epsilon}}\,\sqrt{\E^{(\leq L-2)}(R,\cdot)}\\
&\,+\epsilon\,a^{-1-c\sqrt{\epsilon}}\left(\sqrt{\E^{(\leq L-2)}(\hat{k},\cdot)}+\sqrt{\E^{(\leq L-2)}(N,\cdot)}\right)\\
a\|\mathfrak{R}_{\text{junk},L-1}\|_{L^2_G}\lesssim&\,\epsilon\,a^{-1-c\sqrt{\epsilon}}\left(\sqrt{\E^{(\leq L-1)}(\hat{k},\cdot)}+\sqrt{\E^{(\leq L-2)}(N,\cdot)}\right)+\epsilon\,a^{-c\eta}\,\sqrt{\E^{(\leq L-2)}(R,\cdot)}
\end{align*}
Returning to \eqref{eq:diff-est-R-basic} with these estimates in hand, one obtains the following, using the Young inequality and \eqref{eq:ibp-trick-2+1}. 
\begin{align*}
-\del_t\left(a^2\,\E^{(L-1)}(R,\cdot)\right)\lesssim&\,\left(\change{\epsilon^\frac18\,a^{-2}+a^{-1-c\eta}}\right)\left(a^2\E^{(L-1)}(R,\cdot)\right)+\epsilon\,a^{-2}\left(a^2\,\E^{(L+1)}(N,\cdot)\right)\\
&\,+a^{-1-c\sqrt{\epsilon}}\,\E^{(L)}(N,\cdot)+\epsilon^{-\frac18}\,a^{-2}\,\E^{(L)}(\phi,\cdot)\\
&\,+\epsilon\,a^{-1-c\sqrt{\epsilon}}\left(\E^{(L)}(\hat{k},\cdot)+\E^{(L-2)}(R,\cdot)\right)\\
&\,+a^{-1-c\sqrt{\epsilon}}\,\E^{(\leq L)}_{1,\leq L}(f,\cdot)+\epsilon\,a^{-1-c\sqrt{\epsilon}}\,\E^{(\leq L-2)}(N,\cdot)\\
&\,+\left(\epsilon \,a^{-2-c\sqrt{\epsilon}}+\sqrt{\epsilon}\,a^{-1-c\sqrt{\epsilon}}\right)\E^{(\leq L-2)}(\phi,\cdot)\\
&\,+\epsilon\,a^{-1-c\sqrt{\epsilon}}\left(\E^{(\leq L-2)}(\hat{k},\cdot)+\E^{(\leq L-4)}(R,\cdot)\right)
\end{align*}
\change{The estimate \eqref{eq:en-est-Ric-top-2+1} }now follows from inserting the lapse estimate \eqref{eq:en-lapse1} up to order $L$ (which, in particular, also controls $a^2\E^{(L+1)}(N,\cdot)$ sufficiently) \change{and then integrating in time.\\

Regarding \eqref{eq:en-est-Ric-scaled-top-2+1}, the individual estimates lead to the following bound by distributing the scale factor weight between individual factors.
\begin{align*}
-\del_t\bigr(a^{2+\frac{\omega}4}&\,\E^{(L-1)}(R,\cdot)\bigr)\lesssim\left(\epsilon^\frac18\,a(s)^{-2}+a(s)^{-1-c\eta}\right)\left(a^2\E^{(L-1)}(R,\cdot)\right)\\
&\,+\sqrt{a^{2+\frac{\omega}4}\E^{(L-1)}(R,\cdot)}\left\{\epsilon a^{-2}\sqrt{a^{2+\frac{\omega}4}\E^{(L+1)}(N,\cdot)}+a^{-1-c\sqrt{\epsilon}+\frac{\omega}8}\sqrt{\E^{(L)}(N,\cdot)}+\right.\\
&\,\qquad+a^{-2+\frac{\omega}8}\sqrt{\E^{(L)}(\phi,\cdot)}+\epsilon a^{-1-c\sqrt{\epsilon}}\left(\sqrt{a^\frac{\omega}4\E^{(L)}(\hat{k},\cdot)}+\sqrt{a^\frac{\omega}4\E^{(L-2)}(R,\cdot)}\right)\\
&\,\qquad+\left(\epsilon a^{-2-c\sqrt{\epsilon}+\frac{\omega}8}+a^{-1-c\sqrt{\epsilon}+\frac{\omega}8}\right)\sqrt{\E^{(\leq L-2)}(\phi,\cdot)}\\
&\,\qquad+\left.\epsilon a^{-1-c\sqrt{\epsilon}}\left(\sqrt{a^\frac{\omega}4\E^{(\leq L-2)}(\hat{k},\cdot)}+\sqrt{a^\frac{\omega}4\E^{(\leq L-4)}(R,\cdot)}\right)\right\}\\
\end{align*}
One again inserts the lapse estimate \eqref{eq:en-lapse1}, distributing weights for the resulting terms in a similar manner. In particular, the leading order shear term from $a^2\E^{(L+1)}(N,\cdot)$ enters as follows.
\begin{align*}
&\,\left(\epsilon a^{-2}+a^{-1-c\sqrt{\epsilon}+\frac{\omega}8}\right)\sqrt{a^{2+\frac{\omega}4}\E^{(L-1)}(R,\cdot)}\sqrt{\epsilon^2a^{\frac{\omega}4}\E^{(L)}(\hat{k},\cdot)}\\
\lesssim &\,\epsilon^\frac18 a^{-2}a^{2+\frac{\omega}4}\E^{(L-1)}(R,\cdot)+\epsilon^2a^{-2}\,a^{\frac{\omega}4}\E^{(L)}(\hat{k},\cdot)
\end{align*}
Using Young's inequality as well as \eqref{eq:ibp-trick-2+1} and integrating in time, this then implies \eqref{eq:en-est-Ric-scaled-top-2+1}.}
\end{proof}

\subsection{Energy estimates for matter components}\label{subsec:en-est-matter}

Herein, we collect the energy estimates for scalar field and Vlasov matter. The proof for the former goes along similar lines as in \cite[Lemma 6.2]{FU24} and \cite[Lemma 5.1]{FU25}, but we verify that the delicate cancellations therein which rely on the respective Friedman equations still extend to this setting, including to $C<0$.

\begin{lemma}[Scalar field energy estimate]
\begin{align*}
\numberthis\label{eq:en-est-SF-2+1}&\,\E^{(L)}(\phi,t)+\int_t^{t_0}\dot{a}(s)\,a(s)\,\E^{(L+1)}(N,s)+\frac{\dot{a}(s)}{a(s)}\,\E^{(L)}(N,s)\,ds\\
\lesssim&\,\epsilon^4+\int_t^{t_0}\left({\epsilon}\,a(s)^{-2}+a(s)^{-1-c\sqrt{\epsilon}}\right)\E^{(L)}(\phi,s)\,ds\\
&\,+\int_t^{t_0}\left\{\epsilon\,a(s)^{-2}\E^{(L)}(\hat{k},s)+\epsilon^\frac32 \,a(s)^{-2}\,\E^{(L-2)}(R,s)\right.\\
&\,\phantom{+\int_t^{t_0}}+a(s)^{-1-c\sqrt{\epsilon}-(L+1)\omega}\cdot \left(a(s)^{(L+1)\omega}\,\E^{(\leq L)}_{1,\leq L}(f,s)+a(s)^{\frac{\omega}2}\,\|G-\gamma\|_{L^2_G(M_s)}^2\right)\\
&\,\phantom{+\int_t^{t_0}}+\sqrt{\epsilon}\,a(s)^{-2-c\sqrt{\epsilon}}\E^{(\leq L-2)}(\phi,s)+\epsilon\,a(s)^{-2-c\sqrt{\epsilon}}\,\E^{(\leq L-2)}(\hat{k},s)\\
&\,\phantom{+\int_t^{t_0}}\left.+\,\epsilon^\frac32\,a(s)^{-2-c\sqrt{\epsilon}}\E^{(\leq L-4)}(R,s)\right\}\,ds\,.
\end{align*}
\end{lemma}
\begin{proof}
Schematically, after commuting \eqref{eq:REEqnablaphi} and \eqref{eq:REEqWave-2+1} with $\Lap_G^\frac{L}2$ and applying integration by parts to cancel the highest order terms in $\nabla\phi$, one obtains the following.
\begin{align*}
-\del_t\E^{(L)}(\phi,\cdot)=&\,\int_M\left\{4\,C\,\frac{\dot{a}}a\,\Lap_G^\frac{L}2N\,\Lap^\frac{L}2\Psi+2\langle\nabla\Lap_G^\frac{L}2N,\nabla\phi\rangle_G\,\Lap_G^\frac{L}2\Psi\right.\\
&\,\phantom{\int_M}\left.-2\,C\,\langle\nabla\Lap_G^\frac{L}2N,\nabla\Lap_G^\frac{L}2\phi\rangle_G-2\,\dot{a}\,a\,\lvert\nabla\Lap_G^\frac{L}2\phi\rvert_G^2\right\}\,\vol{G}+\langle\text{error terms}\rangle\,.
\end{align*}
The first two terms arise from \eqref{eq:REEqWave-2+1}, the third from \eqref{eq:REEqnablaphi} and the fourth from differentiating the weight within the energy. Regarding the first term, applying $\Lap_G^\frac{L}2$ to both sides of \eqref{eq:REEqN} implies
\begin{subequations}
\begin{align*}
0=&\,-\frac1{4\pi}\,\dot{a}\,a\int_M\div_G\left(\nabla\Lap_G^\frac{L}2N\,\Lap_G^\frac{L}2N\right)\,\vol{G}\\
=&\,\int_M\left[-\frac{1}{4\pi}\,\dot{a}\,a\,\lvert\nabla\Lap_G^\frac{L}2N\rvert_G^2-\frac{1}{4\pi}\frac{\dot{a}}a\,\Lap_G^\frac{L}2(H+hN)\,\Lap_G^\frac{L}2N\right]\,\vol{G}\\
\numberthis\label{eq:lapse-0-1}=&\,\int_M\left[-\frac{1}{4\pi}\,\dot{a}\,a\,\lvert\nabla\Lap_G^\frac{L}2N\rvert_G^2-\left(4\,C^2\,\frac{\dot{a}}a+4\,\dot{a}\left(\rho^{Vl}_{FLRW}+2\,\mathfrak{p}^{Vl}_{FLRW}\right)-\frac{\kappa}{2\pi}\, \dot{a}\,a\right)\lvert\Lap_G^{\frac{L}2}N\rvert_G^2\right.\\
\numberthis\label{eq:lapse-0-2}&\,\left.-\change{\,4\,C\,}\frac{\dot{a}}a\,\Lap_G^\frac{L}2\Psi\,\Lap_G^\frac{L}2N-4\,\dot{a}\change{\Lap^\frac{L}2\left(\rho^{Vl}-\rho^{Vl}_{FLRW}}\right)\,\Lap_G^\frac{L}2N+\langle\text{nonlinear error terms}\rangle\right]\,\vol{G}\,.
\end{align*}
\end{subequations}

Inserting this zero into \eqref{eq:diff-en-est-SF}, the first two terms of \eqref{eq:lapse-0-1} have definite sign and lead to the terms on the left hand side of \eqref{eq:en-est-SF-2+1} after absorbing some more terms discussed below, while the remaining ones in \eqref{eq:lapse-0-1} are bounded by $a^{-1}\,\E^{(L)}(N,\cdot)$. The first term in the \eqref{eq:lapse-0-2} cancels the first term on the right hand side of \eqref{eq:diff-en-est-SF}, while the latter can be bounded using \eqref{eq:near-coerc-Vl1}. Altogether, this means that the below bound holds for an appropriate constant $K>0$.
\begin{align*}
\numberthis\label{eq:diff-en-est-SF}&\,\del_t\E^{(L)}(\phi,\cdot)+\frac{1}{4\pi}\,\dot{a}\,a\,\E^{(L+1)}(N,\cdot)+4\,C^2\frac{\dot{a}}a\,\E^{(L)}(N,\cdot)\\
\leq&\,\int_M\left\{2\,\langle\nabla\Lap_G^\frac{L}2N,\nabla\phi\rangle_G\,\Lap_G^\frac{L}2\Psi-2\,C\,\langle\nabla\Lap_G^\frac{L}2N,\nabla\Lap_G^\frac{L}2\phi\rangle_G-2\,\dot{a}\,a\,\lvert\nabla\Lap_G^\frac{L}2\phi\rvert_G^2\right\}\,\vol{G}\\
&\,+K\left(a^{-1}\,\E^{(L)}(N,\cdot)+a^{-1-c\sqrt{\epsilon}}\,\change{\left(\E^{(L)}_{1,L}(f,\cdot)+\|\change{G^{-1}-\gamma^{-1}}\|_{L^2_G}^2\right)}\right)+\langle\text{error terms}\rangle\,.
\end{align*}
Since, by \eqref{eq:APder}, the first term on the right hand side of \eqref{eq:diff-en-est-SF} is bounded in absolute value by
\begin{equation*}
\lesssim\int_M\sqrt{\epsilon}\,a^{-c\sqrt{\epsilon}}\,\lvert\nabla\Lap_G^\frac{L}2N\rvert_G\,\lvert\Lap_G^\frac{L}2\Psi\rvert_G\,\vol{G}\lesssim \int_M\epsilon\,\dot{a}\,a\,\lvert\nabla\Lap_G^\frac{L}2N\rvert_G^2\,\vol{G}+a^{-c\sqrt{\epsilon}}\,\E^{(L)}(\phi,\cdot)\,,
\end{equation*}
we can absorb this resulting lapse term into the left hand side of \eqref{eq:diff-en-est-SF} by updating prefactors, while we simply keep the scalar field term on the right hand side.

Moving on to the second term on the right hand side of \eqref{eq:diff-en-est-SF}, \eqref{eq:Friedman-ineq} and Young's inequality imply
\begin{align*}
\left\lvert 2\,C\,\langle\nabla\Lap_G^\frac{L}2N,\nabla\Lap_G^\frac{L}2\phi\rangle_G\right\rvert\leq&\, \frac{1}{\sqrt{\pi}}\left({\dot{a}}\,a+\sqrt{\max\{0,\kappa\}}\,a\right)\lvert\nabla\Lap_G^\frac{L}2N\rvert_G\,\lvert\nabla\Lap_G^\frac{L}2\phi\rvert_G\\
\leq&\,\left({\dot{a}}\,a+\sqrt{\max\{0,\kappa\}}\,a\right)\left(2\,\lvert\nabla\Lap_G^\frac{L}2\phi\rvert_G^2+\frac{1}{8\pi}\,\lvert\nabla\Lap_G^\frac{L}2N\rvert_G^2\right)\,,
\end{align*}
The first term, combined with the third term in \eqref{eq:diff-en-est-SF}, only leaves
\[\sqrt{\max\{0,\kappa\}}\int_M a\,\lvert\nabla\Lap_G^\frac{L}2\phi\rvert_G^2\,\vol{G}\lesssim a^{-1}\,\E^{(L)}(\phi,\cdot)\,.\]
The second term can also be absorbed into the second term on the left of \eqref{eq:diff-en-est-SF} up to the error term $a^{-1}\E^{(L)}(N,\cdot)$ that we keep on the right.

In summary, this shows
\begin{align*}
-\del_t\E^{(L)}(\phi,\cdot)&\,+\dot{a}\,a\,\E^{(L+1)}(N,\cdot)+\frac{\dot{a}}a\,\E^{(L)}(N,\cdot)\\
\lesssim&\,a^{-1-c\sqrt{\epsilon}}\left(\E^{(L)}(\phi,\cdot)+\E^{(L)}(N,\cdot)+\E^{(\leq L)}_{1,\leq L}(f,\cdot)+\|\change{G^{-1}-\gamma^{-1}}\|_{L^2_G(M_\cdot)}^2\right)\\
&\,+\langle\text{collected borderline and junk error terms}\rangle\,.
\end{align*}

From here, \eqref{eq:en-est-SF-2+1} follows by estimating the error terms with the help of the low order bounds in Section \ref{sec:AP}, the elliptic lapse estimate Lemma \ref{lem:en-lapse} and the near-coercivity estimate Lemma \ref{lem:near-coerc}.
\end{proof}

Regarding the estimates for Vlasov matter, the structural features exploited in \cite[Section 6]{FU25} fully survive in \eqref{eq:REEqVlasov} and thus allow us to extract the energy estimates in Lemma \ref{lem:vlasov-total} as in the higher-dimensional setting. We provide estimates \eqref{eq:vlasov-en-0}-\eqref{eq:vlasov-total}, which are necessary for obtaining bounds on $\E^{(L)}_\textrm{\normalfont total}$, as well as the unscaled estimates \eqref{eq:vlasov-indiv}, which are needed to improve the bootstrap assumptions on Vlasov matter.

\begin{lemma}[Total Vlasov energy estimate]\label{lem:vlasov-total} At order $0$, the following bound holds:
\begin{align*}\numberthis\label{eq:vlasov-en-0}
\change{a(t)^{\omega}\,\E^{(0)}_{1,0}(f,t)}\lesssim&\,\epsilon^4+\int_t^{t_0}\left(\epsilon a(s)^{-2}+a^{-2+\change{\frac{\omega}4}}\right)\, a(s)^\omega\E^{(0)}_{1,0}(f,s)\,ds\\
&\,+\int_t^{t_0}a(s)^{-2-c\sqrt{\epsilon}+\frac{\omega}4}\left(\E^{(0)}(\phi,s)+\change{\epsilon^2\,\E^{(0)}(\hat{k},s)}+a(s)^{2+\frac{\omega}2}\|\Gamma-\Gamhat\|_{L^2_G(M_s)}^2\right)\,ds\\
&\,+\int_t^{t_0}a(s)^{-1-c\sqrt{\epsilon}+\frac{\omega}4}\, a(s)^{\frac{\omega}2}\|\change{G^{-1}-\gamma^{-1}}\|_{L^2_G(M_s)}^2\,ds\,.
\end{align*}
For $L\in\{2,4,6,8,10\}$, the total Vlasov energy $\E^{(L)}_\textrm{\normalfont total,Vl}$ as defined in Definition \ref{def:en} obeys the following energy estimate.
\begin{subequations}\label{eq:vlasov-total}
\begin{align*}\numberthis\label{eq:vlasov-total-scaled-2+1}
\E^{(L)}_\textrm{\normalfont total,Vl}(t)\lesssim&\,\epsilon^4+\int_t^{t_0}\left(a(s)^{-1-{\omega}}+a(s)^{-2-c\sqrt{\epsilon}+\frac{\omega}{2}}+\epsilon^\frac18\,a(s)^{-2}\right)\E^{(L)}_\textrm{\normalfont total,Vl}(s)\,ds\\
&\,+\int_t^{t_0}a(s)^{-2-c\sqrt{\epsilon}+L\omega}\left(\E^{(L)}(\phi,s)+\change{a(s)^\frac{\omega}4}\E^{(L)}(\hat{k},s)\right)\,ds\\
&\,+\int_t^{t_0}a(s)^{-2-c\sqrt{\epsilon}+\change{L\omega}}\, \change{a(s)^{2+\frac{\omega}4}}\,\E^{(L-1)}(R,s)+\change{a(s)^{-2-c\sqrt{\epsilon}+L\omega}}\,a(s)^{\frac{\omega}2}\E^{(L-2)}(R,s)\,ds\\
&\,+\int_t^{t_0}\left(a(s)^{-1-\omega\change{-c\eta}}+a(s)^{-2-c\sqrt{\epsilon}+\frac{\omega}2}\right)\E^{(\leq L-2)}_\textrm{\normalfont total,Vl}(f,s)\,ds\\
&+\int_t^{t_0}a(s)^{-2-c\sqrt{\epsilon}+\omega}\left(\E^{(\leq L-2)}(\phi,s)+\change{a(s)^\frac{\omega}4}\E^{(\leq L-2)}(\hat{k},s)\right)\,ds\\
&\,+\int_t^{t_0}\change{a(s)^{-2-c\sqrt{\epsilon}+\omega}}\,\change{a(s)^\frac{\omega}2}\E^{(\leq L-4)}(R,s)\,ds\\
&\,+\int_t^{t_0}\change{a(s)^{-1-c\sqrt{\epsilon}}\cdot}\, a(s)^\frac{\omega}2\left(\|\Gamma-\Gamhat\|_{H^1_G(M_s)}^2+\|\change{G^{-1}-\gamma^{-1}}\|_{H^2_G(M_s)}^2\right)\,ds\,.
\end{align*}
Recalling $\change{\E^{(L)}_\textrm{\normalfont total}}$ from Definition \ref{def:en}, this becomes
\change{\begin{align*}\numberthis\label{eq:total-Vl-scaled-summ}
\,\E^{(L)}_\textrm{\normalfont total,Vl}(t)\lesssim&\,\epsilon^4+\int_t^{t_0}\left(a(s)^{-1-{\omega}-c\sqrt{\epsilon}}+a(s)^{-2-c\sqrt{\epsilon}+\frac{\omega}{4}}+\epsilon^\frac18\,a(s)^{-2}\right)\E^{(L)}_\textrm{\normalfont total}(s)\,ds\\
+\int_t^{t_0}&\left(a(s)^{-1-\omega-c\eta}+a(s)^{-2-c\sqrt{\epsilon}+\frac{\omega}4}\right)\E^{(\leq L-2)}_\textrm{\normalfont total}(s)\,ds\,.
\end{align*}}
\end{subequations}
Additionally, the following individual unscaled estimates hold for any integers $J,K$ with $0\leq K\leq J\leq 10$:
\begin{align*}\numberthis\label{eq:vlasov-indiv}
\E^{(J)}_{1,K}(f,t)\lesssim&\,\epsilon^4+\int_t^{t_0}\left(a(s)^{-1-c\eta}+\epsilon^\frac18\,a(s)^{-2}\right)\E^{(J)}_{1,K}(f,s)\,ds\\
&\,+\underbrace{\int_t^{t_0} a(s)^{-1-c\eta-(K+2)\omega}\,a(s)^{(K+2)\omega}\E^{(J)}_{1,K+1}(f,s)\,ds}_{\text{if }K<J}\\
&\,+\int_t^{t_0}\left\{\epsilon^\frac78\,a(s)^{-2-c\sqrt{\epsilon}}\left(\E^{(J)}_{1,\leq K-1}(f,s)+\E^{(\leq J-1)}_{1,\leq K}(f,s)\right)\right.\\
&\,\phantom{\int_t^{t_0}}+\epsilon^{-\frac18}\,a(s)^{-2-c\sqrt{\epsilon}}\left(\E^{(\leq J)}(\phi,s)+\E^{(\leq J)}(\hat{k},s)+\change{\E^{(\leq J-2)}(R,s)+a(s)^2\E^{(J-1)}(R,s)}\right)\\
&\,\phantom{\int_t^{t_0}}+\left.a(s)^{-1-c\sqrt{\epsilon}}\left(\|\Gamma-\Gamhat\|_{H^1_G(M_s)}^2+\|\change{G^{-1}-\gamma^{-1}}\|_{H^2_G(M_s)}^2\right)\right\}\,ds\,.
\end{align*}
\end{lemma}
\begin{proof}
Again, we only provide an outline:
Firstly, recalling the Vlasov operator from \eqref{eq:REEqVlasov-gen}, one has that, for $K>0$,
\changereport{
\begin{equation}\label{eq:vlasov-main-mech}
a^{(K+1)\omega}\left\lvert\int_{T^\ast M}\langle v\rangle_G^2\langle\X\nabsak_{\mathrm{hor}}^K\nabsak_{\mathrm{vert}}^{J-K} f,\nabsak_{\mathrm{hor}}^K\nabsak_{\mathrm{vert}}^{J-K}f\rangle_{\G_0}\,\vol{\G}\right\rvert\lesssim \epsilon\,a^{1-c\eta}\,a^{(K+1)\omega}\E^{(L)}_{1,K}(f,\cdot)
\end{equation}}
and the analogous bound for $K=0$ replacing $f$ with $f-f_{FLRW}$\change{, see \cite[Lemma 6.2]{FU25}}. This is shown by integrating by parts and using \eqref{eq:APMom-2+1}, \eqref{eq:bs-ass-HC} and \eqref{eq:APkhat0} to bound terms that arise when derivatives hit variables other than \change{$f$}. This provides the core energy conservation mechanism and leads directly to \eqref{eq:vlasov-en-0} after bounding error terms arising from $\X f_{FLRW}$ as well as various error terms that arise when the time derivative hits momentum weights or metric terms. These can be controlled using the same tools along with \eqref{eq:APder}. The final line of \eqref{eq:vlasov-total-scaled-2+1} occurs from reference terms similar to \eqref{eq:ibp-Vlasov-1-1}.\\

\change{For higher orders, we consider the evolution of $a^{(K+1)\omega}\E^{(L)}_{1,K}(f,\cdot)$ for $L>K>0$ as an example:
\begin{subequations}\label{eq:vlasov-diff-est-2+1}
\begin{align*}
&\,-a^{(K+1)\omega}\del_t\E^{(L)}_{1,K}(f,\cdot)\\
\numberthis\label{eq:vlasov-diff-est-2+1-basic1}=&\,-a^{(K+1)\omega}\int_{T^\ast M}\langle v\rangle_G^2\left(2\frac{\del_t\langle v\rangle_G}{\langle v\rangle_G}+2(L-K)\frac{\del_tv^0}{v^0}\right)\lvert\nabsak^{L-K}_{vert}\nabsak^K_{hor}f\rvert_{\G_0}^2\,\vol{\G}\\
\numberthis\label{eq:vlasov-diff-est-2+1-basic2}&\,-a^{(K+1)\omega}\int_{T^\ast M}\langle v\rangle_G^2\,\del_t\G^{- 1}\ast_{\G_0}\changereport{\nabsak^{L-K}_{vert}\nabsak^K_{hor}f\ast_{\G_0}\nabsak^{L-K}_{vert}\nabsak^K_{hor}f}\,\vol{\G}\\
\numberthis\label{eq:vlasov-diff-est-2+1-delt}&\,-a^{(K+1)\omega}\int_{T^\ast M}\langle v\rangle_G^2\langle[\del_t,\nabsak_{vert}^{L-K}\nabsak_{hor}^K]f,\nabsak_{vert}^{L-K}\nabsak_{hor}^K]f\rangle_{\G_0}\,\vol{\G}\\
\numberthis\label{eq:vlasov-diff-est-2+1-X}&\,-a^{(K+1)\omega}\int_{T^\ast M}\langle v\rangle_G^2\langle[\nabsak_{vert}^{L-K}\nabsak_{hor}^K,\X]f,\nabsak_{vert}^{L-K}\nabsak_{hor}^K]f\rangle_{\G_0}\,\vol{\G}\\
&\,+\langle\text{term controlled via \changereport{\eqref{eq:vlasov-main-mech}}}\rangle\,.
\end{align*}
\end{subequations}
\changereport{The lines \eqref{eq:vlasov-diff-est-2+1-basic1}-\eqref{eq:vlasov-diff-est-2+1-basic2} }can be controlled in absolute value by $\lesssim \epsilon a^{-2}\E^{(L)}_{1,K}(f,\cdot)$ by inserting the evolution of the spatial metric \eqref{eq:REEqG-2+1} (and that of its inverse) and controlling resulting zero order shear and lapse terms with \eqref{eq:APkhat0} and \eqref{eq:bs-ass-HC} respectively. 
\eqref{eq:vlasov-diff-est-2+1-delt} can be controlled by bounds on $\|\del_tR[G]\|_{H^{K-2}}$ and $\|\del_t\Gamma[G]\|_{H^{K-1}}$, see \eqref{eq:REEqR} and \eqref{eq:REEqChr-2+1}, as well as using Lemma \ref{lem:AP-2+1} to estimate low order terms and Lemma \ref{lem:APMom-2+1} to estimate momentum weights where necessary. This leads to
\begin{align*}
&a^{\frac{(K+1)\omega}2}\left\|[\del_t,\nabsak_{\mathrm{hor}}^{K}\nabsak_{\mathrm{vert}}^{L-K}]f\right\|_{L^2_{{\G}_0}(T^\ast M_\cdot)}\\
\lesssim&\,a^{-2+\frac\omega2-c\sqrt{\epsilon}}\left(\sqrt{a^{K\omega}\E^{(\leq K)}(\hat{k},\cdot)}+\sqrt{a^{K\omega}\E^{(\leq K)}(N,\cdot)}+\epsilon\,\sqrt{a^{K\omega}\E^{(\leq K-2)}(\Ric,\cdot)}\right)\\
&\,+\epsilon\,a^{-2+\frac{\omega}2-c\sqrt{\epsilon}}\,\left(\sqrt{a^{K\omega}\E^{(L)}_{1,K-1}(f,\cdot)}+\sqrt{a^{K\omega}\E^{(\leq L-1)}_{1,\leq K-1}(f,\cdot)}\right)\,.
\end{align*}
After applying Young's inequality, \eqref{eq:ibp-trick-2+1} and \eqref{eq:ibp-Vlasov-1-1} as usual, the resulting terms are controlled by the first four lines of \eqref{eq:vlasov-total-scaled-2+1}.

Regarding \eqref{eq:vlasov-diff-est-2+1-X}, }the two main considerations need to be taken when commuting with the horizontal term and the shear term in $\X$. When commuting with the horizontal term, one obtains the following high order terms
\begin{equation}\label{eq:Vlasov-en-est-2+1-interim}
a^{-1}\,\nabla^{K-1}R[G]\ast\frac{v\ast v}{v^0}\ast\nabsak_{\mathrm{vert}}^{L-K+1}f+a^{-1}\left(\frac1{v^0}+\frac{v\ast v}{(v^0)^3}\right)\ast\nabsak_{\mathrm{vert}}^{L-K-1}\nabsak_{\mathrm{hor}}^{K+1}f\,,
\end{equation}
\change{among some lower order nonlinear terms. The first term in \eqref{eq:Vlasov-en-est-2+1-interim} can be controlled in $L^2_{1,\G_0}$ by $a^{-2-c\sqrt{\epsilon}}\sqrt{a^2\,\E^{(K-1)}(R,\cdot)}$ and is the reason we need to include this high order scaled energy in the third line of \eqref{eq:vlasov-total-scaled-2+1}. The second term can be bounded by $a^{-1}\sqrt{\E^{(L)}_{1,K+1}(f,\cdot)}$, taking the momentum weights in ${\G}_0$ into account. Thus, these two leading order terms from \eqref{eq:Vlasov-en-est-2+1-interim} in \eqref{eq:vlasov-diff-est-2+1-X} can be controlled by
\begin{align*}
\lesssim&\, \left(a^{-2+\frac{\omega}2}+a^{-1}\right)\,a^{(K+1)\omega}\E^{(L)}_{1,K}(f,\cdot)\\
&\,+a^{-2+K\omega-c\sqrt{\epsilon}}\E^{(K-1)}(R,\cdot)+a^{-1-\omega}\,a^{(K+2)\omega}\E^{(L)}_{1,K+1}(f,\cdot)\,.
\end{align*}
Note that the final term can be bounded by $a^{-1-\omega}\,\E^{(L)}_\textrm{\normalfont total,Vl}$, and thus is absorbed into the first line of \eqref{eq:vlasov-total-scaled-2+1}.
Regarding the vertical derivative term in the Vlasov operator, high order lapse terms can be controlled using Lemma \ref{lem:en-lapse} and \eqref{eq:APVlasov-2+1}, and high order Vlasov terms have negligible effect since $\nabla N$ converges at low orders due to the bootstrap assumption \eqref{eq:bs-ass-HC}. In summary, integrating \eqref{eq:vlasov-diff-est-2+1} for each $K$ leads to right hand sides controlled by \eqref{eq:vlasov-total-scaled-2+1}, and after summing over $K$, to \eqref{eq:vlasov-total-scaled-2+1}-\eqref{eq:total-Vl-scaled-summ}}.

The bound \eqref{eq:vlasov-indiv} is extracted similarly, where instead of exploiting a scaling hierarchy, one simply keeps energy terms that would otherwise barely fail to be integrable, at the cost of the weight $\epsilon^{-\frac18}$ when applying Young's inequality. Note that, once these spacetime quantities are sufficiently well controlled via estimates for $\E^{(L)}_\textrm{\normalfont total}$, these are simply inhomogeneous error terms. Thus, these estimates will be sufficient to close the Vlasov energy argument.
\end{proof}

\subsection{Total energy estimate}\label{subsec:en-est-total}

Finally, we can combine the energy estimates from this section to obtain the following total energy bounds.
\begin{corollary}[Total energy bounds]\label{cor:en-imp} For any $L\in \{0,2,4,6,8,10\}$ \change{and sufficiently small $\omega\gg\eta$ (e.g., $\omega=\frac{1}{100}$ and $\eta$ sufficiently small)}, the total energy satisfies the following bound.
\begin{equation}\label{eq:total-en-imp-2+1}
\E^{(\leq L)}_\textrm{\normalfont total}(t)\lesssim \epsilon^4\,a(t)^{-c\epsilon^\frac18}\,.
\end{equation}
Furthermore, \change{one has 
\begin{equation}\label{eq:quiesc-en-imp-2+1}
\E^{(\leq L)}_\textrm{\normalfont quiesc}(t)\lesssim \epsilon^4\,a(t)^{-c\epsilon^\frac18}
\end{equation}
as well as the following energy bounds for the lapse and curvature.
\begin{align}\label{eq:lapse-en-imp}
a^4\,\E^{(12)}(N,t)+a^2\,\E^{(11)}(N,t)+\E^{(\leq 10)}(N,t)\lesssim&\, \epsilon^\frac{15}4\,a^{2-c\epsilon^\frac18}\,,\\
a^2\E^{(9)}(R,\cdot)\lesssim&\,\epsilon^\frac{15}4a(t)^{2-c\epsilon^\frac18}\label{eq:Ric-en-top-imp}\,.
\end{align}
Furthermore, for integers $J,K$ with $0\leq K\leq J\leq 10$, one has
\begin{equation}\label{eq:vlasov-en-imp-2+1}
\E^{(\leq J)}_{1,\leq K}(f,t)\lesssim\epsilon^\frac72\,a(t)^{-c\epsilon^\frac18}\,.
\end{equation}}
\end{corollary}
\begin{proof}
First, we combine the energy estimates within this section into the following total estimate.
\begin{align*}\numberthis\label{eq:total-en-est}
\E^{(L)}_\textrm{\normalfont total}(t)\lesssim&\,\epsilon^4+\int_t^{t_0}\left(\epsilon^\frac18\,a(s)^{-2}+a(s)^{-1-c\eta-(L+1)\omega}+a(s)^{-2-c\sqrt{\epsilon}+\change{\frac{\omega}8}}\right)\E^{(L)}_\textrm{\normalfont total}(s)\,ds\\
&\,+\int_t^{t_0}\change{\left(\epsilon^{\frac38}\,a(s)^{-2-c\sqrt{\epsilon}}+a(s)^{-1-c\eta-(L+1)\omega}\right)}\,\E^{(\leq L-2)}_\textrm{\normalfont total}(s)\,ds\,.
\end{align*}
More precisely, this is obtained by estimating terms in the following manner, in the same order as in \eqref{eq:total-en-def}.
\change{\begin{itemize}
\item For any $L$, use the following estimates to control the first line of \eqref{eq:total-en-def}.
\begin{itemize}
\item Apply \eqref{eq:en-est-SF-2+1} to the scalar field energy.
\item Apply \eqref{eq:en-k} and \eqref{eq:en-k-scaled} to the shear energy, with the former to the summand scaled by $\epsilon^\frac14$ and the latter to the one scaled by $a^{\frac{\omega}4}$.
\item Apply \eqref{eq:en-est-Ric-2+1} and \eqref{eq:en-est-Ric-scaled-2+1} to the curvature energies at order $L-2$ analogously.
\end{itemize}
\item The terms in the second line of \eqref{eq:total-en-def} are treated as follows.
\begin{itemize}
\item To control the Vlasov energy, apply \eqref{eq:vlasov-en-0} at order zero and \eqref{eq:total-Vl-scaled-summ} for $L\geq 2$.
\item Apply \eqref{eq:en-est-Ric-top-2+1} and \eqref{eq:en-est-Ric-scaled-top-2+1} to the curvature energies at order $L-1$.
\end{itemize}
\item Moreover, the metric errors in the third line are controlled as follows:
\begin{itemize}
\item For $L=0$, one uses \eqref{eq:G-norm-int-est} and \eqref{eq:en-est-Chr0}.
\item For $L=2$, one uses \eqref{eq:G-norm-int-est} and \eqref{eq:en-est-Chr1}.
\end{itemize}
\end{itemize}}
  
Applying the Gronwall lemma to \eqref{eq:total-en-est} then immediately implies \eqref{eq:total-en-imp-2+1} for $L=0$, since the second line does not occur. Assuming \eqref{eq:total-en-imp-2+1} has been shown for $L-2$, \eqref{eq:total-en-est} implies
\change{\begin{align*}
\E^{(L)}_\textrm{\normalfont total}(t)\lesssim&\,\int_t^{t_0}\left(\epsilon^\frac18\,a(s)^{-2}+a(s)^{-2-c\sqrt{\epsilon}+\frac{\omega}8}+a(s)^{-1-c\eta-(L+1)\omega}\right)\E^{(L)}_\textrm{\normalfont total}(s)\,ds\\
&\,+\epsilon^4\left[1+\int_t^{t_0}\left(\epsilon^\frac18a(s)^{-2-c\epsilon^\frac18}+a(s)^{-1-(L+1)\omega-c\eta}\right)\,ds\right]
\end{align*}}
and thus \eqref{eq:total-en-imp-2+1} at order $L$ after applying \eqref{eq:a-int-est} to the \change{second line }as well as the Gronwall lemma. After completing this iteration, \eqref{eq:total-en-imp-2+1} is proven for all $L$ in the statement.\\

\change{Moving on to $\E^{(L)}_\textrm{\normalfont quiesc}$, one obtains the following estimate combining \eqref{eq:en-est-SF-2+1}, \eqref{eq:en-k} and \eqref{eq:en-est-Ric-2+1}.
\begin{align*}
\E^{(L)}_\textrm{\normalfont quiesc}&\,(t)\lesssim\epsilon^4+\int_t^{t_0}\left(\epsilon^\frac18a(s)^{-2}+a(s)^{-1-c\eta}\right)\E^{(L)}_\textrm{\normalfont quiesc}(s)\,ds\\
&\,+\int_t^{t_0}\left(\epsilon^\frac38 \,a(s)^{-2-c\sqrt{\epsilon}}+a(s)^{-1-c\sqrt{\epsilon}}\right)\E^{(\leq L-2)}_\textrm{\normalfont quiesc} (s)\,ds\\
&\,+\int_t^{t_0}a(s)^{-1-c\sqrt{\epsilon}-(L+1)\omega}\left\{\E^{(\leq L)}_\textrm{\normalfont total,Vl}(s)+a(s)^\frac{\omega}2\left(\|G^{-1}-\gamma^{-1}\|^2_{H^2_G(M_s)}+\|\Gamma-\Gamhat\|_{H^1_G(M_s)}^2\right)\right\}\,ds\,.
\end{align*}
In the final line, we collect all possible metric error terms across different orders for the sake of simplicity. Using \eqref{eq:total-en-imp-2+1} for the final line, and assuming that the bound \eqref{eq:quiesc-en-imp-2+1} has been obtained for order $L-2$ or lower, one again obtains
\[\E^{(L)}_\textrm{\normalfont quiesc}(t)\lesssim \epsilon^4a(t)^{-c\epsilon^\frac18}+\int_t^{t_0}\left(\epsilon^\frac18a(s)^{-2}+a(s)^{-1-c\eta}\right)\E^{(L)}_\textrm{\normalfont quiesc}(s)\,ds\]
and thus \eqref{eq:quiesc-en-imp-2+1} at order $L$ via the Gronwall lemma. This proves the bound for $L\in 2\N,\,L\leq 10$ by iteration. Furthermore, \eqref{eq:lapse-en-imp} follows from inserting \eqref{eq:total-en-imp-2+1} and \eqref{eq:quiesc-en-imp-2+1} into \eqref{eq:en-lapse2}, using the former to bound Vlasov energies and metric errors and the latter for any remaining terms. Similarly, applying \eqref{eq:total-en-imp-2+1} and \eqref{eq:quiesc-en-imp-2+1} to \eqref{eq:en-est-Ric-top-2+1} for $L=10$ implies
\[a(t)^2\E^{(9)}(R,t)\lesssim \epsilon^4+\epsilon^\frac{15}4a(t)^{-c\epsilon^\frac18}+\int_t^{t_0}\left(\epsilon^\frac18a(s)^{-2}+a(s)^{-1-c\eta}\right)\,a(s)^2\E^{(L-1)}(R,s)\,ds\]
and thus \eqref{eq:Ric-en-top-imp} by the Gronwall lemma.\\

Finally, we turn to \eqref{eq:vlasov-indiv} to prove \eqref{eq:vlasov-en-imp-2+1}. Note that we can use \eqref{eq:total-en-imp-2+1} to bound the second line of \eqref{eq:vlasov-indiv} by
\[\leq \int_t^{t_0}a(s)^{-1-c\eta-{(K+2)}\omega}\,\E^{(J)}_\textrm{\normalfont total}(s)\,ds\lesssim \int_t^{t_0}\epsilon^4a(s)^{-1-(K+2)\omega-c\eta}\,ds\,.\]
Additionally, one can use \eqref{eq:quiesc-en-imp-2+1}, along with \eqref{eq:Ric-en-top-imp} for the top order curvature term, to bound the penultimate line of \eqref{eq:vlasov-indiv}, as well as \eqref{eq:total-en-imp-2+1} to bound the metric errors in the final line. This yields
\begin{align*}
\E^{(J)}_{1,K}(f,t)\lesssim&\,\epsilon^4+\int_t^{t_0}\left(a(s)^{-1-c\eta-\change{(K+2)\omega}}+\epsilon^\frac18\,a(s)^{-2}\right)\E^{(J)}_{1,K}(f,s)\,ds\\
&\,+\int_t^{t_0}\epsilon^\frac78\,a(s)^{-2-c\sqrt{\epsilon}}\left(\E^{(J)}_{1,\leq K-1}(f,s)+\E^{(\leq J-1)}_{1,\leq K}(f,s)\right)\,ds\\
&\,+\int_t^{t_0}\epsilon^{\frac{29}8}\,a(s)^{-2-c\epsilon^\frac18}+\epsilon^4a(s)^{-1-(K+1)\omega-c\eta}\,ds\,.
\end{align*}
For $\omega<\frac{1}{24}$ \change{and $\eta\ll\omega$}, the last line is simply bounded by $\epsilon^4a(t)^{-c\epsilon^\frac18}$ up to a constant. Again, \eqref{eq:vlasov-en-imp-2+1} for $J=K=0$ now follows directly from the Gronwall lemma, since the second line does not appear, and the full statement follows by iterating over $K$ for each $J>0$ and then over $J$.}
\end{proof}

\section{Past Stability}\label{sec:main}

\subsection{The $(2+1)$-dimensional ESFV system}

Corollary \ref{cor:en-imp} now puts us in the position to close the bootstrap argument and prove the main result of this paper:

\begin{theorem}[Past stability of FLRW solutions to the $(2+1)$-dimensional ESFV system]\label{thm:main-2+1}
Let $(M,\mathring{g},\mathring{k},\mathring{\pi},\mathring{\psi},\mathring{f})$ be CMC \change{initial data for }the Einstein scalar-field Vlasov system as in Section \ref{subsec:init-2+1} that is close to the FLRW solution from Lemma \ref{lem:FLRW-2+1} in the sense of Assumption \ref{ass:init-2+1}. Further, assume $\mathring{f}$ to have compact momentum support. If $m=0$, additionally assume the momentum support of $\mathring{f}$ to be bounded away from the origin.

Then, the past maximal globally hyperbolic development $(\M,\g,\phi,\f)$ admits a CMC foliation $(M_t)_{t\in(0,t_0]}$ along which the following estimates hold for some $c^\prime>0$.
\begin{equation}\label{eq:bs-imp}
\mathcal{H}+\mathcal{C}\lesssim\epsilon^\frac74\,\changereport{a}^{-c^\prime\epsilon^\frac18},\quad \|f-f_{FLRW}\|_{C^5_{1,\underline{\gamma}_0}(\changereport{T^\ast M})}\lesssim \sqrt{\epsilon}\,a^{-c^\prime\sqrt{\epsilon}}\,.
\end{equation}
Consequently, the solution is asymptotically velocity term dominated in the following sense: There exist a scalar function $\Psi_\textrm{\normalfont Bang}$, a tracefree $(1,1)$-tensor field $\hat{K}_\textrm{\normalfont Bang}$ and a \change{$(0,2)$-tensor }$H_\textrm{\normalfont Bang}$ \change{on the Big Bang hypersurface $M_0$, which we identify with $M$ from here on out, }satisfying
\begin{equation}
\|\Psi_\textrm{\normalfont Bang}\|_{C^{7}_\gamma(M)}+\change{\|\hat{K}_\textrm{\normalfont Bang}\|_{C^{7}_\gamma(M)}}+\|H_\textrm{\normalfont Bang}-\gamma\|_{C^{7}_\gamma(M)}\lesssim\epsilon
\end{equation}
such that the following estimates hold.
\begin{align*}\numberthis\label{eq:asymp-quiesc}
\|n-1\|_{C^6_\gamma(M_t)}+\|a^2\,\del_t\phi-(\Psi_\textrm{\normalfont Bang}+C)\|_{C^6_\gamma(M_t)}+
\|a^2\,\hat{k}^\sharp-\hat{K}_\textrm{\normalfont Bang}\|_{C^6_\gamma(M_t)}&\,\\
+\left\|a^{-2}\,g\odot\exp\left[-2\int_t^{t_0}a(s)^{-2}\,ds\cdot\hat{K}_\textrm{\normalfont Bang}\right]-H_\textrm{\normalfont Bang}\right\|_{C^6_\gamma(M_t)}&\,\lesssim\epsilon\,a(t)^{1-c^\prime\epsilon^\frac18}\,,\\
\numberthis\label{eq:asymp-phi-2+1}\left\|\phi-\phi(t_0,\cdot)+\int_t^{t_0}a(s)^{-2}\,ds\cdot(\Psi_\textrm{\normalfont Bang}+C)\right\|_{C^6_\gamma(M_t)}&\,\lesssim\epsilon\,a(t)^{1-c^\prime\epsilon^\frac18}\,.
\end{align*}
Furthermore, one has
\begin{align}\label{eq:asymp-Ham}
-8\pi\,\Psi_\textrm{\normalfont Bang}\,(\Psi_\textrm{\normalfont Bang}+2\,C)=\hat{K}_\textrm{\normalfont Bang}\cdot\hat{K}_\textrm{\normalfont Bang}\,
\end{align}
where, for $(1,1)$ tensor fields $\mathfrak{S},\mathfrak{T}$ on $M$, we write $\mathfrak{S}\cdot\mathfrak{T}=\mathfrak{S}^i_{\ j}\mathfrak{T}^j_{\ i}$. Additionally, there exists \change{$f_\textrm{\normalfont Bang}\in C^{4}_{1,\underline{\gamma}_0}(T^\ast M)$ }satisfying 
\change{\begin{equation}\label{eq:Vlasov-control}
\|f_\textrm{\normalfont Bang}-\mathring{f}\|_{C^4_{1,\underline{\gamma}_0}(T^\ast M)}\lesssim \sqrt{\epsilon}
\end{equation}}
such that, 
one has
\change{\begin{equation}\label{eq:asymp-Vlasov}
\left\|f-f_\textrm{\normalfont Bang}\right\|_{C^4_{1,\underline{\gamma}_0}(T^\ast M_t)}\lesssim\sqrt{\epsilon}\,a(t)^{1-c^\prime\epsilon^\frac18}\,.
\end{equation}
If, on top of Assumption \ref{ass:init-2+1}, one assumes 
\begin{equation}
\mathrm{dist}_{\underline{\gamma}}(\supp\mathring{f},\supp f_{FLRW})\leq \epsilon^2\,,
\end{equation}
where $\mathrm{dist}_{\underline{\gamma}}$ denotes the metric on $T^\ast M$ induced by the Riemannian metric $\underline{\gamma}$, then one also has
\begin{equation}
\mathrm{dist}_{\underline{\gamma}}(\supp f_{\mathrm{Bang}},\supp f_{FLRW})\lesssim \epsilon^\frac74\,.
\end{equation}}
Finally, $(\M,\g)$ is geodesically past incomplete, forming a stable crushing singularity as $t\downarrow 0$ at which the Kretschmann scalar $\mathcal{K}=\Riem[\g]_{\alpha\beta\gamma\delta}\,\Riem[\g]^{\alpha\beta\gamma\delta}$ exhibits stable blow-up:
\begin{equation}\label{eq:blowup-2+1}
\|a^4\,\lvert k\rvert_g^2-K_\textrm{\normalfont Bang}^2\|_{C^0_\gamma(M_t)}+\left\|a^8\,\mathcal{K}-6(8\pi)^2\,(\Psi_\textrm{\normalfont Bang}+C)^4\right\|_{C^0_{\gamma}(M_t)}\lesssim \epsilon\,a(t)^{1-c\epsilon^\frac18}\,.
\end{equation}
\end{theorem}
As explained after Lemma \ref{lem:scale-factor-2+1}, if $M\cong \S^2$, the FLRW solutions admit a time reflection symmetry so that a fully analogous stability statment holds toward the future and, consequently, that these solutions are globally stable.

\begin{proof}
On the bootstrap interval $(t_{Boot},t_0]$, the improved estimate \eqref{eq:bs-imp} follows directly for all terms in $\mathcal{H}$ except for $\|G-\gamma\|_{H^{10}_G(M_t)}$ by inserting the improved energy estimates obtained in Corollary \ref{cor:en-imp} into the near-coercivity estimates in Lemma \ref{lem:near-coerc}. Regarding the metric term, the preceding argument gives
\[\|N\|_{H^{10}_G(M_t)}^2+\|\hat{k}\|_{H^{10}_G(M_t)}^2\lesssim\epsilon^\frac{15}4\,a^{-c\,\epsilon^\frac18}\,.\]
Thus the improved bound for $\|G-\gamma\|_{H^{10}_G(M_t)}$ follows from \eqref{eq:G-norm-sup-est}. In particular, Lemma \ref{lem:near-coerc} also implies
\[\|\rho-\rho^{Vl}\|_{H^{10}_G(M_t)}+\|S^{Vl,\parallel}\|_{H^{10}_G(M_t)}+\|\j^{Vl}\|_{H^{10}_G(M_t)}\lesssim \epsilon^\frac{7}4\,a(t)^{-c\,\epsilon^\frac18}\,.\]
Using the improved bound on $\|G-\gamma\|_{H^{10}_G(M_t)}$, one straightforwardly computes that all bounds extend to the respective Sobolev spaces with respect to the reference metric $\gamma$ after updating constants. The bounds for $\mathcal{C}$ in \eqref{eq:bs-imp} then follow by Sobolev embedding and similarly changing back from $\gamma$ to $G$, proving the first estimate in \eqref{eq:bs-imp}.\\

Recalling the low order Vlasov bound in \eqref{eq:bs-imp} from \eqref{eq:APVlasov-2+1}, this improves the bootstrap assumptions; see Assumption \ref{ass:bs}. Along with the momentum support bounds from Lemma \ref{lem:APMom-2+1}, one also checks that all continuation criteria from Section \ref{subsec:lwp-2+1} hold as $t\downarrow t_{Boot}$. Hence, the solution can be extended beyond $t=t_{Boot}$ and consequently to $(M_t)_{t\in(0,t_0]}$ now that the bootstrap argument has been completed and \eqref{eq:bs-imp} holds throughout.\\

The remainder of the theorem largely follows as in \cite[Theorem 15.1]{Rodnianski2014} and \cite[Theorem 8.1 \change{and Corollary 8.3}]{FU25}. \change{While more details can be found in these works, we demonstrate the spirit of these arguments by proving \eqref{eq:asymp-Vlasov} at order $0$ with \eqref{eq:asymp-quiesc} already proven; the other asymptotic bounds are obtained similarly and \eqref{eq:blowup-2+1} follows from these by straightforward computations. \\
The rescaled Vlasov equation \eqref{eq:REEqVlasov} can be rewritten as follows, explicitly computing $\A_jf_{FLRW}$ in the third \changereport{term}:
\begin{align*}
\numberthis\label{eq:REEqVlasovVTD}\del_tf=&\,-N\,a^{-1}\,\frac{v^j}{v^0}\,\A_jf-a^{-1}\,\frac{v^j}{v^0}\left[\A_j(f-f_{FLRW})+2\mathcal{F}^{\prime}\,(\lvert v\rvert_\gamma)^2\,\left(\Gamma[G]_{ij}^m-\Gamma[\gamma]^m_{ij}\right)(\gamma^{-1})^{jl}\,v_m\,v_l\right]\\
&\,+a^{-1}\,v^0\,\nabla^{\sharp j}N\,\B_jf\,.
\end{align*}
Now applying \eqref{eq:bs-imp} and \eqref{eq:APMom-2+1} throughout, as well as \eqref{eq:asymp-quiesc} for the third line, we obtain
\[\left\lvert\frac{d}{dt}\left(\langle v\rangle_\gamma f\right)\right\rvert\lesssim\,a^{-1-c\epsilon^\frac18}\,.\]
and consequently, for any $s,t\in(0,t_0]$,
\[\sup_{x,v\in TM}\langle v\rangle_\gamma\lvert f(t,x,v)-f(s,x,v)\rvert\lesssim \sqrt{\epsilon}\left(a(t)^{1-c\epsilon^\frac18}-a(s)^{1-c\epsilon^\frac18}\right)\,.\]}

Hence, \change{$f(t,\cdot,\cdot)$ }uniformly converges to a limit $f_\textrm{\normalfont Bang}$ in \change{$C^0_{1,\underline{\gamma}_0}(T^\ast M)$ }as $t\downarrow 0$ and \eqref{eq:asymp-Vlasov} holds at order $0$. Note that we lose one order of control on the Vlasov distribution compared to \eqref{eq:bs-imp} since, when applying \eqref{eq:REEqVlasovVTD} and higher order analogues, we need to estimate $f-f_{FLRW}$ at order $\ell+1$ to control $f_\textrm{\normalfont asymp}$ at order $\ell$. Further, note that evaluating \eqref{eq:asymp-Vlasov} at $t=t_0$ yields \eqref{eq:Vlasov-control}.
\end{proof}

We formulate the asymptotic behaviour of the spatial metric $g$ in \eqref{eq:asymp-quiesc} to be consistent with the FLRW Big Bang stability works \cite{Rodnianski2018, Speck2018, FU25}. This expression is in line with the formal solution 
\[(g_{VTD})_{ij}=a(t)^2\,\left\{\exp\left(2\int_t^{t_0}a(s)^{-2}\,ds\cdot \hat{K}_{\normalfont VTD}\right)\odot H_{\normalfont VTD}\right\}_{ij}\]
that one obtains by solving the velocity term dominated equations; see also \cite[Remark 19.2]{Speck2018}. Alternatively, one can renormalise the spatial metric as in \cite[Definition 6]{GPR23} such that the resulting metric converges toward the Big Bang. 
\begin{corollary}\label{cor:renorm}
\changereport{Let $(\M,\g,\phi,\f)$ be as in Theorem \ref{thm:main-2+1} and consider, for $t\in(0,t_0]$ and $x\in (M_t)$,
\begin{equation*}
\change{\mathcal{W}_{t,x}: T_\cdot M_t\longrightarrow T_x M_t,\qquad \mathcal{W}_{t,x}(w)}=\exp\left[-\int_t^{t_0}a(s)^{-2}\,ds\cdot (\hat{k}^\sharp)_{t,x}\right]\odot w\,
\end{equation*}
as well as the family of renormalised metrics $\mathcal{G}_t$ given by
\change{\[\mathcal{G}_{t,x}:T_x M_t\times T_x M_t\rightarrow\R,\quad \mathcal{G}_{t,x}(X,Y)=a(t)^{-2}\,\left(g_{t,x}(\mathcal{W}_{t,x}(X),\mathcal{W}_{t,x}(Y))\right).\]}}
Then, there exists a symmetric $(0,2)$-tensor field $\mathcal{G}_{\normalfont Bang}$ on $M$ such that the following bounds hold.
\begin{align}
\label{eq:control-renorm}\|\mathcal{G}_{\normalfont Bang}-\gamma\|_{C^7_\gamma(M)}\lesssim&\,\epsilon\,,\\
\label{eq:asymp-renorm}\|\mathcal{G}-\mathcal{G}_{\normalfont Bang}\|_{C^6_\gamma(M_t)}\lesssim&\, \epsilon\,a(t)^{1-c^\prime\,\epsilon^\frac18}\,.
\end{align}
Moreover, $\hat{K}_{\normalfont Bang}$ is self-adjoint with respect to $\mathcal{G}_{\normalfont Bang}$.
\end{corollary}
\change{\begin{proof}
The asymptotics in \eqref{eq:asymp-renorm} are proven identically to the $(3+1)$-dimensional analogue in \cite[Corollary 8.4]{FU25}, using the evolution equations \eqref{eq:REEqG-2+1} and \eqref{eq:REEqk} along with Theorem \ref{thm:main-2+1} to obtain
\[\lvert \del_t\mathcal{G}\rvert_\gamma\lesssim\epsilon\,a^{-1-c\,\epsilon^\frac18}(\lvert\mathcal{G}\rvert_\gamma+1)\]
and then extracting a Big Bang limit of $\mathcal{G}$ as for the Vlasov distribution above. The bound \eqref{eq:control-renorm} follows directly from \eqref{eq:asymp-renorm} and the initial data assumption, and $\hat{K}_{\normalfont Bang}$ is self-adjoint with respect to $\mathcal{G}_{\normalfont Bang}$ since $\hat{k}^\sharp$ is self-adjoint with respect to $G$ and, consequently, $\mathcal{G}$ on any $M_t$ for $t\in(0,t_0]$.
\end{proof}}

\change{\begin{remark}[Asymptotics of the constraint equations] Recall that \eqref{eq:asymp-Ham} is the limit of the renormalised Hamiltonian constraint. Unfortunately, one cannot obtain a similar limit for the momentum constraint since the Levi-Civita connection associated to the spatial metric degenerates along with the metric itself, preventing a limit on the divergence of the shear, along with causing degeneracies in $\j^{Vl}$. While it might be possible to recast both of these terms with the help of $\mathcal{G}$, this would require tracking eigendirections of the shear along the evolution. Even works that use such frames explicitly, including \cite{GPR23}, lose control of this frame towards the past, preventing an asymptotic momentum constraint. In light of \cite{FG24}, where one can ensure that a Gaussian approximate VTD solution satisfies the momentum constraint given Einstein nonlinear scalar-field data on the Big Bang hypersurface unlike for the VTD solution in CMCTC gauge, this might be an issue with the foliation itself. We leave this question open to future work. 
\end{remark}}

\begin{remark}[Concentration of Vlasov particle velocities]\label{rem:Vlasov-contra-2+1}
\change{Note that \eqref{eq:asymp-Vlasov} implies that the Vlasov distribution is asymptotically velocity term dominated along with the other solution variables if viewed on the co-mass shell. However, as in \cite[Remark 8.5]{FU25}, the picture changes when one tries to analyse the velocity of Vlasov particles rather than their (conjugate) momenta and views the distribution function on the mass shell. In short, \enquote{lifting} to the mass shell means one must renormalize the spatial metric as indicated by Corollary \ref{cor:renorm}. The resulting velocity characteristics are, consequently, dominated by their VTD component
\[Q_{t,x}(q)=\exp\left[-2\int_t^{t_0}a(s)^{-2}\,ds\cdot\left(\hat{K}_{\normalfont Bang}\right)_x\right]\odot a(t)^{-2}q\,.\]
Thus and in line with Corollary \ref{cor:renorm}, it is useful to consider the FLRW renormalised particle velocities $w^i=a(t)^{-2}q^i$ and the a posteriori expansion-renormalised velocities
\[W_{t,x}(w)=\exp\left[2\int_t^{t_0}a(s)^{-2}\,ds\cdot\left(\hat{K}_{\normalfont Bang}\right)_x\right]\odot w.\]
Indeed, one can prove by similar means as in Theorem \ref{thm:main-2+1} that, for
\[f^\sharp_\textrm{\normalfont asymp}(t,x,w):=f^\sharp(t,x,W^{-1}_{t,x}(w))\,,\]
the support of $f^\sharp_\textrm{\normalfont asymp}(t,\cdot,\cdot)$ also remains compact and that this renormalised function converges to an asymptotic footprint $f^\sharp_{\textrm{\normalfont Bang}}$. In particular, one has
\[\supp f^\sharp(t,x,\cdot)=\left\{W_{t,x}(w)\,\vert\,w\in\supp f^\sharp_\textrm{\normalfont asymp}(t,x,\cdot)\right\}\] 
by definition of $f^\sharp_{\textrm{\normalfont asymp}}$. When approaching the Big Bang, particle velocities thus concentrate in directions $w$ where $\lim_{t\to 0}W_{t,x}(q)$ remains bounded.

To identify these directions, first consider the massive case ($m=1$) where particles can be initially at rest, i.e., consider $x\in M$ such that $\mathring{f}(x,0)\neq 0$.  In the case where $(\hat{K}_\textrm{\normalfont Bang})_x=0$, $f^\sharp_\textrm{\normalfont asymp}$ simply agrees with the (FLRW renormalised) initial distribution and $f^\sharp$ converges to an asymptotic profile $f_\textrm{\normalfont Bang}$ that is close to the FLRW distribution. Otherwise, the tracefree and self-adjoint operator $(\hat{K}_\textrm{\normalfont Bang})_x$ has precisely one positive and one negative eigendirection. Denoting these eigenspaces by $E_{x,\pm}$, it follows that $\lim_{t\to 0}W_{t,x}(q)=0$ is satisfied if one takes $w\in E_{x,-}$. Otherwise, $\lim_{t\to 0}W_{t,x}(q)=\infty$ holds and $w$ lies outside of $\supp f^\sharp(t,x,\cdot)$ for small enough $t>0$. Thus, \eqref{eq:asymp-Vlasov} implies the following pointwise limit for any $w\in T^\ast_xM$. 
\[\lim_{t\to 0}f^\sharp(t,x,q)=f^\sharp_{\textrm{\normalfont Bang}}(x,0)\,\chi_{E_{x,-}}(w(q))\,.\]
If $f^\sharp_\textrm{\normalfont Bang}(\cdot,0)\neq 0$, which holds for sufficiently small $\epsilon>0$ due to \eqref{eq:Vlasov-control}, the velocity support thus concentrates toward $E_{x,-}$ and any element of $E_{x,-}$ lies in the support.\\

For $m=0$, or more generally when the initial particle distribution vanishes in the neighbourhood of the zero section for some $x\in M$, $f^\sharp_\textrm{\normalfont Bang}$ may also vanish at the origin for some $x\in M$. In that case, for any $q\in T^\ast_xM$, the Vlasov distribution vanishes both if $W_{t,x}(q)$ becomes too large or too small. Denoting the orthogonal projectors onto $E_{x,\pm}$ by $\mathbb{P}_{\pm}$, one has
\[\sup_{q\in \supp f^\sharp(t,x,\cdot)}\lvert \mathbb{P}_+w(q)\rvert_\gamma \to 0\quad \text{as}\quad t\downarrow 0\]
as well as
\[\inf_{q\in \supp f^\sharp(t,x,\cdot)}\lvert \mathbb{P}_-w(q)\rvert_\gamma \to \infty\, \quad \text{as}\quad t\downarrow 0\]
while the momentum support still remains bounded for any $t>0$; see \eqref{eq:APMom-2+1}. In other words, particle velocities concentrate in the direction of $E_{x,-}$ toward the Big Bang, but become arbitrarily large.\\

Finally, one expects that $\hat{K}_{Bang}=0$ is a non-generic case, and thus that velocity concentration is generic. However, this is not immediately apparent from Theorem \ref{thm:main-2+1}, as we do not obtain a scattering result. However, nearby Kasner solutions must be contained in the solutions covered by Theorem \ref{thm:main-2+1} for which $\hat{K}_{Bang}$ only vanishes if it coincides with the FLRW solution. Since the construction above also ensures that the footprint states depend on initial data continuously, it follows that the set of solutions for which velocities concentrate is, at least, of positive measure. In light of the scattering result \cite{Li24} for the linearized Einstein scalar-field system, it further is reasonable to conjecture that $\hat{K}_{Bang}$ vanishing entirely is truly non-generic, but this still remains open even in the full Einstein scalar-field setting.}

\end{remark}

Additionally, we can now state the full version of Corollary \ref{cor:u1-intro} and show how Theorem \ref{thm:main-2+1} can be translated to the polarized $U(1)$-symmetric vacuum setting:

\begin{corollary}[Past stability of non-extremal polarized $U(1)$-symmetric $(3+1)$-dimensional vacuum solutions]\label{cor:u1}
Consider $3+1$ vacuum spacetimes of the form
\begin{equation}\label{eq:u1-ref-cor}
\left(\check{M}=(0,t_0]\times M\times\S^1,\ \check{g}_{\text{ref}}=e^{-2\sqrt{4\pi}\,\phi_{\textrm{\normalfont ref}}}\,(-dt^2+a(t)^2\,\gamma)+e^{2\sqrt{4\pi}\,\phi_{\textrm{\normalfont ref}}}\,\change{(dx^3)}^2\right),
\end{equation}
where \change{$\phi_{\textrm{\normalfont ref}}$ }is a spatially homogeneous function with $\del_t\phi_{\textrm{\normalfont ref}}=C\,a^{-2}$ and $C<0$ and $(M,\gamma)$ is a surface of constant curvature; see Remark \ref{rem:u1-ref} and Appendix \ref{subsec:ref-u1}. On a constant time slice $\{t=t_0\}$ of \eqref{eq:u1-ref-cor}, take polarized $U(1)$-symmetric initial data $(M\times\S^1,\mathring{\tilde{g}},\mathring{\tilde{k}})$ that is close to that of \eqref{eq:u1-ref-cor} in the sense of Assumption \ref{ass:init-2+1} and the initial data correspondence in Remark \ref{rem:u1-init}. Additionally, assume without loss of generality that the induced mean curvature $\mathring{k}$ on $M_{t_0}$ has constant trace; see \eqref{eq:u1-init-corresp}.\\

Then, such a solution admits a past maximal global hyperbolic development $(\check{M},\check{g})$ that is polarized $U(1)$-symmetric in the sense that the transported coordinate derivative \change{$\del_{x^3}$ }on $\S^1$ acts as a Killing vector field. Writing 
\begin{equation}\label{eq:u1-shape}
\check{g}=e^{-2\sqrt{4\pi}\,\phi}\,(-n^2dt^2+g)+e^{2\sqrt{4\pi}\,\phi}\,\change{(dx^3)}^2
\end{equation}
for $\phi: I\times M\rightarrow \R$, $(g,k,n,\phi)$ satisfy the asymptotic bounds of Theorem \ref{thm:main-2+1} with $f\equiv 0$ along the constant time foliation $(M_t\times \S^1)_{t\in(0,t_0]}$. Moreover, let $\tilde{k}$ be the second fundamental form with respect to the foliation $(M_t\times\S^1, \tilde{g}=\check{g}\vert_{M_t\times\S^1})_{t\in(0,t_0]}$ and let $\check{\mathcal{K}}$ denote the Kretschmann scalar with respect to $\check{g}$. Then, there exist positive continuous functions $Y_\textrm{\normalfont Bang}, Z_\textrm{\normalfont Bang}$ on $M$ such that the following bounds hold.
\begin{subequations}
\begin{align}
\|e^{-2\sqrt{4\pi}\,\phi}\,a^4\,\lvert \tilde{k}\rvert_{\tilde{g}}^2-Y_\textrm{\normalfont Bang}\|_{C^0(M_t)}\lesssim&\,\epsilon\,a(t)^{2-c\epsilon^\frac18}\,,\label{eq:u1-asymp-k}\\
\|e^{-4\sqrt{4\pi}\,\phi}\,a^8\,\check{\mathcal{K}}-Z_\textrm{\normalfont Bang}\|_{C^0(M_t)}\lesssim&\,\epsilon\,a(t)^{4-c\epsilon^\frac18}\,.\label{eq:u1-asymp-K}
\end{align}
\end{subequations}
In particular, the $(\check{M},\check{g})$ is $C^2$-inextendible toward the past.
\end{corollary}
\begin{remark}[Curvature blow-up]\label{rem:u1-curv-blow-up}
Applying the Friedman equation \eqref{eq:Friedman-2+1}, note that one has 
\[e^{\sqrt{4\pi}\,\phi_{\text{ref}}(t)}=\mathfrak{c}_\kappa(t)\,e^{\sqrt{4\pi}a(t_0)}\,a(t)^{-1}\,,\] where $\mathfrak{c}_\kappa\in C^\infty([0,t_0])$ is identical to $1$ for $\kappa=0$ and converges to a positive constant approaching $t=0$ otherwise. Combining this with \eqref{eq:asymp-phi-2+1}, this implies that $\lvert\tilde{k}\rvert_g^2$ and $\check{\mathcal{K}}$ blow up at order close to $a^{-6}$ and $a^{-12}$ respectively, up to a factor of $a^{\pm c\,\epsilon}$. Using the identities \eqref{eq:u1-mghd-translate} below, the components of $\check{g}$ can be bounded similarly. 
\end{remark}
\begin{proof}
Following the initial data correspondence in Remark \ref{rem:u1-ref}, Theorem \ref{thm:main-2+1} yields the past maximal $(2+1)$-dimensional Einstein scalar field solution $(\M,\g,\phi)$ that corresponds to $(3+1)$-dimensional vacuum solution
\[\left(\check{M}=\M\times \S^1, \check{g}=e^{-2\sqrt{4\pi}\,\phi}\,\g+e^{2\sqrt{4\pi}\,\phi}\,(dx^3)^2\right)\,\]
launched by $(M\times \S^1,\mathring{\tilde{g}},\mathring{\tilde{k}})$. Since $\g_{0i}=0$ for $i=1,2$, we can write
\[\check{g}=-n^2\,e^{-2\sqrt{4\pi}\,\phi}\,dt^2+\tilde{g},\]
where $\tilde{g}$ is a time-dependent Riemannian metric on $M\times\S^1$. Additionally writing the second fundamental form and shear with respect to this foliation as $\tilde{k}$ and $\hat{\tilde{k}}$, as well as $\psi=n^{-1}\,\del_t\phi$, the components of $\tilde{g}$ and $\tilde{k}$ are as follows.
\begin{subequations}\label{eq:u1-mghd-translate}
\begin{align}
\tilde{g}_{ij}=&\,e^{-2\sqrt{4\pi}\,\phi}\,g_{ij},\quad \tilde{g}_{33}=e^{2\sqrt{4\pi}\,\phi}\,,\\
\tilde{k}_{ij}=&\,e^{-\sqrt{4\pi}\,\phi} \left(k_{ij}+{\sqrt{4\pi}}\,\psi\, g_{ij}\right),\quad \tilde{k}_{33}=-\sqrt{4\pi}\,e^{3\sqrt{4\pi}\,\phi}\,\psi\,,\\
\tilde{g}_{13}=&\,\tilde{g}_{23}=\tilde{k}_{13}=\tilde{k}_{23}=0\,.
\end{align}
\end{subequations}

From this, we directly compute, using the Friedman equation \eqref{eq:Friedman-2+1}, that
\begin{align*}
\lvert \tilde{k}\rvert_{\tilde{g}}^2=&\,e^{2\sqrt{4\pi}\,\phi}\left[\left\lvert\hat{k}+\left(\frac{\tau}2+\sqrt{4\pi}\,\psi\right)g\right\rvert_g^2+4\pi\,\psi^2\right]\\
=&\,e^{2\sqrt{4\pi}\,\phi}\,a^{-4}\left[\lvert \hat{k}\rvert_G^2+2\left(\sqrt{4\pi}\,C-\sqrt{4\pi\,C^2-\kappa\,a^2}+\sqrt{4\pi}\,\Psi\right)^2+4\pi\,(C+\Psi)^2\right]\,.
\end{align*}
Now, \eqref{eq:u1-asymp-k} follows from \eqref{eq:asymp-quiesc} and \eqref{eq:asymp-Ham}, defining
\begin{align*}
Y_\textrm{\normalfont Bang}:=&\,\hat{K}_\textrm{\normalfont Bang}\cdot\hat{K}_\textrm{\normalfont Bang}+8\pi\,(\Psi_\textrm{\normalfont Bang}+2\,C)^2+4\pi\,(\Psi_\textrm{\normalfont Bang}+C)^2\\
=&\,16\pi\,C\,(\Psi_\textrm{\normalfont Bang}+2\,C)+4\pi\,(\Psi_\textrm{\normalfont Bang}+C)^2\\
=&\,4\pi\,(\Psi_\textrm{\normalfont Bang}+3\,C)^2\,.
\end{align*}

After some lengthy computations, for which we defer to Section \ref{subsec:curvature-u1}, and again applying the Friedman equation \eqref{eq:Friedman-2+1} \change{in the second line below}, one also sees that the Kretschmann scalar $\check{\mathcal{K}}$ takes the following form.
\begin{align*}
\numberthis\label{eq:K-formula}e^{-4\sqrt{4\pi}\,\phi}\,\check{\mathcal{K}}=&\,32\cdot 4\pi\,\psi^2\left(\frac{\tau}2+\sqrt{4\pi}\,\psi\right)^2+8\cdot 4\pi\,\psi^2\,\lvert \hat{k}\rvert_g^2+\langle \textrm{Err}\rangle\\
=&\,32\pi\,a^{-8}\,(C+\Psi)^2\left[4\left(\sqrt{4\pi}\,C-\sqrt{4\pi\,C^2-\kappa \,a^2}+\sqrt{4\pi}\,\Psi\right)^2+\lvert\hat{k}\rvert_G^2\right]+\langle \textrm{Err}\rangle\,.
\end{align*}
The error term $\langle \textrm{Err}\rangle$ can be bounded as below using Theorem \ref{thm:main-2+1}.
\[\|\langle \textrm{Err}\rangle\|_{C^0_\gamma(M_t)}\lesssim \epsilon a(t)^{-4-c\epsilon^\frac18}\,.\] 
The bound \eqref{eq:u1-asymp-K} again follows from \eqref{eq:asymp-quiesc} and \eqref{eq:asymp-Ham} for
\begin{align*}
Z_\textrm{\normalfont Bang}:=&\,32\pi\left(\Psi_\textrm{\normalfont Bang}+C\right)^2\left(16\pi(\Psi_\textrm{\normalfont Bang}+2\,C)^2+\hat{K}_\textrm{\normalfont Bang}\cdot\hat{K}_\textrm{\normalfont Bang}\right)\\
=&\,256\pi^2\left(\Psi_\textrm{\normalfont Bang}+C\right)^2(\Psi_\textrm{\normalfont Bang}+2\,C)(\Psi_\textrm{\normalfont Bang}+4\,C)\,.
\end{align*}
Since \eqref{eq:u1-asymp-K} shows that the Kretschmann scalar blows up as noted in Remark \ref{rem:u1-curv-blow-up}, $C^2$-inextendibility follows immediately.
\end{proof}
\begin{remark}[The induced foliation is not CMC]\label{rem:not-cmc}
As \eqref{eq:u1-mghd-translate} shows, the foliation $(M_t\times \S^1)_{t\in(0,t_0)}$ is no longer a constant mean curvature foliation. More precisely, one has
\[\text{tr}_{\tilde{g}}\tilde{k}=e^{\sqrt{4\pi}\,\phi}\,(\tau+\sqrt{4\pi}\,\psi)\]
and thus
\[e^{-\sqrt{4\pi}\,\phi_{\text{ref}}(t)}\,\lvert\nabla\text{tr}_{\tilde{g}}\tilde{k}\rvert_{(\gamma+\change{(dx^3)}^2)}\lesssim \epsilon\,\lvert\log(t)\rvert\]
by \eqref{eq:asymp-phi-2+1}. This slight difference in foliation leads to differences in asymptotic bounds of our Kretsch\-mann scalar in the spatially flat case compared to \cite{RodSpFou20}, even if one were to express $\check{g}$ in a coordinate system in which the lapse converges to $1$.
\end{remark}
\begin{remark}[Instability of extremal Kasner(-like) solutions]\label{rem:extreme}
In principle, one can attempt to redo the argument of Corollary \ref{cor:u1} corresponding to reference solutions with $\del_t\phi_{\text{ref}}=C\,a^{-2},\, C>0$. As noted in Remark \ref{rem:u1-ref}, this corresponds to extremal Kasner-like spacetimes. Since one has
\[\frac{\tau}2+\sqrt{4\pi}\,\del_t\phi_{\text{ref}}=a^{-2}\left(\sqrt{4\pi\,C^2}-\sqrt{4\pi\,C^2+\kappa\,a^{2}}\right)=\O{\kappa\,a^{-1}}\,,\]
the reference solution does not exhibit curvature blow-up in the spatially flat case by the first line of \eqref{eq:K-formula}. For $\kappa\neq 0$, since $e^{-4\sqrt{4\pi}\,\phi_{\text{ref}}}$ is now proportional to $a^{-4}$ at leading order, the Kretschmann scalar blows up at order \change{$\O{a^{-2}}$}. However, for the development of nearby initial data, \eqref{eq:K-formula} becomes
\begin{align*}
e^{-4\sqrt{4\pi}\,\phi}\,\check{\mathcal{K}}=32\pi\,a^{-8}\left(C+\Psi\right)^2\left(\lvert \hat{k}\rvert_G^2+\Psi^2\right)+\langle\text{lower order terms}\rangle\,.
\end{align*}

Consequently, if $\Psi_\textrm{\normalfont Bang}$ and $\hat{K}_\textrm{\normalfont Bang}$ do not vanish,\footnote{By \eqref{eq:asymp-Ham} and since $\hat{K}_\textrm{\normalfont Bang}$ is self-adjoint, if $\epsilon>0$ is sufficiently small, $\Psi_\textrm{\normalfont Bang}(x)=0$ is satisfied for some $x\in M$ if and only if $\hat{K}_\textrm{\normalfont Bang}(x)=0$.} the Kretschmann scalar blows up at an order above $a^{-4+c\epsilon^\frac18}$ for some $c>0$. While instability of these extremal solutions is to be expected simply by considering the Kasner family itself, this seems to indicate that this instability is generic. To make this rigorous, one would need to consider the scattering theory for initial data near that of extremal solutions in detail to find positive measure initial data sets for which $\Psi_\textrm{\normalfont Bang}$ does not vanish or only vanishes on a null set.
\end{remark}


\section{Appendix: Formulas for polarized $U(1)$-symmetric vacuum spacetimes}\label{sec:app}

\subsection{Kasner-like form of reference metrics for $\kappa\neq 0$}\label{subsec:ref-u1}
In this section, we collect the necessary transformations to bring the reference solutions described by \eqref{eq:u1-FLRW} into Kasner-like form
\[-dS^2+b(S)^2\,\gamma+b_3(S)^2\,\change{(dx^3)}^2\]
for $\kappa\neq 0$, as discussed in Remark \ref{rem:u1-ref}. Recall from Remark \ref{rem:u1-ref} that, for some $t_1>0$, the time variable $S\equiv S(t)$ is determined by
\begin{equation}\label{eq:S}
\change{\frac{dS}{dt}}=\exp\left(\sqrt{4\pi}\,C\int_t^{t_1}a(s)^{-2}\,ds\right),\ S(0)=0\,.
\end{equation}
Note that $S:[0,T/2)\rightarrow [0,\infty)$ is invertible and we can view $t\equiv t(S)$ as a function of $S$. 

The non-negative scale factors $b(S)$ and $b_3(S)$ are then given by
\begin{subequations}\label{eq:b-b3}
\begin{align}
b(S)^2=&\,a(t(S))^2\,\exp\left(2\sqrt{4\pi}\,C\,\int_{t(S)}^{t_1}a(s)^{-2}\,ds\right)\,,\\
b_3(S)^2=&\,\exp\left(-2\sqrt{4\pi}\,C\,\int_{t(S)}^{t_1}a(s)^{-2}\,ds\right)\,.
\end{align}
\end{subequations}
Without loss of generality, we assume $\kappa=\pm 1$ for this section; the remaining cases can be reduced to these by an appropriate rescaling of $\gamma$, which can be achieved by choosing $t_1>0$ appropriately. 

To obtain the asymptotic behaviours of $b$ and $b_3$ as $S$ approaches $0$, we will first compute expressions for these functions in terms of $t(S)$ and then expand $S$ in terms of $t$.\\

First, rearranging \eqref{eq:Friedman-2+1} with ${\rho}_{FLRW}=\mathfrak{p}_{FLRW}=0$ and integrating over the interval $(0,t)$, we obtain
\begin{align*}
t=&\,\int_0^{a(t)}\left(4\pi\,C^2\,x^{-2}-\kappa\right)^{-\frac12}dx\\
=&\,2\,\sqrt{\pi}\,\lvert C\rvert\,\int_0^{\frac{a(t)}{2\,\sqrt{\pi}\,\lvert C\rvert}}\frac{1}{\left(y^{-2}-\kappa\right)^{-\frac12}}\,dy\\
=&\,\sqrt{\pi}\,\lvert C\rvert\,\int_0^{\frac{a(t)}{2\,\sqrt{\pi}\,\lvert C\rvert}}\frac{2y}{\sqrt{1-\kappa\,y^2}}\,dy\\
=&\,-\kappa\left(\sqrt{4\pi\,C^2-\kappa\,a(t)^2}-2\sqrt{\pi}\,\lvert C\rvert\right)\,.
\end{align*}
After rearranging, this yields
\begin{equation}\label{eq:scale-factor-explicit}
a(t)^2=4\sqrt{\pi}\,\lvert C\rvert\,t-\kappa\,t^2=t\,(4\sqrt{\pi}\,\lvert C\rvert\,-\kappa\,t)\,.
\end{equation}
Below, we write $D:=(4\sqrt{\pi}\,\lvert C\rvert)^{-1}$.\\
We now solve the inital value problem \eqref{eq:S}, depending on the signs of $C$ and $\kappa$. Note that \eqref{eq:scale-factor-explicit} implies
\begin{align*}
\numberthis\label{eq:u1-conf-formula}\exp\left(\sqrt{4\pi}\,C\int_t^{t_1}a(s)^{-2}\,ds\right)=&\,\exp\left(\sqrt{4\pi}\,C\,\left.D\,\left[\log(s)-\log(D^{-1}-\kappa\,s)\right]\right\vert_{s=t}^{t_1}\right)\,\\
=&\,\left[\frac{t}{D^{-1}-\kappa\,t}\,\frac{D^{-1}-\kappa\,t_1}{t_1}\right]^{-\frac{\mathrm{sgn}(C)}2}\,.
\end{align*}
For $C<0$ and $\kappa=-1$, one thus obtains the following for $K_{t_1}:=t_1^{\frac{\mathrm{sgn}(C)}2}\,\left(D^{-1}-\kappa\,t_1\right)^{-\frac{\mathrm{sgn}(C)}2}$.
\begin{align*}
S(t)=&\,K_{t_1}\,\int_0^{t}\sqrt{\frac{s}{D^{-1}+ s}}\,ds=\frac1D\int_1^{1+D\,t}\sqrt{1-\frac1u}\,du\\
=&\,2\,K_{t_1}\,D^{-1}\,\int_0^{\mathrm{arcosh}(\sqrt{1-D\,t})}\sinh(x)^2\,dx\\
=&\,K_{t_1}\,D^{-1}\,\left[\sinh(x)\,\cosh(x)-x\right]\big\vert_{x=0}^{\mathrm{arsinh}\left(\sqrt{D\,t}\right)}\\
=&\,K_{t_1}\,D^{-1}\,\left(\sqrt{D\,t}\,\sqrt{D\,t+1}-\mathrm{arsinh}(\sqrt{D\,t})\right)\,.
\end{align*}
Note that both summands are analytic in $\sqrt{D\,t}$ near $\sqrt{D\,t}=0$. Thus, computing the Taylor expansion of both summands in $\sqrt{D\,t}$, one observes
\begin{align*}
S(t)=&\,K_{t_1}\,D^{-1}\left((D\,t)^\frac12+\frac12(D\,t)^\frac32+\O{(D\,t)^\frac52}\right)-K_{t_1}\,D^{-1}\left((D\,t)^\frac12-\frac16 (D\,t)^\frac32+\O{(D\,t)^\frac52}\right)\\
=&\,\frac23\,K_{t_1}\,D^{\frac12}\,t^\frac32+\O{t^\frac52}\quad\text{as}\ t\downarrow 0
\end{align*}
For $C<0$ and $\kappa=1$, one analogously computes
\begin{align*}
S(t)=&\,K_{t_1}\,\int_0^t\sqrt{\frac{s}{D^{-1}-s}}\,ds\\
=&\,2K_{t_1}\,D^{-1}\,\int_0^{\mathrm{arccos}(\sqrt{1+D\,t})}\sin(x)^2\,dx\\
=&\,-K_{t_1}\,D^{-1}\,\left(\sqrt{D\,t}\,\sqrt{1-D\,t}+\mathrm{arcsin}\left(\sqrt{D\,t}\right)\right)
=&\,\frac23\,K_{t_1}\,D^\frac12\,t^\frac32+\O{t^\frac52}\quad\text{as}\ t\downarrow 0\,.
\end{align*}
In both cases, this implies
\[t(S)=\left(\frac32\right)^\frac23\,K_{t_1}^{-\frac23}\,D^{-\frac13}\,S^\frac23+\O{S}\,.\]
Plugging this back into \eqref{eq:b-b3}, along with using \eqref{eq:scale-factor-explicit} and \eqref{eq:u1-conf-formula}, implies the following expansions for $b(S)^2$ and $b_3(S)^2$ as $S$ approaches $0$.
\begin{align*}
b(S)^2=&\,a(t(S))^2\exp\left(4\sqrt{\pi}\,C\,\int_{t(S)}^{t_1}a(s)^{-2}\,ds\right)\\
=&\,K^2_{t_1}\,\left(4{\sqrt{\pi}\,\lvert C\rvert}\,t(S)-\kappa\,t(S)^2\right)\frac{t(S)}{4\sqrt{\pi}\,\lvert C\rvert-\kappa\,t(S)}\\
=&\,K^2_{t_1}\,t(S)^2+\O{t(S)^3}\\
=&\,\left(\frac32\right)^\frac43\,K_{t_1}^{\frac23}\,D^{-\frac23}\,S^{\frac43}+\O{S^2}\\
b_3(S)^2=&\,\exp\left(-2\sqrt{4\pi}\,C\,\int_{t(S)}^{t_1}a(s)^{-2}\,ds\right)\\
=&\,K_{t_1}^{-2}\,D\,t(S)^{-1}+\O{1}\\
=&\,\left(\frac32\right)^{-\frac23}\,K_{t_1}^ {-\frac13}\,D^\frac43\,S^{-\frac23}+\O{1}\,.
\end{align*}
Turning now to the case where $C>0$, we first compute for $\kappa=-1$ that
\begin{align*}
S(t)=&\,K_{t_1}\,\int_0^{t}\sqrt{\frac{D^{-1}+s}{s}}\,ds=K_{t_1}\,D^{-1}\,\int_0^{D\,t}\sqrt{1+\frac1u}\,du\\
=&\,2K_{t_1}\,D^{-1}\,\int_0^{\mathrm{arsinh}(\sqrt{D\,t})}\cosh(x)^2\,dx\\
=&\,K_{t_1}\,D^{-1}\left(\sqrt{D\,t}\sqrt{1+D\,t}+\mathrm{arsinh}\left(\sqrt{D\,t}\right)\right)\\
=&\,2K_{t_1}\,\sqrt{\frac{t}D}+\O{t^\frac32}\quad\text{as}\ t\downarrow 0\,.
\end{align*}
For $\kappa=1$, one has
\begin{align*}
S(t)=&\,K_{t_1}\,\int_0^t\sqrt{\frac{D^{-1}-s}{s}}\,ds=D^{-1}\,\int_0^t\sqrt{\frac1u-1}\,du\\
=&\,2K_{t_1}\,D^{-1}\,\int_0^{\mathrm{arcsin}(\sqrt{D\,t})}\sin(x)^2\,ds\\
=&\,K_{t_1}\,D^{-1}\left[\sqrt{D\,t}\sqrt{1-D\,t}+\mathrm{arcsin}(\sqrt{D\,t})\right]\\
=&\,2K_{t_1}\,\sqrt{\frac{t}D}+\O{t^\frac32}\quad\text{as}\ t\downarrow 0\,.
\end{align*}
Thus, as $S$ approaches $0$, one now has
\begin{align*}
b(S)^2=&\,K_{t_1}^2\,\left(4\sqrt{\pi}\,C\,t(S)-\kappa\,t(S)^2\right)\,\frac{4\sqrt{\pi}\,C-\kappa\,t(S)}{t(S)}\\
=&\,16\pi\,C^2\,K_{t_1}+\O{t(S)}\\
=&\,16\pi\,C^2\,K_{t_1}+\O{S^2}\ \text{and}\\
b_3(S)^2=&\,K_{t_1}^{-2}\,\frac{t(S)}{4\sqrt{\pi}\,C+t(S)}=K_{t_1}^{-2}\,\frac{t(S)}{4\sqrt{\pi}\,C}+\O{t(S)^2}\\
=&\,4\pi\,C^2\,K_{t_1}^{-4}\,S^2+\O{S^3}\,.
\end{align*}

\subsection{Christoffel symbols and curvature components}\label{subsec:curvature-u1}

In this section, we collect formulas for the Christoffel symbols and Riemann curvature tensor for spacetime metrics under consideration in Corollary \ref{cor:u1}. Therein, the spacetime metric on $I\times M\times\mathbb{S}^1$ takes the form
\[\check{g}=e^{-2\sqrt{4\pi}\,\phi}\left(-n^2\,dt^2+g\right)+e^{2\sqrt{4\pi}\,\phi}\,\change{(dx^3)}^2\,,\]
where $g$ is a Riemannian metric on $M$. Additionally, if $k$ denotes the second fundamental form with respect to the metric $g$ along $(M_t)_{t\in(0,t_0]}$, $\text{tr}_gk=\tau$ only depends on time. Below, the index $t$ corresponds to the vector field $\del_t$ and we write $\hat{k}_{ij}=k_{ij}-\frac{\tau}2\,g_{ij}$ as well as $\psi=n^{-1}\,\del_t\phi$.\\

The Christoffel symbols $\check{\Gamma}$ with respect to $\check{g}$ take the following form:
\begin{align*}
\check{\Gamma}^{t}_{tt}=&\,\del_tn-\sqrt{4\pi}\,n\,\psi,& \check{\Gamma}^{t}_{ti}=&\,\del_i\log(n)-\sqrt{4\pi}\,\del_i\phi\,,\\
\check{\Gamma}^l_{tt}=&\,n\,g^{lm}\,\del_m n+\sqrt{4\pi}\,n^2\,g^{lm}\,\del_m\phi& \check{\Gamma}^l_{tj}=&\,-n\left(k^l_{\ j}+\sqrt{4\pi}\,\psi\,\I^l_j\right)\,,\\
\check{\Gamma}^{t}_{ij}=&\,-n^{-1}\left(k_{ij}+\sqrt{4\pi}\,\psi\,g_{ij}\right) &\check{\Gamma}^{l}_{ij}=&\,\Gamma^l_{ij}-\sqrt{4\pi}\,\del_i\phi\,\I^l_j-\sqrt{4\pi}\,\del_j\phi\,\I^l_i+\sqrt{4\pi}\,\del_m\phi\,g^{lm}g_{ij}\,,\\
 \check{\Gamma}^3_{t3}=&\,\sqrt{4\pi}\,n\,\psi & \check{\Gamma}^3_{3i}=&\,\sqrt{4\pi}\,\del_i\phi\,,\\
  \check{\Gamma}^{t}_{33}=&\,\sqrt{4\pi}\,n^{-1}\,e^{4\sqrt{4\pi}\phi}\,\psi &
\check{\Gamma}^i_{33}=&\,-\sqrt{4\pi}\,e^{4\sqrt{4\pi}\phi}\,g^{ij}\,\del_j\phi\,,\\
\check{\Gamma}^{3}_{\mu\nu}=&\,\check{\Gamma}^\rho_{3\nu}=\check{\Gamma}^3_{33}=0\,. &&
\end{align*}

With these, one can compute the curvature components listed below. For components in which the index $3$ does not occur, it can simplify computation to view $\check{g}\vert_{I\times M}=e^{-2\sqrt{4\pi}\,\phi}\,\g$ as a conformal transformation and use standard formulas that then relate $\Riem[\check{g}]$ with $\Riem[\g]$. Additionally, we use $\square_{\g}\phi=0$, i.e.,
\[e_0\psi=n^{-1}\,\del_t\psi=\Lap_G\phi+\langle\nabla \log(n),\nabla\phi\rangle_g+{\tau}\,\psi\,,\]
as well as $\Ric[\g]_{\mu\nu}=8\pi\,\nabbar_\mu\phi\,\nabbar_\nu\phi$ and that the Weyl tensor $W[\g]$ vanishes.\\

The Riemann curvature tensor is fully described by the following components:
\begin{align*}
{\Riem[\check{g}]^{t}}_{itj}=&\,{\Riem[\g]^t}_{itj}+\sqrt{4\pi}\,\nabla_i\nabla_j\phi+\sqrt{4\pi}\,n\,k_{ij}\,\psi+4\pi\,\nabla_i\phi\,\nabla_j\phi\\
&\,+\left(-\sqrt{4\pi}\,e_0\psi+\sqrt{4\pi}\,\langle\nabla\log(n),\nabla\phi\rangle_g+4\pi\,\lvert\nabla\phi\rvert_g^2\right)g_{ij}\\
=&\,-\sqrt{4\pi}\left(\sqrt{4\pi}\,\psi+\frac{\tau}2\right)\psi\,g_{ij}+\sqrt{4\pi}\,n\,\hat{k}_{ij}\psi+\left(-\sqrt{4\pi}\,\Lap_G\phi-4\pi\,\lvert\nabla\phi\rvert_g^2\right)g_{ij}\\
&\,+\sqrt{4\pi}\,\nabla_i\nabla_j\phi\,,\\
{\Riem[\check{g}]^t}_{ijl}=&\,{\Riem[\g]^t}_{ijl}+\left(\sqrt{4\pi}\,n^{-1}\,\nabla_j\psi-\sqrt{4\pi}\,k^r_{\ l}\,\nabla_r\phi-8\pi\,n\,\psi\,\nabla_l\phi\right)g_{ij}\\\
&\,-\left(\sqrt{4\pi}\,n^{-1}\,\nabla_j\psi-\sqrt{4\pi}\,k^r_{\ j}\,\nabla_r\phi-8\pi \,n\,\psi\,\nabla_j\phi\right)g_{il}\\
=&\,\left(\sqrt{4\pi}\,n^{-1}\,\nabla_l\psi-\sqrt{4\pi}\,k^r_{\ l}\,\nabla_r\phi\right)g_{ij}-\left(\sqrt{4\pi}\,n^{-1}\,\nabla_j\psi-\sqrt{4\pi}\,k^r_{\ j}\,\nabla_r\phi\right)g_{il}\,,\\
{\Riem[\check{g}]^{12}}_{12}=&\,{\Riem[\g]^{12}}_{12}+\left(\sqrt{4\pi}\,\nabla^1\nabla_1\phi-4\pi\,\nabla^1\phi\,\nabla_1\phi-\sqrt{4\pi}\,{k}^1_{\ 1}\,\psi\right)\\
&\,+\left(\sqrt{4\pi}\,\nabla^2\nabla_2\phi+4\pi\,\nabla^2\phi\,\nabla_2\phi+\sqrt{4\pi}\,{k}^2_{\ 2}\,\psi\right)-4\pi\left(-\psi^2+\lvert\nabla\phi\rvert^2\right)\\
=&\,2\sqrt{4\pi}\,\psi\left(\frac{\tau}2+\sqrt{4\pi}\,\psi\right)+\sqrt{4\pi}\,\Lap_G\phi\,,\\
{\Riem[\check{g}]^3}_{t3t}=&\,-\sqrt{4\pi}\,n^2\,e_0\psi-8\pi\,(n\,\psi)^2+\sqrt{4\pi}\,n\,\langle\nabla\phi,\nabla n\rangle_g+4\pi \,n^2\,\lvert \nabla\phi\rvert_g^2\\
=&-2\sqrt{4\pi}\,n^2\,\psi\,\left(\frac{\tau}2+\sqrt{4\pi}\,\psi\right)-\sqrt{4\pi}\,n^2\,\Lap_G\phi+\sqrt{4\pi}\,n\,\langle\nabla\phi,\nabla n\rangle_g+4\pi\,n^2\,\lvert \nabla\phi\rvert_g^2\,,\\
{\Riem[\check{g}]^3}_{i3j}=&\,-\sqrt{4\pi}\,\nabla_i\nabla_j\phi-12\pi\, n\,\psi\,\nabla_i\phi-\sqrt{4\pi}\,\psi\,\left(\hat{k}_{ij}+\left(\frac{\tau}2+\sqrt{4\pi}\psi\right)\,g_{ij}\right)+4\pi\,\lvert\nabla\phi\rvert_g^2\,g_{ij}\,,\\
{\Riem[\check{g}]^3}_{i3t}=&\,-\sqrt{4\pi}\,\nabla_i(n\,\psi)+\sqrt{4\pi}\psi\,\nabla_in-12\pi\,n\,\psi\,\nabla_i\phi-\sqrt{4\pi}\,n\,\hat{k}^r_{\ i}\,\nabla_r\phi-\sqrt{4\pi}\,n\,\tau\,\nabla_i\phi\,,\\
{\Riem[\check{g}]^3}_{\mu\nu\rho}=&\,
0\,.
\end{align*}

\section*{Declarations}

\subsection*{Competing Interests} This research was funded in part by the Austrian Science Fund (FWF) 10.55776/Y963 \changereport{and 10.55776/PAT7614324}. 
The author \changereport{was }a recipient of a DOC Fellowship of the Austrian Academy of Sciences at the Faculty of Mathematics at the University of Vienna, \change{and received a scholarship by the German Academic Scholarship Foundation (Studienstiftung des deutschen Volkes) during most of his work on this paper.} 

\subsection*{Acknowledgements} The author thanks David Fajman for suggesting to tackle this problem and Michael Eichmair, David Fajman, Hans Ringström \change{and the anonymous referee }for their constructive comments that helped improve the manuscript. Additionally, he thanks Mihalis Dafermos for his warm hospitality during the author's stay at Princeton University and the Erwin Schrödinger International Institute for Mathematics and Physics in Vienna for hosting the workshop \enquote{Nonlinear Waves and General Relativity}, during both of which parts of this work were written.

\bibliographystyle{alpha}
\bibliography{bibliography}

\end{document}